\documentclass{article}[12pt]
\usepackage{authblk}
\usepackage{tabto} 
\usepackage{relsize} 
\usepackage{graphicx} 
\usepackage{tikz} 
\usepackage{float} 
\usetikzlibrary{calc,arrows}
\usepackage{amsmath,amssymb,amsthm} 
\theoremstyle{Theorem}
\newtheorem{theorem}{Theorem}
\newtheorem{lemma}{Lemma}
\newtheorem{remark}{Remark}

\newcommand{\tarc}{\mbox{\large$\frown$}}
\newcommand{\arc}[2][-3ex]{{#2}{\kern #1{\raisebox{1.5ex}{\tarc}}}}
\usepackage{amsfonts} 
\usepackage[export]{adjustbox} 
\usepackage{algpseudocode}
\usepackage{algorithm}
\usetikzlibrary{calc,arrows}
\usepackage{comment}
\usepackage{wrapfig}
\usepackage{hyperref} 
\hypersetup{
    colorlinks,
    linkcolor={red!50!black},
    citecolor={green!50!black},
    urlcolor={blue!80!black}
}
\usepackage{siunitx} 
\usepackage{appendix}
\usepackage{multirow} 
\usepackage[
backend=biber,
style=alphabetic,
sorting=ynt
]{biblatex}
\usepackage{geometry}
\usepackage{array}
\usepackage{longtable}

\usepackage{todonotes}

\def\Xint#1{\mathchoice
   {\XXint\displaystyle\textstyle{#1}}%
   {\XXint\textstyle\scriptstyle{#1}}%
   {\XXint\scriptstyle\scriptscriptstyle{#1}}%
   {\XXint\scriptscriptstyle\scriptscriptstyle{#1}}%
   \!\int}
\def\XXint#1#2#3{{\setbox0=\hbox{$#1{#2#3}{\int}$}
     \vcenter{\hbox{$#2#3$}}\kern-.5\wd0}}

\def\dashint{\Xint-}

\title{Extensible membranes in inviscid flow: aerodynamics and singular limits}
\author[1]{Yu Jun Loo \thanks{looyujun@umich.edu}}
\author[1]{Silas Alben \thanks{alben@umich.edu}}
\affil[1]{Department of Mathematics, University of Michigan, Ann Arbor, MI 48109, USA}
\date{}

\begin{document}

\maketitle
\begin{abstract}
We develop a two-dimensional inviscid model for extensible membranes with continuous vortex shedding from both edges, coupled to a weak spectral Galerkin discretization of the membrane equation that permits direct computation at zero bending rigidity. Viewing the pressure jump as a flux for bound vorticity explains how membranes respond to concentrated fluid loading. For fixed-fixed membranes, compliance draws the body toward the leading-edge vortex, delaying its detachment and enhancing lift primarily during start-up; after periodic shedding develops, compliance mainly amplifies the force oscillations rather than mean lift. Vortex detachment can unload heavy membranes into compression, producing bending-regularized wrinkles. For $R_2>0$, the membrane equation remains well-posed independently of the sign of the tension, provided the membrane remains an immersion. Releasing the trailing edge creates an intrinsic degeneracy: the tension vanishes at the free endpoint, making the $R_2=0$ problem formally underdetermined, while the weak formulation selects its bounded-energy branch. This vanishing-tension endpoint traps pressure disturbances and drives rapid tip-snapping; the dominant frequency and effective wave number approach their zero-rigidity limits with an inverse-logarithmic correction arising from a singular Bessel branch. Finally, decreasing the stretching modulus destabilizes periodic flutter, and the repeated fluid-driven departure and elastic return suggest a homoclinic tangle underlying the transition from periodic to chaotic flutter.

\end{abstract}

\section{Introduction}
The vortex wake behind an elastic body in a high-Reynolds-number flow is shaped by the body's compliance. As the body bends and stretches, it transfers momentum to the fluid through shed vortices whose strength and topology depend sensitively on its elastic properties. This coupling has attracted interest because compliance can enhance lift \cite{gordnier2009high, gehrke2025highly} and reduce drag \cite{alben2002drag}. For filaments and plates with moderate-to-large bending rigidity ($R_2$), the coupled body--wake problem has been studied extensively in simulation \cite{shelley2011flapping, alben2009simulating, alben2022dynamics, connell2007flapping, argentina2005fluid, huang2010three, tian2013role, hwa2014flapping} and experiment \cite{jia2018experimental, eloy2008aeroelastic, eloy2012origin}. The present work concerns the softer limit: extensible membranes with small bending rigidity, which can stretch substantially and exhibit fine wrinkles under fluid loading \cite{attar2012experimental, tzezana2019thrust, mavroyiakoumou2020large, mavroyiakoumou2025sail}. A membrane pinned at both ends can model a sail, while one pinned at one end and free at the other models a flag. Both configurations arise in engineering and biology, e.g. textile sails, the compliant wings of bats, flexible flow sensors, and energy harvesters \cite{konow2015spring, wang2016stability, yu2019review,yeh2025influence}, and both configurations are studied here.
\newline

In this softer regime, the membrane deformation actively reorganizes the separated flow. In the hovering experiments of \cite{gehrke2025highly, gehrke2022aeroelastic}, for example, an elastic membrane cambers into a C-shape and stretches toward the core of the leading-edge vortex. This delays vortex detachment, concentrates circulation near the leading edge, and substantially amplifies lift. Capturing this and related mechanisms requires continuous, stable leading-edge vortex shedding, yet previous vortex sheet treatments of both inextensible and extensible bodies have suppressed leading-edge shedding for numerical stability \cite{alben2009simulating, mavroyiakoumou2020large, mavroyiakoumou2025sail}.
\newline

In this work, we develop an inviscid model of an extensible membrane that resolves continuous vortex shedding from both the leading and trailing edges. The fluid model extends the continuous leading-edge shedding framework for thin rigid bodies developed in \cite{loo2025falling, loo2026comparisoninviscidviscousvortex} to deformable bodies. In the rigid case, this framework produces average relative force agreement to within \(10\)--\(20\%\) of direct Navier--Stokes simulations across roughly seventy distinct body motions \cite{loo2026comparisoninviscidviscousvortex}. The membrane equation is discretized using a spectral Galerkin approximation of its weak form (section \ref{membrane section} and appendix \ref{solving the membrane equation appendix}). This formulation naturally incorporates the variational boundary conditions and allows the fixed-free problem to be computed directly in the singular limit of zero bending rigidity. In that limit the strong system is formally underdetermined because the vanishing tension at the free edge removes one independent boundary condition. Previous finite-difference and collocation treatments have resolved this either by introducing an additional boundary condition \cite{mavroyiakoumou2020large} or by retaining a small positive bending rigidity \cite{connell2007flapping}. In the weak formulation used here, the finite-energy constraint excludes the singular branch, therefore selecting the bounded-energy solution.
\newline

We interpret the resulting dynamics through a physical picture afforded by the inviscid formulation. As shown in section~\ref{pressure jump equation as a conservation law section}, the pressure-jump equation is obtained by projecting the Euler momentum conservation law onto the membrane and taking its jump across the surface. The pressure jump then appears as a flux for bound vorticity. In the same way that pressure gradients transport fluid momentum, pressure-jump gradients transport bound vortex sheet strength along the membrane. Vorticity therefore tends to accumulate near extrema of the pressure jump before being shed from the membrane edges into the wake. This picture makes the coupling between membrane geometry, surface pressure, and vortex shedding transparent and provides a common explanation for the principal phenomena observed in the simulations. Together with Kelvin's theorem, this conservation-law viewpoint also explains the alternating signs of the shed vortices: an attached vortex requires compensating bound circulation of the opposite sign on the body, which is transported to an edge by the conservation law and shed as the next vortex.
\newline

For fixed-fixed membranes (section \ref{fixed fixed membranes section}), the low-pressure core of the leading-edge vortex draws the membrane toward it at moderate angles of attack. The resulting camber keeps the vortex close to the surface, delays its detachment, and increases lift and aerodynamic efficiency during start-up. Once alternating shedding develops, however, compliance amplifies both the force peaks and troughs and does not produce a universal increase in mean lift. At normal incidence, where symmetry precludes lift, the same attraction toward the symmetric vortex pair instead increases drag. Vortex detachment can also unload sufficiently heavy membranes so rapidly that their tension becomes locally negative and fine wrinkles develop. The zero-rigidity equation is then ill-posed because arbitrarily short waves grow arbitrarily rapidly \cite{triantafyllou1994dynamic}, whereas positive bending rigidity selects a finite wrinkling scale. In appendix \ref{well-posedness of the membrane equation appendix}, we prove that under sufficiently regular prescribed loading, the membrane equation with $R_2>0$ is locally well-posed regardless of the sign of the tension, provided that the membrane remains an immersion.
\newline

For fixed-free membranes (section \ref{fixed free membranes section}), decreasing the bending rigidity produces rapid snapping at the free trailing edge. When $R_2=0$, the tension and hence the membrane wave speed vanish at this edge, while the Kutta condition forces the pressure jump to vanish at the endpoint itself. Pressure disturbances consequently accumulate immediately upstream of the edge, forming an alternating pressure extremum that drives the tip back and forth and produces small hooks. Positive bending rigidity smooths these hooks, but its influence persists to remarkably small values of \(R_2\). We show numerically that the dominant tip frequency approaches its zero-rigidity limit with an approximately inverse-logarithmic correction before saturating at a constant value. That is, 
$$
    f(R_2)-f_0,
    \;
    k_{\mathrm{eff}}(R_2)-k_{\mathrm{eff},0}
    =
    O\left(\frac{1}{|\log R_2|}\right).
$$

Linear analysis traces this slow convergence to the logarithmically-singular \(Y_0\) branch of the zero-rigidity problem, whose coefficient must vanish inverse-logarithmically as the bending layer contracts toward the free edge. More generally, an energy estimate for uniformly-tensioned membranes with Dirichlet boundary data gives convergence of the positive-rigidity solutions to the \(R_2=0\) solution at a rate of at worst \(O(R_2^{1/2})\), with an improved \(O(R_2)\) rate under additional regularity and boundary assumptions. Although the uniform-tension hypothesis fails at the free endpoint, the numerical results show the same slow approach to the zero-rigidity flag dynamics, consistent with the linearized analysis of \cite{manela2017hanging}.
\newline

Finally, in section \ref{effect of R3 section}, we examine how the stretching modulus ($R_3$) alters the periodic fixed-free flag dynamics observed in \cite{alben2008flapping} at $R_2=0.25$. For sufficiently large $R_3$, stretching remains negligible and the instability of the undeformed state saturates into a stable periodic flapping orbit. When $R_3$ decreases below $\mathcal{O}(1)$, appreciable stretching destabilizes this orbit and the motion becomes chaotic. During each excursion, the flow drives the flag away from its undeformed state before bending and stretching return it toward that state, where the flutter instability initiates the next excursion. This repeated departure and return is consistent with a global return from the unstable manifold of the undeformed state toward its stable manifold. A transverse intersection between these manifolds would produce a homoclinic tangle and provide a dynamical-systems mechanism for the observed transition from periodic to chaotic flutter.
\newline

The remainder of the paper is organized as follows. Section~\ref{model section} develops the coupled membrane--fluid model, including the membrane equations in section~\ref{membrane section} and the inviscid fluid equations in section~\ref{fluid equations section}. Section~\ref{pressure jump equation as a conservation law section} then derives the pressure-jump equation directly from the Euler momentum conservation law. Section~\ref{numerical results section} presents the numerical results for fixed-fixed (section~\ref{fixed fixed membranes section}) and fixed-free (section~\ref{fixed free membranes section}) membranes, and section~\ref{conclusion section} summarizes the main findings. Complete derivations and supplementary analysis of the fluid and elastic models, numerical method, and results are collected in appendices A--F, which are arranged as a self-contained account rather than in order of first reference from the main text.
\newline

\section{Model}
\label{model section}
The inviscid model we employ builds on prior work by the present authors \cite{loo2025falling, loo2026comparisoninviscidviscousvortex} while incorporating several key improvements that extend it from rigid plates to elastic membranes. The coupled membrane model that we employ is similar to that described in \cite{mavroyiakoumou2020large,mavroyiakoumou2025sail}. The critical distinction lies in our numerical schemes for both the ambient inviscid fluid and the elastic membrane. The inviscid model is based on recent advances that enable continuous, stable leading-edge vortex shedding \cite{loo2025falling, loo2026comparisoninviscidviscousvortex}. The numerical scheme for the membrane model is based on a spectral Galerkin approximation of the weak form of the equations.

\subsection{Membrane equation}
\label{membrane section}
A membrane is a thin elastic sheet with small bending rigidity. In this section we will derive the weak form of the membrane equation from a variational principle. We derive the constitutive relations for the membrane by assuming that the body is a pencil of parallel Hookean fibers, and linearizing the bending energy about a small stretch assumption. A direct derivation (without bending rigidity) can be found in \cite{mavroyiakoumou2020large}.
\newline

Consider the motion of an extensible membrane that is fixed at two endpoints and is parameterized by the complex-valued function $$\zeta(\alpha, t) = x(\alpha, t) + i y(\alpha, t).$$
Here $-L\leq\alpha \leq L$ is the material coordinate of the membrane, and represents its signed arc length coordinate when unstretched, and $i = \sqrt{-1}$ is the unit imaginary number. By abuse of notation, we let $\zeta$ also refer to the vector $[x,y]^T$.
\newline

First, we fix our notation. Let $s$ denote the arc length parameter of the deflected membrane, so that its stretch factor is given by $\lambda = \partial_\alpha s$. The symbols $h$, $W$ and $\rho_s$ denote the membrane thickness, in-plane width, and mass density per unit volume respectively. Let $E$ denote the Young's modulus of the membrane and $B=EI$ its flexural rigidity. Let $\hat{\boldsymbol{s}} = \partial_{\alpha}\zeta/\partial_{\alpha}s=\lambda^{-1}\partial_\alpha\zeta$ and $\hat{\boldsymbol{n}} = \hat{\boldsymbol{s}}^\perp = i\hat{\boldsymbol{s}}$ represent the unit tangent and normal vectors of the membrane respectively. We refer to the side of the body that $\pm\hat{\boldsymbol{n}}$ points to as the $\pm$ side of the membrane. Let $\kappa_0 = \partial_\alpha\theta$ denote the material curvature, where $\theta$ is the unit tangent angle on the membrane (i.e.~$\exp(i\theta) = \hat{\boldsymbol{s}}$) \cite{gerstmayr2008correct}. Note that the geometric curvature of the membrane $\kappa$ is related to the material curvature by
$$\kappa = \partial_s\theta = (\partial_\alpha s)^{-1} \partial_\alpha\theta = \lambda^{-1}\kappa_0.$$ 

In figure \ref{model figure} we illustrate these quantities, showing the membrane ($\zeta$) in dark gray. The shed vorticity $\omega$ is shown on a blue-red scale, and is parameterized by $\zeta_{\pm}$ which we describe carefully in the subsequent subsection \ref{fluid equations section}.

\begin{figure}[H]
    \centering
    \includegraphics[width=1\linewidth, trim = 1.5cm 1cm 0.5cm 2cm, clip]{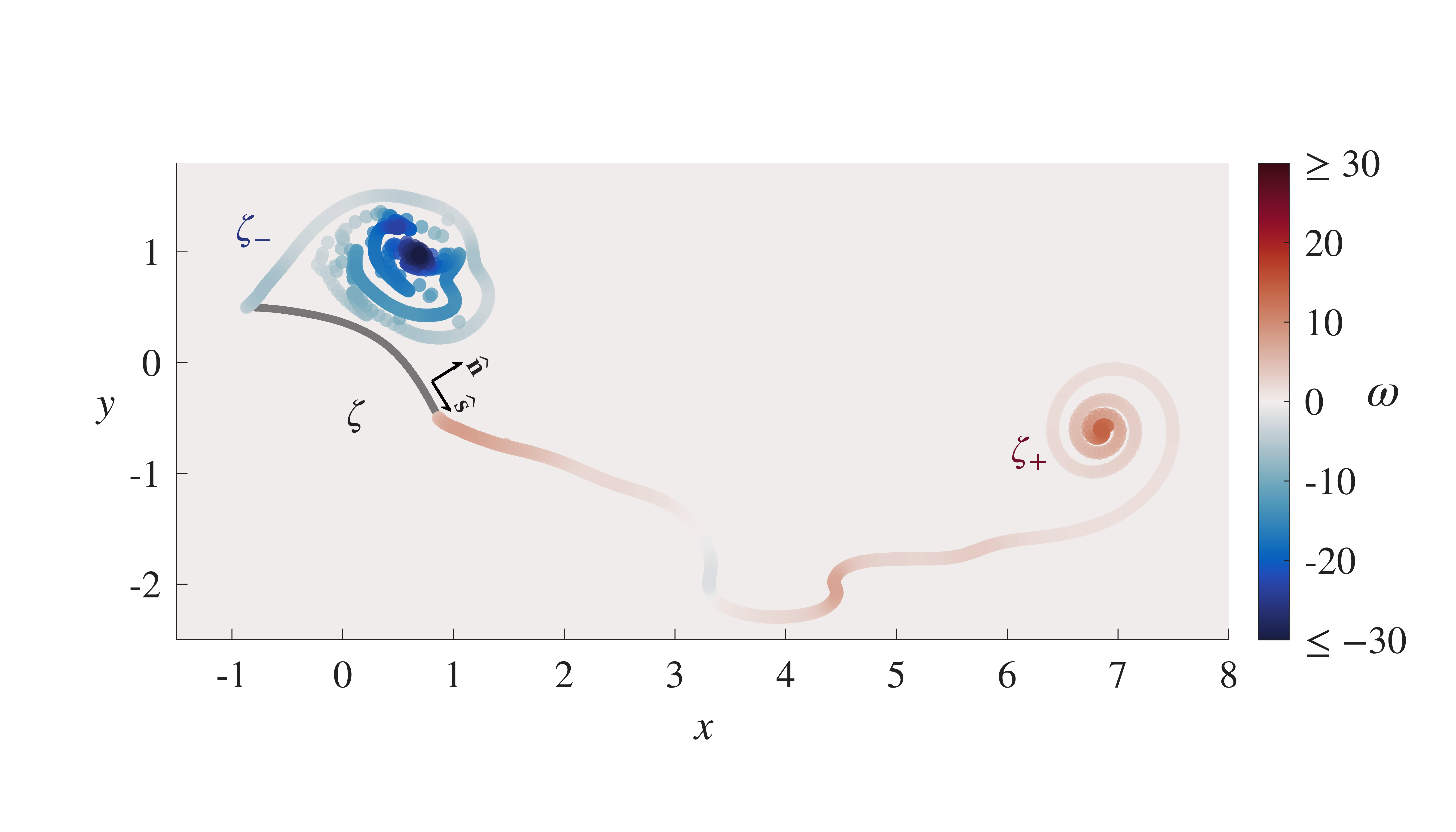}
    \caption{Diagram showing the membrane $\zeta$ in gray in a uniform flow. The shed vorticity $\omega$ in the ambient fluid is shown in blue-red and is parameterized by $\zeta_-,\zeta_+$. The local tangent frame spanned by $\hat{\boldsymbol{s}},\hat{\boldsymbol{n}}$ is shown in black.  The vorticity $\omega$ is obtained using equation \eqref{vorticity formula}.}
    \label{model figure}
\end{figure}
In this notation, the Lagrangian $\mathcal{L}$ is given by 

\begin{align}
    \mathcal{L} &= \int_{-L}^L\frac{1}{2}\rho_shW|\partial_t\zeta|^2d\alpha - \bigg(\int_{-L}^{L} \frac{1}{2}EhW(\lambda - 1)^2d\alpha + \int_{-L}^{L} \frac{1}{2}B|\kappa_0|^2d\alpha \bigg) \label{lagrangian equation} \\
    &=: \mathcal{K} - \bigg(\mathcal{U}_{stretch} + \mathcal{U}_{bend}\bigg) \nonumber
\end{align}

Here $\mathcal{K}$ is the kinetic energy and the terms $\mathcal{U}_{stretch}, \mathcal{U}_{bend}$ represent the stored stretching and bending energies of the membrane. The forms that these terms take depend on the constitutive relations that the membrane respects. The constitutive relations that we use to determine the bending and stretching energies of the body in \eqref{lagrangian equation} are consistent with an elastic body that is formed by stacking Hookean fibers with the same elastic coefficient $E$ into a sheet. In appendix \ref{bending and stretching energy derivation appendix} we show how the constitutive relations can be derived from that assumption.
\newline

Let $U$ denote the velocity of the background flow, which we use as our characteristic velocity. 
Similarly, let $L$ (half the chord length) denote the characteristic length scale. Hence, if we let $\tilde{(\cdot)}$ denote the dimensionless variable corresponding to $(\cdot)$ so that  $$t = \tilde{t}\frac{L}{U},\ \alpha = \tilde{\alpha}L,\ \zeta = \tilde{\zeta}L,\ \kappa_0 = \tilde{\kappa}_0\frac{1}{L},$$

then the Lagrangian may be written in dimensionless form as 

\begin{equation}
\tilde{\mathcal{L}} = \int_{-1}^1 \frac{1}{2}R_1 |\partial_{\tilde{t}} \tilde{\zeta}|^2d\tilde{\alpha} - \bigg(\int_{-1}^1\frac{1}{2}R_3(\lambda-1)^2d\tilde{\alpha} + \int_{-1}^{1}\frac{1}{2}R_2|\tilde{\kappa}_0|^2d\tilde{\alpha} \bigg)
\end{equation}
where we have $$\tilde{\mathcal{L}} = \frac{\mathcal{L}}{\rho_fU^2L^2W},\  R_1 = \frac{\rho_sh}{\rho_fL},\ R_2 = \frac{B}{\rho_fU^2WL^3},\ R_3 = \frac{Eh}{\rho_fU^2L},$$

Henceforth we suppress all $\tilde{(\cdot)}$ to obtain our final dimensionless Lagrangian 

\begin{align}
\mathcal{L}&= \int_{-1}^1 \frac{1}{2}R_1 |\partial_{t} \zeta|^2d\alpha - \bigg(\int_{-1}^1\frac{1}{2}R_3(\lambda-1)^2d\alpha + \int_{-1}^{1}\frac{1}{2}R_2|\kappa_0|^2d\alpha \bigg)
\\
    &=: \mathcal{K} - \bigg(\mathcal{U}_{stretch} + \mathcal{U}_{bend} \bigg).
\end{align}

Let $\boldsymbol{f}$ denote the dimensionless force density exerted by the fluid on the membrane. At any time instant, its virtual work under an admissible displacement $\delta\zeta$ is
\begin{equation}
\delta\mathcal{W}_{\mathrm{ext}}
= \int_{-1}^1\boldsymbol{f}\cdot\delta\zeta,d\alpha .
\end{equation}
The extended Hamilton's principle \cite[pp.~268--272]{meirovitch2001fundamentals} therefore gives
\begin{equation}
\delta\int_0^{\tau_f}
\mathcal{L}\,dt
+
\int_0^{\tau_f}\delta\mathcal{W}_{ext}\,dt = 
\int_0^{\tau_f}
\delta\left(
\mathcal{K}
-\mathcal{U}_{stretch}
-\mathcal{U}_{bend}
+ \mathcal{W}_{ext}
\right)\,dt
=:\int_0^{\tau_f}\delta\mathcal{S}
=0,
\end{equation}
where $\tau_f$ is the final time and $\mathcal{S}$ is the action. Hence
\begin{equation}
\int_0^{\tau_f}\delta\mathcal{S}\,dt =\int_0^{\tau_f}\int_{-1}^1
\bigg(
-R_1\partial_{tt}\zeta
+R_3\partial_\alpha\big((\lambda-1)\hat{\boldsymbol{s}}\big)
-R_2\partial_\alpha(\partial_s\kappa_0\hat{\boldsymbol{n}})
+\boldsymbol{f}
\bigg)\cdot\delta\zeta\,d\alpha\,dt
=0.
\label{weak form membrane equation}
\end{equation}

In appendix \ref{deriving membrane equation appendix} we give a complete derivation of \eqref{weak form membrane equation}. Since this integral equation must hold for all virtual displacements $\delta\zeta$, we have 

\begin{equation}
     R_1\partial_{tt}\zeta  = R_3\partial_\alpha\big( (\lambda-1)\hat{\boldsymbol{s}}\big)- R_2\partial_\alpha( \partial_s\kappa_0\hat{\boldsymbol{n}}) \\ +\boldsymbol{f}
     \label{dimensionless membrane equation with arbitrary force}
\end{equation}

The bending rigidity term  $- R_2\partial_\alpha( \partial_s\kappa_0\hat{\boldsymbol{n}})$ in equation \eqref{dimensionless membrane equation with arbitrary force} often appears in a linearized form $-R_2\partial_\alpha^4\zeta$ in the literature \cite{zhu2002simulation, zhang2026flapping, connell2007flapping}. 
Both terms are equivalent upon applying a ``small stretching" assumption to the bending energy. Indeed, as discussed above the membrane potential energy incurred from bending is given by 
 \begin{align}
     \int_{-1}^1\frac{1}{2}R_2(\kappa_0)^2d\alpha = \int_{-1}^1\frac{1}{2}R_2(\partial_\alpha\theta)^2d\alpha 
     = \int_{-1}^1\frac{1}{2}R_2(\lambda\kappa)^2d\alpha 
     = \int_{-1}^1\frac{1}{2}R_2|\lambda\partial_s\partial_s\zeta|^2d\alpha &= \int_{-1}^1\frac{1}{2}R_2|\partial_\alpha(\lambda^{-1}\partial_\alpha\zeta)|^2d\alpha \notag \\
     &\approx\int_{-1}^1\frac{1}{2}R_2|\partial_\alpha^2\zeta|^2d\alpha. 
 \end{align}
 
The last approximation follows from the assumption that $\lambda =\partial_\alpha s \approx 1$ (i.e.~membrane stretching is small).  With this small-stretch approximation of the energy,
we may recompute the resulting force by taking its first variation:
\begin{equation}
    \delta\mathcal{U}_{bend} =\int_0^{\tau_f}\int_{-1}^1\frac{1}{2}R_2\delta|\partial_\alpha^2\zeta|^2d\alpha\,dt =\int_0^{\tau_f}\int_{-1}^1R_2\partial_\alpha^2\zeta\cdot \partial_{\alpha}^2(\delta\zeta)d\alpha\,dt 
    =\int_0^{\tau_f}\int_{-1}^1R_2\partial_\alpha^4\zeta\cdot (\delta\zeta)d\alpha\,dt 
\end{equation}
Hence substituting this into equation \eqref{weak form membrane equation} we see that the new bending rigidity term is $-R_2\partial_\alpha^4\zeta$ under this assumption.  This bending force results in the following membrane equation, which we use in this study:

\begin{equation}
     R_1\partial_{tt}\zeta  = R_3\partial_\alpha\big( (\lambda-1)\hat{\boldsymbol{s}}\big)- R_2\partial_{\alpha}^4\zeta \\ +\boldsymbol{f}
     \label{bending linearized membrane equation with arbitrary force}
\end{equation}

This bending force also has a physical motivation. Indeed, the constitutive law for the bending energy that we consider in this study $$\int_{-1}^1\frac{1}{2}R_2|\partial_\alpha^2\zeta|^2d\alpha$$ can be obtained as the continuum limit of a discrete extensible worm-like chain model \cite{PhysRevE.92.032603,barkema2012model,barkema2017efficient}, and the resulting continuous model has been considered in \cite{obermayer2009tension}. The discrete worm-like chain model represents a polymer chain with beads connected by springs; experimental studies support this worm-like chain description of individual biological polymers, including for example the stretching of tropoelastin molecules \cite{baldock2011shape} and the bending of collagen molecules \cite{rezaei2018environmentally}. Intuitively, this suggests that constitutive laws for gels of such polymers can likewise be constructed by taking the worm-like chain model to its continuum limit. Indeed, constitutive laws based on continuous worm-like chains reproduce the measured strain-stiffening of collagen gels at modest strains \cite{storm2005nonlinear} and of fibrin gels over a range that includes stretching of the constituent fibers \cite{piechocka2016multi}. Bat wing membranes provide a biological example in which collagen and elastin fibers contribute directly to aerodynamic function: elastin fibers promote wrinkling and extensibility, while the surrounding collagen-containing tissue bears increasing loads as the membrane becomes taut \cite{cheney2015wrinkle}. These connections motivate the use of equation \eqref{bending linearized membrane equation with arbitrary force} to investigate how the elasticity of extensible fibrous membranes influences their deformation under aerodynamic loading.
\newline

In this paper, we focus on the membrane limit in which $R_2$ is either zero or nonzero but small. Some of the aforementioned studies have fixed the bending rigidity at a small nonzero value, e.g. to study a very flexible but nearly inextensible flag with fixed-free boundary conditions \cite{connell2007flapping} or to study membranes fixed at both ends that undergo significant stretching \cite{sun2022dynamics, zhang2026flapping}. A small but nonzero bending rigidity has been found to improve the robustness of numerical simulations \cite{connell2007flapping,zhang2026flapping}. We find a similar effect in our simulations, i.e. a small nonzero $R_2$ stabilizes the membrane dynamics by inhibiting the growth of spurious small-scale features while having little effect on the large-scale features. In appendix \ref{convergence in the limit of zero bending rigidity section} we clarify the effect of small but nonzero $R_2$. We show that for the boundary conditions considered in this study, solutions to equation \eqref{bending linearized membrane equation with arbitrary force} with $R_2 >0$ converge to the solution with $R_2 = 0$ in both position and velocity with respect to the energy norm, at a rate of at worst $\mathcal{O}(R_2^{1/2})$.
\newline

In the present study, the force density $\boldsymbol{f}$ is applied by an ambient inviscid fluid. In the absence of viscosity, only pressure forces remain. Let $p^\pm$ denote the limiting values of the pressure on the $\pm$ side of the membrane. The pressure force applied to any infinitesimal membrane segment of length $\Delta s$ by the fluid is thus  $ p^-\hat{\boldsymbol{n}}\Delta s + p^+(-\hat{\boldsymbol{n}})\Delta s=: [p]^-_+\hat{\boldsymbol{n}}\Delta s$. Hence the (material) force density is given by $$\boldsymbol{f} = [p]^-_+\hat{\boldsymbol{n}}\frac{\Delta s}{\Delta\alpha} =  [p]^-_+\hat{\boldsymbol{n}}\partial_\alpha s.$$ 

 The dimensionless membrane equation is hence
 \begin{equation}
     R_1\partial_{tt}\zeta  = R_3\partial_\alpha\big( (\lambda-1)\hat{\boldsymbol{s}}\big)- R_2\partial_\alpha^4\zeta + [p]^-_+\hat{\boldsymbol{n}}\partial_\alpha s.
     \label{dimensionless membrane equation}
\end{equation}

When $R_2 = 0$, this is exactly the equation in \cite{mavroyiakoumou2020large} in the absence of pretension. When $\partial_\alpha s = 1$, the membrane is inextensible and this equation agrees with that in \cite{alben2009simulating}. Similar to \cite{sun2022dynamics, zhang2026flapping}, we solve the weak formulation of \eqref{dimensionless membrane equation}, given by 

\begin{equation}
    \delta\mathcal{S} =\int_0^{\tau_f}\int_{-1}^1\bigg(-R_1\partial_{tt}\zeta +R_3\partial_\alpha\big( (\lambda-1)\hat{\boldsymbol{s}}\big) - R_2\partial_\alpha^4\zeta +\boldsymbol{f} \bigg)\cdot\delta\zeta\,d\alpha\, dt = 0
    \label{linearized weak form membrane equation}
\end{equation}
This integral equation is discretized on the Chebyshev-Lobatto grid using Clenshaw-Curtis weights via the spectral Galerkin method with numerical integration (G-NI) \cite{canuto2006spectral}. Solving the weak formulation has the advantage that the discrete energy of the membrane can only change due to work done by the fluid on the membrane through the pressure term. Moreover, the eigenvalues of the discretized operators have the same signs as the eigenvalues of the actual operators in the continuous equations they represent. We find that this improves the stability of the numerical model. In appendix \ref{solving the membrane equation appendix} we describe how this is done.
\subsubsection{The effect of $R_2 > 0$}
\label{R2 > 0 section}

In the physical world, membranes have small but nonzero thicknesses and therefore nonzero bending rigidities ($R_2 > 0$). Such membranes sometimes lose tension locally, and so may buckle, wrinkle, and crumple. This is not surprising to anyone who has tried to lay out a bed sheet neatly.
\newline

When $R_2 = 0$, equation \eqref{bending linearized membrane equation with arbitrary force} cannot model such phenomena. Let $\boldsymbol{f} = \boldsymbol{0}$. It is easy to show that in the absence of bending rigidity ($R_2 = 0$), equation \eqref{bending linearized membrane equation with arbitrary force} becomes ill-posed whenever $T = R_3(\lambda - 1) < 0$ on any open material interval $\alpha_1 < \alpha < \alpha_2$ \cite{triantafyllou1994dynamic}. To illustrate the mechanism, consider a horizontal, uniformly compressed base state with constant stretch $\lambda_0<1$, and linearize the transverse dynamics about this state. Since a transverse perturbation changes $\lambda$ only at second order, the base tension is constant to first order. The resulting linearized transverse equation is
\begin{equation}
R_1\partial_{tt}y = -R_3|e|\partial_{\alpha\alpha}y,
\label{linearized membrane equation with negative tension}
\end{equation}
where
\[
e=\frac{\lambda_0-1}{\lambda_0}<0.
\]
Consider the bounded sinusoidal perturbation $y = \sin({ik\alpha + \sigma t})$ with wave number (spatial frequency) $k$ and growth exponent $\sigma$. Substituting this ansatz into equation \eqref{linearized membrane equation with negative tension}, we see that the dispersion relation that $(\sigma,k)$ must satisfy is
\begin{equation}
R_1\sigma^2 = R_3|e|k^2
\implies
|\sigma| = \sqrt{\frac{R_3|e|}{R_1}}\,|k|.
\end{equation}

Hence every sinusoidal perturbation with wave number $k$ has an exponentially growing branch with growth rate
\[
\sigma = \sqrt{\frac{R_3|e|}{R_1}}\,|k|.
\]
In particular, the growth rate is unbounded as $|k|\rightarrow\infty$, so arbitrarily short-wavelength perturbations grow on arbitrarily short time scales. As the growth rate $\sigma$ is unbounded, the solution operator is thus discontinuous, and the idealized membrane with $R_2 = 0$ is therefore ill-posed with negative tension. In plain language, this states that equation \eqref{dimensionless membrane equation with arbitrary force} cannot model membrane compression. Intuitively this makes sense. Suppose a membrane with zero bending rigidity is under compression. At any point on the membrane, without bending rigidity, the compressed membrane cannot resist buckling as it can reach a lower energy state by bending sharply there. The membrane thus buckles at every point, wrinkling with arbitrarily large spatial frequency. Since this is impossible, the problem is ill-posed. In the language of PDEs, this can be explained concisely as follows: when $R_2 = 0$ and $T<0$ the equation transforms from hyperbolic to elliptic in time and hence becomes ill-posed as an initial value problem. When $R_2 > 0$, the PDE remains hyperbolic, independent of the sign of $T$, and so is well-posed, independent of the sign of $T$.
\newline

Membranes in the real world, however, do have bending rigidity, and often crumple and wrinkle in ways that equation \eqref{dimensionless membrane equation with arbitrary force} with $R_2 = 0$ cannot model (see section \ref{compression example section} for one such example). On the other hand, for any $R_2 >0$ the linearized problem becomes
\begin{equation}
R_1\partial_{tt}y = -R_3|e|\partial_{\alpha\alpha}y - R_2\partial_{\alpha\alpha\alpha\alpha}y
\label{linearized membrane equation with negative tension and R2}
\end{equation}
so that the dispersion relation is given by
\begin{equation}
R_1\sigma^2 = R_3|e|k^2 - R_2k^4.
\label{dispersion relation equation}
\end{equation}
The unstable modes therefore satisfy

$$
    |k| < \sqrt{\frac{R_3|e|}{R_2}},
$$

and their growth exponent is bounded by

$$
    0\leq \sigma \leq \frac{R_3|e|}{2\sqrt{R_1R_2}}.
$$

Hence, the equation is unstable only to sinusoidal perturbations with a limited range of wave numbers, and their growth rates are bounded by $\sigma \leq \frac{R_3|e|}{2\sqrt{R_1R_2}}$.
Hence, its linearized solution operator remains continuous, and the problem does not necessarily become ill-posed when the tension becomes negative. In appendix \ref{well-posedness of the membrane equation appendix} we confirm this, and provide a rigorous proof that the membrane equation with $R_2 > 0$ remains well-posed regardless of the sign of the tension.
\newline

In summary, when $R_2 = 0$ the membrane equation \eqref{dimensionless membrane equation with arbitrary force} becomes ill-posed in the presence of any compression. On the other hand, when $R_2 >0$, however small, the same equation \eqref{dimensionless membrane equation with arbitrary force} remains well-posed. Moreover, as discussed in the previous section, when the tension is uniformly positive, the solutions with $R_2 > 0$ converge to the solution with $R_2 = 0$ uniformly on any closed time interval.

\subsubsection{Boundary conditions}
In this study we consider two sets of boundary conditions: fixed-fixed and fixed-free.

\paragraph{Fixed-fixed boundary conditions:}

For fixed-fixed boundary conditions we prescribe the motion of both of the endpoints $\zeta(\pm1,t)$, as follows:

\begin{align}
    \zeta|_{\alpha=\pm1} &= \pm e^{i\beta} \\
\partial_\alpha^2\zeta\big|_{\alpha =\pm1} &= 0. 
\end{align}
Here, $\beta$ is the orientation angle of the chord connecting the two membrane endpoints. Physically, $\partial_\alpha^2\zeta\big|_{\alpha =\pm1} = 0$ means zero resultant moments at the endpoints, or just zero material curvature vector, so the ends can rotate freely. Our fixed-fixed boundary conditions are also called pinned boundary conditions \cite{timoshenko1959theory}.

\paragraph{Fixed-free boundary conditions:}
For the fixed-free boundary conditions, we apply the pinned condition at the leading edge only:

\begin{align}
    \zeta|_{\alpha=-1}= -e^{i\beta}, \ &\ R_2\partial_\alpha^2\zeta\big|_{\alpha =-1} = 0 \\ &\text{and} \nonumber \\
    \big(R_3(1-\lambda^{-1})\partial_\alpha\zeta +R_2\partial_\alpha^3\zeta\big) \big|_{\alpha = 1}= 0,\ &\ R_2\partial_\alpha^2\zeta\big|_{\alpha =1}\ \  = 0.
\end{align}

The position of the leading edge of the membrane is prescribed. The trailing edge is allowed to move freely. Physically, $\big(R_3(1-\lambda^{-1})\partial_\alpha\zeta +R_2\partial_\alpha^3\zeta\big) \big|_{\alpha = 1}= 0$ means zero resultant force at the free end.
\newline

A complete derivation of the boundary conditions from the variational formulation including their physical interpretation is included in appendix \ref{solving the membrane equation appendix}.

\subsection{Fluid equations}
\label{fluid equations section}
Having discussed our model for the membrane, we now describe our inviscid fluid model. A complete description of the model can be found in \cite{loo2025falling} in the context of falling rigid bodies including a convergence study that shows that it is approximately first-order in time and second-order in space. In \cite{loo2026comparisoninviscidviscousvortex} we discuss some small improvements to the model that allow it to match direct Navier-Stokes simulations of rigid bodies with an average error of about $10-20\%$ in almost all motions considered, including pitching, heaving, and uniformly rotating plates. 
\newline

In our zero-viscosity model there is no diffusion of vorticity from the vortex sheets, only advection from the separation points, which we assume to occur at the two edges of the membrane only.  Consequently, two ``free" vortex sheets emanate from the membrane, one from each edge. The bound vortex sheet along the membrane and the free vortex sheets form one continuous vortex sheet that we compute at each time step. 
\newline

We assume that our membrane is held in a uniform stream with velocity $\boldsymbol{U}= [1,0]^T$. As time increases, vorticity emanates from both endpoints, and the lengths of these free vortex sheets increase at each edge, starting from zero. The bound vortex sheet coincides with the elastic membrane $\zeta(\alpha, t), -1 \leq \alpha \leq 1$. The free vortex sheets emanating from the $\alpha = \pm 1$ edges are denoted $\zeta_{+}(\alpha, t),\ \alpha\in [1,s_{+} + 1]$ and $\zeta_{-}(s,t),\ s \in [-s_{-} -1,-1]$. Here, $s_{\pm}$ are the lengths of the positive and negative vortex sheets respectively, and the vortex sheet strengths $\gamma_{\pm}(\alpha, t)$ are defined analogously.
\newline

The strength of the entire (bound and free) vortex sheet is denoted $\gamma(\alpha, t)$, with $\alpha \in[-1 -s_-, 1 + s_+]$. We denote the restriction of $\gamma(\alpha, t)$ to the negative, positive, and bound vortex sheets as $\gamma_-(\alpha, t)$, $\gamma_+(\alpha, t)$, and $\gamma_b(\alpha, t)$ respectively. Note that on the free vortex sheets, $\alpha = s$ is the arc length coordinate. On the bound vortex sheet (i.e.~the membrane) $\alpha$ is instead the material coordinate (the arc length coordinate of the undeflected membrane) and $\partial_\alpha s$ is the stretching factor.
\newline

Integrating the vortex sheet strengths $\gamma_b$, and $\gamma_\pm$ against the Biot-Savart kernel, we obtain an equation for the conjugate fluid velocity  $w(z) = u(z) - iv(z)$,
 \begin{equation}
    u(z) - i v(z) = \frac{1}{2\pi i}\int_{-1}^{1}{\frac{\gamma_b(\alpha, t)}{z - \zeta(\alpha, t)}\partial_\alpha s\,  d\alpha}  +\frac{1}{2\pi i}\int_{-s_{-} -1}^{-1}{\frac{\gamma_{-}(s,t)}{z - \zeta_{-}(s,t)}\, ds} + \frac{1}{2\pi i}\int_{1}^{s_{+} + 1}{\frac{\gamma_{+}(s,t)}{z - \zeta_{+}(s,t)}\, ds} + \overline{\boldsymbol{U}}.
\end{equation}

Here $\overline{\boldsymbol{U}}$ is the complex conjugate of the free-stream velocity. Let $\Gamma(s,t)$ denote the arc length integral of $\gamma$ along the positive and negative vortex sheets, defined piecewise:
\begin{align}
\Gamma(s,t) = \begin{cases}
    \int_{-s_--1}^s \gamma(s',t)\,ds', \quad s \in [-s_--1, -1] \\
   \int_{1}^{s}\gamma_+(s',t)\,ds', \quad s \in [1, 1+s_+]
\end{cases} 
\end{align}

and $\Gamma_\pm(t)$ denote the total circulations in the positive and negative vortex sheets, i.e.
\begin{align}
\Gamma_+(t) \equiv \int_{1}^{1+s_{+}}\gamma_{+}(s,t)\, ds \quad ; \quad \Gamma_-(t) \equiv \int_{-s_{-}-1}^{-1}\gamma_{-}(s,t)\, ds.
\end{align}

Then, by reparameterizing the positions of the free vortex sheets in terms of $\Gamma$, we may rewrite the equation for the conjugate velocity as
\begin{equation}
     u(z) - i v(z) = \frac{1}{2\pi i}\int_{-1}^{1}{\frac{\gamma_b(\alpha, t)}{z - \zeta(\alpha, t)}\partial_\alpha s\,  d\alpha}  + \frac{1}{2\pi i}\int_{0}^{\Gamma_-(t)}{\frac{d\Gamma}{z - \zeta_{-}(\Gamma,t)}} + \frac{1}{2\pi i}\int_{0}^{\Gamma_+(t)}{\frac{d\Gamma}{z - \zeta_{+}(\Gamma,t)}}+ \overline{\boldsymbol{U}}. \label{ConjFluidVel}
\end{equation}

\subsubsection{Birkhoff-Rott equation}
By demanding that the pressure be a continuous field across the free vortex sheets, the free vortex sheets must evolve with velocities equal to the averages of the fluid velocities across their surfaces \cite[84]{eldredge_inviscid_flows}. A consequence of the Sokhotski–Plemelj formulae is that the principal value of a singular integral is the average of its limiting values on either side. Hence, by evaluating \eqref{ConjFluidVel} on the free sheets in the principal-value sense, we obtain the following Birkhoff-Rott equations for the (conjugate) fluid velocity of the free sheets.
\begin{multline}
     \partial_{t}\bar\zeta_{\pm}(\Gamma,t) = \frac{1}{2\pi i}\int_{-1}^{1}{\frac{\gamma_b(s,t)}{\zeta_{\pm}(\Gamma,t) - \zeta(s,t)}\, ds} + \frac{1}{2\pi i}\int_{0}^{\Gamma_-(t)}{\frac{d\Gamma'}{\zeta_{\pm}(\Gamma,t) - \zeta_{-}(\Gamma',t)}} \\ + 
     \frac{1}{2\pi i}\int_{0}^{\Gamma_+(t)}{\frac{d\Gamma'}{\zeta_{\pm}(\Gamma,t) - \zeta_{+}(\Gamma',t)}} + \overline{\boldsymbol{U}}. \label{Birkhoff-Rott}
\end{multline}

Due to the singular kernel in the contribution of the free sheets, the vortex sheets develop a curvature singularity at a finite critical time \cite{moore1979spontaneous, Krasny_1986, Shelley_1992}. To subvert this issue, instead of \eqref{ConjFluidVel}, one can evaluate blob-regularized equations for the fluid velocity given by 

\begin{multline}
     u(z) - iv(z) = \frac{1}{2\pi i}\int_{-1}^{1}{\gamma_b(s,t)\frac{\overline{z- \zeta(s,t)}}{|z- \zeta(s,t)|^2 + \delta^2}\, ds} + \frac{1}{2\pi i}\int_{0}^{\Gamma_-(t)}{\frac{\overline{z - \zeta_{-}(\Gamma',t)}}{|z - \zeta_{-}(\Gamma',t)|^2 + \delta^2}d\Gamma'} \\ +
     \frac{1}{2\pi i}\int_{0}^{\Gamma_+(t)}{\frac{\overline{z - \zeta_{+}(\Gamma',t)}}{|z - \zeta_{+}(\Gamma',t)|^2 + \delta^2}d\Gamma'}+ \overline{\boldsymbol{U}}. \label{ConjVelSmoothed}
\end{multline}
on the free sheets. Here $\delta$ is a small parameter that needs to be specified. This results in the following blob-regularized Birkhoff-Rott equations
\begin{multline} \label{blob regularized birkhoff rott equations}
     \partial_{t}\bar\zeta_{\pm}(\Gamma,t)= \frac{1}{2\pi i}\int_{-1}^{1} \gamma_b(s,t) \frac{\overline{\zeta_\pm(\Gamma,t) - \zeta(s,t)}}{|\zeta_\pm(\Gamma,t) - \zeta(s,t)|^2 + \delta^2}\, ds
    + \frac{1}{2\pi i}\int_{0}^{\Gamma_-(t)}{\frac{\overline{\zeta_{\pm}(\Gamma,t) - \zeta_{-}(\Gamma',t)}}{|\zeta_{\pm}(\Gamma,t) - \zeta_{-}(\Gamma',t)|^2 + \delta^2}d\Gamma'} \\ + 
     \frac{1}{2\pi i}\int_{0}^{\Gamma_+(t)}{\frac{\overline{\zeta_{\pm}(\Gamma,t) - \zeta_{+}(\Gamma',t)}}{|\zeta_{\pm}(\Gamma,t) - \zeta_{+}(\Gamma',t)|^2 + \delta^2}d\Gamma'}+ \overline{\boldsymbol{U}}
\end{multline}
In this study we take $\delta = 0.1$ as was done in \cite{loo2026comparisoninviscidviscousvortex}.

\subsubsection{Computing the vorticity}

The Biot--Savart law states that the complex velocity induced by a blob-regularized point vortex with circulation $\Gamma_0$ situated at $z_0$ is
\begin{equation}
u(z)+iv(z)
=
\frac{i\Gamma_0}{2\pi}
\frac{z-z_0}{|z-z_0|^2+\delta^2}.
\label{blob vortex}
\end{equation}
Following \cite{holm2006euler}, by taking the curl of equation \eqref{blob vortex} the corresponding blob-regularized vorticity is
\begin{equation}
\omega_\delta(z)
=
\frac{\delta^2\Gamma_0}
{\pi\left(|z-z_0|^2+\delta^2\right)^2}.
\end{equation}
For a vortex sheet of length $L$ with cumulative circulation $\Gamma(s)$, $0<s<L$, whose position is parameterized by $z'(\Gamma(s))$, we define
\begin{equation}
W_\delta(z)
=
\frac{\delta^2}
{\pi\left(|z|^2+\delta^2\right)^2}.
\end{equation}
The blob-regularized vorticity induced by the vortex sheet is then
\begin{equation}
\omega_\delta(z)
=
\int_0^{\Gamma(L)}
W_\delta\!\left(z-z'(\Gamma')\right)d\Gamma'.
\label{vorticity formula}
\end{equation}

\subsubsection{No-penetration condition}
Let the body velocity $\partial_t \zeta(s,t)$ have tangential and normal components $\tau$ and $\nu$ respectively: 
\begin{eqnarray}
    \partial_t\zeta(s,t) = \tau(s,t)\hat{\textbf{s}}(s,t) + \nu(s,t)\hat{\textbf{n}}(s,t).
\end{eqnarray} 
Taking limits as $z$ approaches the body in \eqref{ConjVelSmoothed}, we see that the (conjugate) fluid velocity along the body is given by 
\begin{multline}
    \mu(s,t)\overline{\hat{\boldsymbol{s}}(s,t)} + \xi(s,t)\overline{\hat{\boldsymbol{n}}(s,t)} = \\ 
    \frac{1}{2\pi i}\dashint_{-1}^{1}{\gamma_b(s',t)\frac{\overline{\zeta(s,t) - \zeta(s',t)}}{|\zeta(s,t) - \zeta(s',t)|^2 + \delta^2}\, ds'}  +  \frac{1}{2\pi i}\int_{0}^{\Gamma_-(t)}{\frac{\overline{\zeta(s,t) - \zeta_{-}(\Gamma,t)}}{|\zeta(s,t) - \zeta_{-}(\Gamma,t)|^2 + \delta^2}d\Gamma} \\ + \frac{1}{2\pi i}\int_{0}^{\Gamma_+(t)}{\frac{\overline{\zeta(s,t) - \zeta_{+}(\Gamma,t)}}{|\zeta(s,t) - \zeta_{+}(\Gamma,t)|^2 + \delta^2}d\Gamma} + \overline{\boldsymbol{U}}.
\end{multline}
The no-penetration condition states that the fluid and body velocities have the same components normal to the body at the body surface. Hence, we see that
\begin{multline}
    \overline{\hat{\textbf{n}}(s,t)}\cdot \partial_t\bar{\zeta}(s,t)  = 
        \overline{\hat{\textbf{n}}(s,t)}\cdot \bigg(\frac{1}{2\pi i}\dashint_{-1}^{1}{\frac{\gamma_b(s',t)}{\zeta(s,t) - \zeta(s',t)}\, ds'} + \frac{1}{2\pi i}\int_{0}^{\Gamma_-(t)}{\frac{\overline{\zeta(s,t) - \zeta_{-}(\Gamma,t)}}{|\zeta(s,t) - \zeta_{-}(\Gamma,t)|^2 + \delta^2}d\Gamma} \\ + \frac{1}{2\pi i}\int_{0}^{\Gamma_+(t)}{\frac{\overline{\zeta(s,t) - \zeta_{+}(\Gamma,t)}}{|\zeta(s,t) - \zeta_{+}(\Gamma,t)|^2 + \delta^2}d\Gamma} + \overline{\boldsymbol{U}}\bigg).
    \label{NoPenetration}
\end{multline}
Or more concisely, 
\begin{equation}
    \nu(s,t) = \xi(s,t)
\end{equation}
 Note that we have removed the blob regularization from the contribution of the body to the fluid velocity. As explained in \cite{loo2026comparisoninviscidviscousvortex}, any regularization to this term in the no-penetration condition makes solving for $\gamma_b$ ill-posed. This results in a mismatch with the Birkhoff-Rott equations that introduces an $\mathcal{O}(\delta)$ error into the no-penetration condition. This error can be controlled by projecting the velocities of points that violate the no-penetration condition onto the tangent space of the nearest grid point on the body. The result is a scheme that is approximately first-order in time \cite{loo2025falling}.
 \newline
 
 Given known positive and negative circulations, the no-penetration condition \eqref{NoPenetration} uniquely determines $\gamma_b$ provided we add the additional constraint that the total circulation in the fluid vanishes \cite{Golberg1990IntroductionTT} \cite[285]{eldredge_inviscid_flows}. Indeed, assuming the fluid is initially at rest, the initial circulation present in the fluid is zero. Combined with Kelvin's circulation theorem, which asserts that the total circulation is independent of time \cite[36]{eldredge_inviscid_flows}, we must have
\begin{equation} \label{Kelvin's Circulation Theorem}
    \Gamma_+(t) + \Gamma_-(t) + \Gamma_b(t) = 0.
\end{equation}
Here, $\Gamma_b(t) = \int_{-1}^1\gamma_b(s,t)\, ds$ is the total circulation on the body. Alongside the no-penetration condition \eqref{NoPenetration}, this is an additional constraint that must be satisfied by the bound vortex sheet strength $\gamma_b$. At each time step, $\Gamma_\pm(t)$ are chosen such that the Kutta conditions 
\begin{equation}
    \gamma_b(\pm1) = \gamma_\pm(\pm 1)
    \label{kutta condition}
\end{equation}are satisfied. 
Mathematically, this is the assertion that the circulation density $\gamma$ is continuous across the entire vortex sheet.
\newline

For a self-contained and comprehensive treatment of the inviscid fluid equations and their numerical implementation, we refer to \cite{loo2025falling}, which develops the same framework in similar notation for the situation of a falling rigid plate.

\section{The pressure jump equation as a conservation law}
\label{pressure jump equation as a conservation law section}
Having described both the fluid and the membrane model, we now discuss how we compute the force density 
\begin{equation}
    \boldsymbol{f} = [p]^{-}_+\hat{\boldsymbol{n}}\partial_\alpha s
\end{equation}
which couples the membrane with the surrounding fluid. This is done through the pressure jump equation:

\begin{equation}
    \partial_t\gamma + \partial_s((\mu - \tau)\gamma + [p]_{+}^{-})=0.
    \label{vortex sheet strength equation}
\end{equation}
A complete derivation of the equations studied in this section can be found in appendix \ref{pressure jump derivation appendix}. It is convenient to view the pressure jump equation as a conservation law in the bound vorticity $\gamma_b$, as it helps to explain the dynamics in the following simulation results. Indeed, consider figure \ref{pressure and vortex sheet strength example figure}. This figure shows the typical vortex sheet strength $\gamma_b$, pressure jump $[p]^-_+$ and relative fluid velocity $\mu - \tau$ at various times demonstrated by a fixed-free membrane held at $\beta = 30^\circ$ in a uniform flow with $(R_1,R_2,R_3) = (10,10^{-6},10^{2})$ . Panel (a) shows the pressure jump $[p]^-_+$ (blue) and bound vortex sheet strength $\gamma_b$ (orange) at time $t = 7.1$, and panel (b) shows the corresponding vorticity field. Panels (c-d) show the same but at $t = 11.1$. 
\newline 

In both panels (a) and (c), the pressure jump has a single maximum away from the edges, induced by an attached leading and trailing edge vortex shown in panels (b) and (d) respectively. Moreover, both panels show that this maximum shares a critical point with the bound vortex sheet strength $\gamma_b$.

\begin{figure}[H]
    \centering
    \includegraphics[width=0.8\linewidth, trim = 2cm 0cm 0cm 0cm, clip]{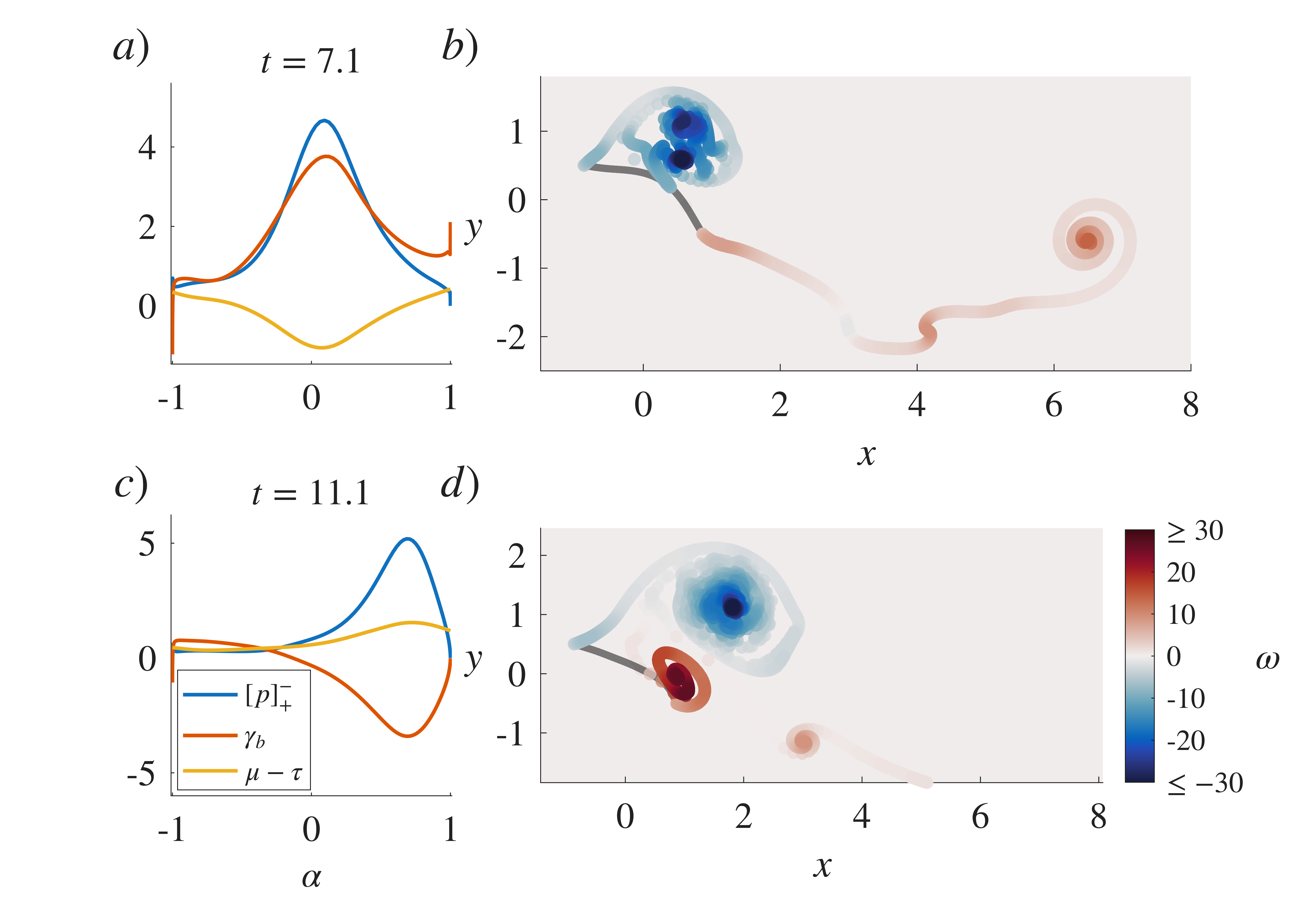}
    \caption{(a) The bound vortex sheet strength $\gamma_b$, pressure jump $[p]^-_+$, and relative tangential fluid velocity $\mu - \tau$, as well as (b) the corresponding vorticity field $\omega$ at time $t = 7.1$. Panels (c-d) show the same at $t = 11.1$. Here $(R_1,R_2,R_3) = (10,10^{-6},10^{2})$ and the membrane is held at angle of attack $\beta = 30^\circ$. }
    \label{pressure and vortex sheet strength example figure}
\end{figure}

This observation can be explained by viewing \eqref{vortex sheet strength equation} as a conservation law. To see this, consider the Euler equation given by 
\begin{equation}
    \partial_t\boldsymbol{u} + \boldsymbol{u}\cdot\boldsymbol{\nabla}\boldsymbol{u} + \boldsymbol{\nabla}p = 0
    \label{euler equation}
\end{equation}
The Euler equation can be placed in the form of a conservation law using tensor notation. Indeed, by incompressibility $\boldsymbol{u}\cdot\boldsymbol{\nabla}\boldsymbol{u} = \boldsymbol{\nabla}\cdot(\boldsymbol{u}\otimes\boldsymbol{u})$. Hence \eqref{euler equation} can be written as 
\begin{equation}
     \partial_t\boldsymbol{u} + \boldsymbol{\nabla}\cdot(\boldsymbol{u}\otimes\boldsymbol{u} + pI) = 0
     \label{euler equation as conservation law}
\end{equation}
This is a conservation law for $\boldsymbol{u}$ with (matrix-valued) flux $\boldsymbol{F}(\boldsymbol{u}) = \boldsymbol{u}\otimes\boldsymbol{u} + pI$. Notice the pressure term in the flux. This term drives fluid momentum from regions of high to low pressure.
\newline

In the reference frame of a fluid parcel following a material point on the membrane with velocity $\partial_t\zeta$, Euler's equation becomes
\begin{equation}
     \partial_t\boldsymbol{u} + \boldsymbol{\nabla}\cdot(\boldsymbol{w}\otimes\boldsymbol{w} + pI) = 0
     \label{euler equation as conservation law in reference frame}
\end{equation}
where $\boldsymbol{w} = \boldsymbol{u} - \partial_t\zeta$ (see appendix \ref{pressure jump derivation appendix}). Projecting this equation onto the membrane surface with tangent $\hat{\boldsymbol{s}}$ results in the following equations:
\begin{equation}
    \partial_t u^\pm_s +\partial_s\big(\frac{1}{2}(u_s^\pm - \tau)^2 + p^\pm \big) = 0.
    \label{1D euler equations for slip}
\end{equation} 
Here $u_s = \hat{\boldsymbol{s}}\cdot\boldsymbol{u}$ is the tangential component of $\boldsymbol{u}$ and the $(\cdot)^\pm$ superscript indicates its limiting value on the $\pm$ side of the membrane. In plain language, equations \eqref{1D euler equations for slip} say that the slip velocities $u_s^\pm$ obey a 1D Euler equation. Since $\gamma = u_s^- - u_s^+=[u_s]^-_+$ (the jump in the slip velocity across the membrane surface), by subtracting both $\pm$ equations, we obtain 

\begin{equation}
    \partial_t\gamma + \partial_s((\mu - \tau)\gamma + [p]_{+}^{-})=0,
    \label{vortex sheet strength conservation law}
\end{equation}
where $\mu = \frac{1}{2}(u_s^+ + u_s^-)$. Note that this equation uniquely determines the pressure jump $[p]^-_+$ from the tangential components of the fluid and body velocities. The resulting equation \eqref{vortex sheet strength conservation law} is the difference between two 1D Euler equations, both of which are conservation laws. Hence equation \eqref{vortex sheet strength conservation law} can be thought of as a conservation law as well.
\newline

Indeed, view the union of the bound and free vortex sheet as a single continuous vortex sheet with strength $\gamma(s)$ parameterized by the arc length parameter $s$. Suppose also that the motion of the membrane is known at each instant. By extension, the tangential velocity of the membrane $\tau$ is known at each instant as well. Then we have
\begin{equation}
    \partial_t\gamma + \partial_s((\mu - \tau)\gamma + [p]_{+}^{-}):= \partial_t\gamma + \partial_s\big(f(\gamma)\big)=0,
    \label{vortex sheet strength conservation law in flux form}
\end{equation}
where $f(\gamma) = (\mu - \tau)\gamma + [p]^-_+$. Consequently, the pressure jump equation may be viewed as a conservation law for $\gamma$ with flux operator $f(\gamma)= (\mu-\tau)\gamma +[p]^-_+$.
\newline

The conservation form of \eqref{vortex sheet strength conservation law in flux form}, together with two observations from the simulations, is enough to explain the coincidence of critical points seen in figure \ref{pressure and vortex sheet strength example figure}. Firstly, when a leading- or trailing-edge vortex is attached, the bound vortex sheet strength evolves coherently over intervals short compared with the shedding time; that is, locally in time $\gamma$ is approximately a traveling wave,
\begin{equation}
    \gamma(s,t) \approx \tilde{\gamma}(s - ct)
    \implies
    \partial_t\gamma \approx -c\,\partial_s\gamma,
    \label{travelling wave ansatz}
\end{equation}
for some speed $c = c(t)$. Secondly, over the same interval the relative tangential velocity $\mu - \tau$ varies slowly in $s$ compared with $\gamma$ and $[p]^-_+$ (figure \ref{pressure and vortex sheet strength example figure}(a,c)), so that $\gamma\,\partial_s(\mu - \tau)$ is negligible compared to $(\mu - \tau)\,\partial_s\gamma$. Substituting \eqref{travelling wave ansatz} into \eqref{vortex sheet strength conservation law in flux form} and expanding the flux gives
\begin{equation}
    (\mu - \tau - c)\,\partial_s\gamma
    + \gamma\,\partial_s(\mu - \tau)
    + \partial_s[p]^-_+ = 0,
    \label{expanded flux relation}
\end{equation}
and discarding the second term by the slow-variation assumption,
\begin{equation}
    \partial_s[p]^-_+ \approx -(\mu - \tau - c)\,\partial_s\gamma.
    \label{pressure gradient relation}
\end{equation}
Integrating once in $s$ yields the local algebraic relation
\begin{equation}
    (\mu - \tau - c)\,\gamma + [p]^-_+ \approx C(t).
    \label{algebraic conservation law}
\end{equation}
In plain language, equation \eqref{algebraic conservation law} states that where the relative tangential velocity is nearly constant, the vortex sheet strength is an affine function of the pressure. Equation \eqref{pressure gradient relation} shows that generically, $\partial_s\gamma$ and $\partial_s[p]^-_+$ vanish at the same arc length locations, which is precisely the behavior observed in figure \ref{pressure and vortex sheet strength example figure}(a,c). 
This makes intuitive sense. Due to the pressure flux, signed fluid momentum is transported from high to low pressure regions in the Euler equation. As the solution to the difference of two 1D Euler equations, $\gamma$ too is transported from regions of high to low pressure jump; hence $\gamma$ tends to ``pool" at the critical points of the pressure.
\newline

Panels (a) and (c) of figure~\ref{pressure and vortex sheet strength example figure} show the bound vortex sheet strength attaining an interior maximum and minimum, respectively. Their signs can be understood from Kelvin's circulation theorem~\eqref{Kelvin's Circulation Theorem}. Let \(\Gamma_f=\Gamma_-+\Gamma_+\) denote the total circulation in the free vortex sheets. Kelvin's theorem gives \(\Gamma_b=-\Gamma_f\), where \(\Gamma_b\) is the total bound circulation. At the time shown in panel (b), the free circulation is dominated by the negative leading-edge vortex, so \(\Gamma_f<0\) and hence \(\Gamma_b>0\). Since the bound vortex sheet strength is small near the membrane endpoints and has a single interior critical point, this critical point is a maximum, as observed in panel (a). Conversely, at the time shown in panel (d), the positive attached trailing-edge vortex dominates the free circulation, so \(\Gamma_f>0\) and \(\Gamma_b<0\); the corresponding critical point is therefore a minimum, as observed in panel (c).
\newline

Crucially, Kelvin's theorem and the pressure-jump conservation law together provide a mechanism for alternating vortex shedding. An attached vortex of one sign requires compensating bound circulation of the opposite sign on the membrane. The conservation-law flux concentrates and transports this bound vorticity along the membrane until it reaches an edge and is transferred to the free vortex sheet, where it rolls up into the next attached vortex with the opposite sign. This new vortex reverses the sign of the compensating bound circulation required by Kelvin's theorem, and the cycle repeats, producing the alternating-sign pattern of vortices in the wake.

\section{Numerical results}
\label{numerical results section}
In this section, we explore the space of membranes parameterized by $(R_1,R_2,R_3)$ in both fixed-fixed and fixed-free configurations and highlight the main results. In all cases, at $t = 0$ the background flow changes from 0 to 1 impulsively and the membrane is straight, unstretched, and at a fixed angle of attack. In the fixed-fixed case, the angle of attack is fixed for all time at $\beta$, while in the fixed-free configuration, the angle of attack is only set at $t = 0$ and is taken to be small, specifically $1^\circ$.

\subsection{Fixed-fixed membranes}
\label{fixed fixed membranes section}

Fixed-fixed membranes with zero bending rigidity in an inviscid fluid have been studied extensively in 2D \cite{mavroyiakoumou2020large,mavroyiakoumou2021eigenmode, mavroyiakoumou2025sail} as well as 3D \cite{mavroyiakoumou2022membrane} in the absence of leading-edge shedding at zero angle-of-attack. In this section we expand on these results by including leading-edge shedding effects that allow us to study the fluid-membrane interaction at larger angles of attack.

\subsubsection{Zero angle of attack}
We begin by revisiting the zero-angle-of-attack case studied in \cite{mavroyiakoumou2020large}, this time in the presence of leading-edge shedding. In figure \ref{fig: effect of leading-edge shedding at zero angle of attack} we show the effect of leading-edge shedding on the vorticity field at early times; the leading-edge vortex is shed beneath the membrane, forming a concentrated recirculation region there.
\begin{figure}[H]
    \centering
    \includegraphics[width=1\linewidth, trim = 2cm 0cm 0cm 1.5cm, clip]{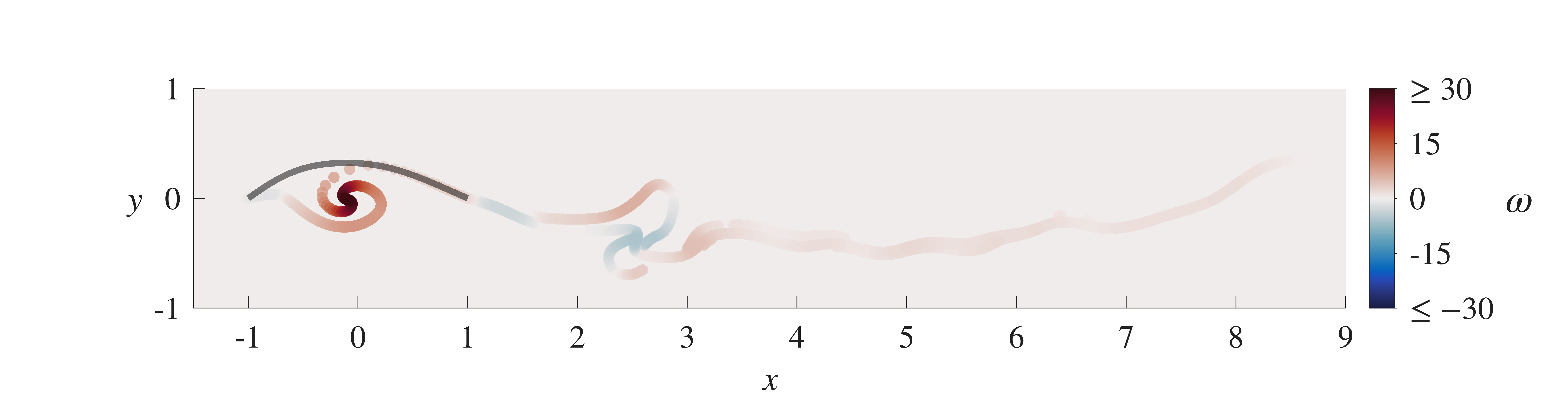}
    \caption{Snapshot of the vorticity field at time $t = 7$ for a membrane with $(R_1,R_2,R_3) = (1,0,20)$ held in a uniform stream at zero angle of attack $(\beta = 0)$.}
    \label{fig: effect of leading-edge shedding at zero angle of attack}
\end{figure}
After $t\approx 10$, this recirculation region is swept away, and by $t \approx 100$, the membrane approximately reaches a steady shape. In figure \ref{fig: steady state shape zero angle} we show the lift $C_L$, drag $C_D$, and steady shape of an elastic membrane with $R_2 = 0$ at zero angle of attack for $R_1 = 0,\ 10^{-3},\ 10^{-2},\ 10^{-1},\ 10^{0}$ respectively. The 3 rows correspond to  $R_3 = 20,\ 100,\ 500$ respectively. The steady shape is predictably independent of the density $R_1$, since the body inertia term is zero at steady state. Surprisingly however, it also only has a very mild effect on the initial transient period before the lift stabilizes at a constant value. This can be seen clearly in the first column of figure \ref{fig: steady state shape zero angle}.

\begin{figure}[H]
    \centering
    \includegraphics[width=1\linewidth, trim = 3cm 1cm 0cm 0cm, clip]{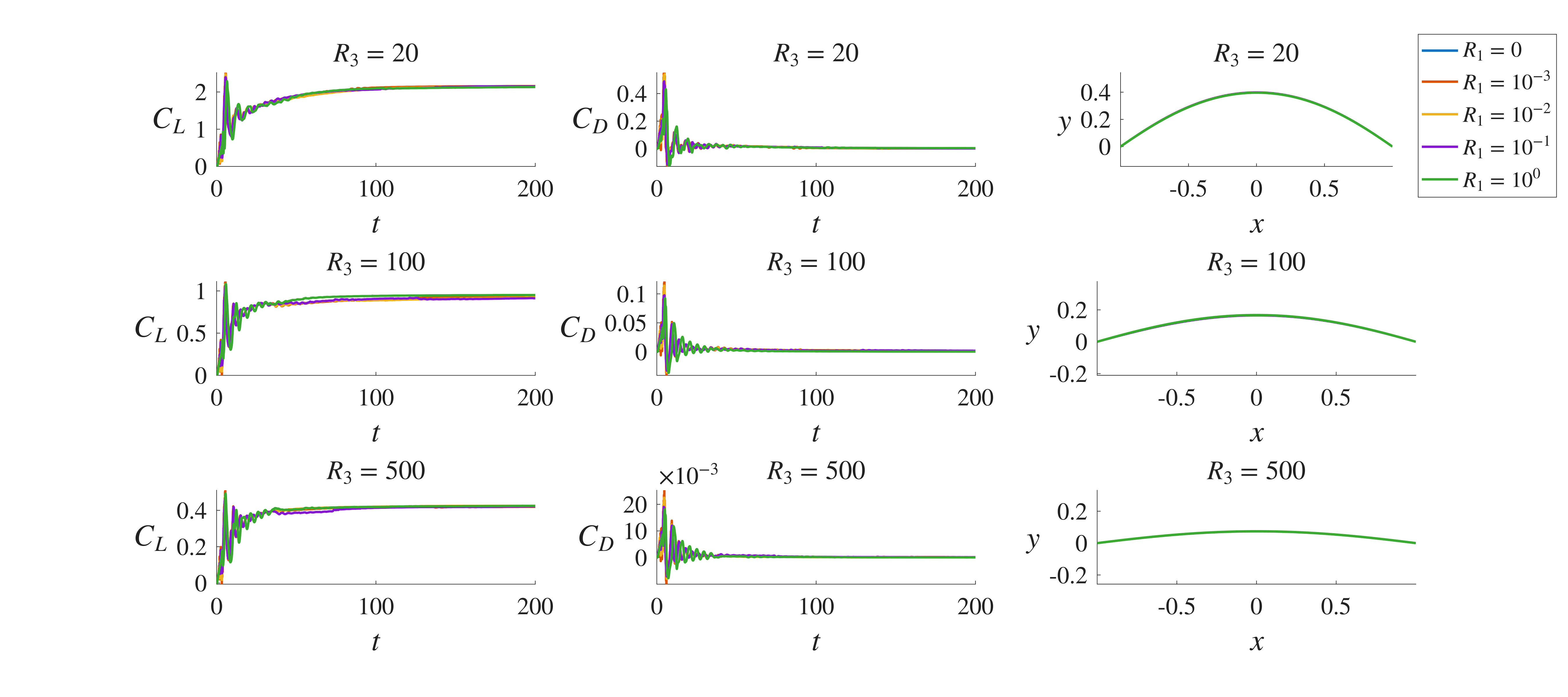}
    \caption{The lift $C_L$, drag $C_D$, and steady state membrane shape of an elastic membrane with $R_2 = 0$ at zero angle of attack for $R_1 = 0,\ 10^{-3},\ 10^{-2},\ 10^{-1},\ 10^{0}$ respectively. The three rows correspond to  $R_3 = 20,\ 100,\ 500$ respectively.}
    \label{fig: steady state shape zero angle}
\end{figure}
Decreasing the stretching modulus $R_3$ (therefore making the membrane more extensible) increases the steady state lift coefficient $C_L$. Although it increases the start-up drag $C_D$, in the steady state, the drag coefficients of all membranes vanish. Hence at zero angle of attack, elasticity confers significantly more lift, by cambering to decrease the pressure above the membrane.
\subsubsection{Moderate angles of attack}
Here we consider angles of attack $\beta = 30^\circ,60^\circ$. At these larger angles of attack, the leading-edge vortex is shed above the membrane instead and forms a concentrated negative recirculation region there. This is shown in the first column of figure \ref{fig: elasticity effect on LES 30 deg} with snapshots of vorticity field induced by membranes held at $\beta = 30^\circ$ taken at $t = 7.5$. Above each panel we list the average lift and drag coefficients ($\langle C_L\rangle$ and $\langle C_D\rangle$ respectively) from $t$ = 0 to 7.5. The second column shows the lift and drag coefficients (blue and orange respectively) against time $t$ until the final time $T = 200$. In all of the panels $R_1 = R_2 = 0$ and the rows correspond to $R_3 = 20,\ 100,\ 500,\ \infty$. Figure \ref{fig: elasticity effect on LES 60 deg} shows the same quantities for $\beta = 60^\circ$ and snapshot time $t = 5$. In both figures, the lift and drag coefficient plots of the second column show rapid oscillations leading up to the snapshot times, for $R_3 < \infty$. The impulsive start causes rapid oscillations in the positions of the membrane and leading-edge vortex sheet before it establishes the spiral shown in the snapshots.
\newline

\begin{figure}[H]
    \centering
    \includegraphics[width=0.9\linewidth, trim = 1cm 2.5cm 1cm 1cm, clip]{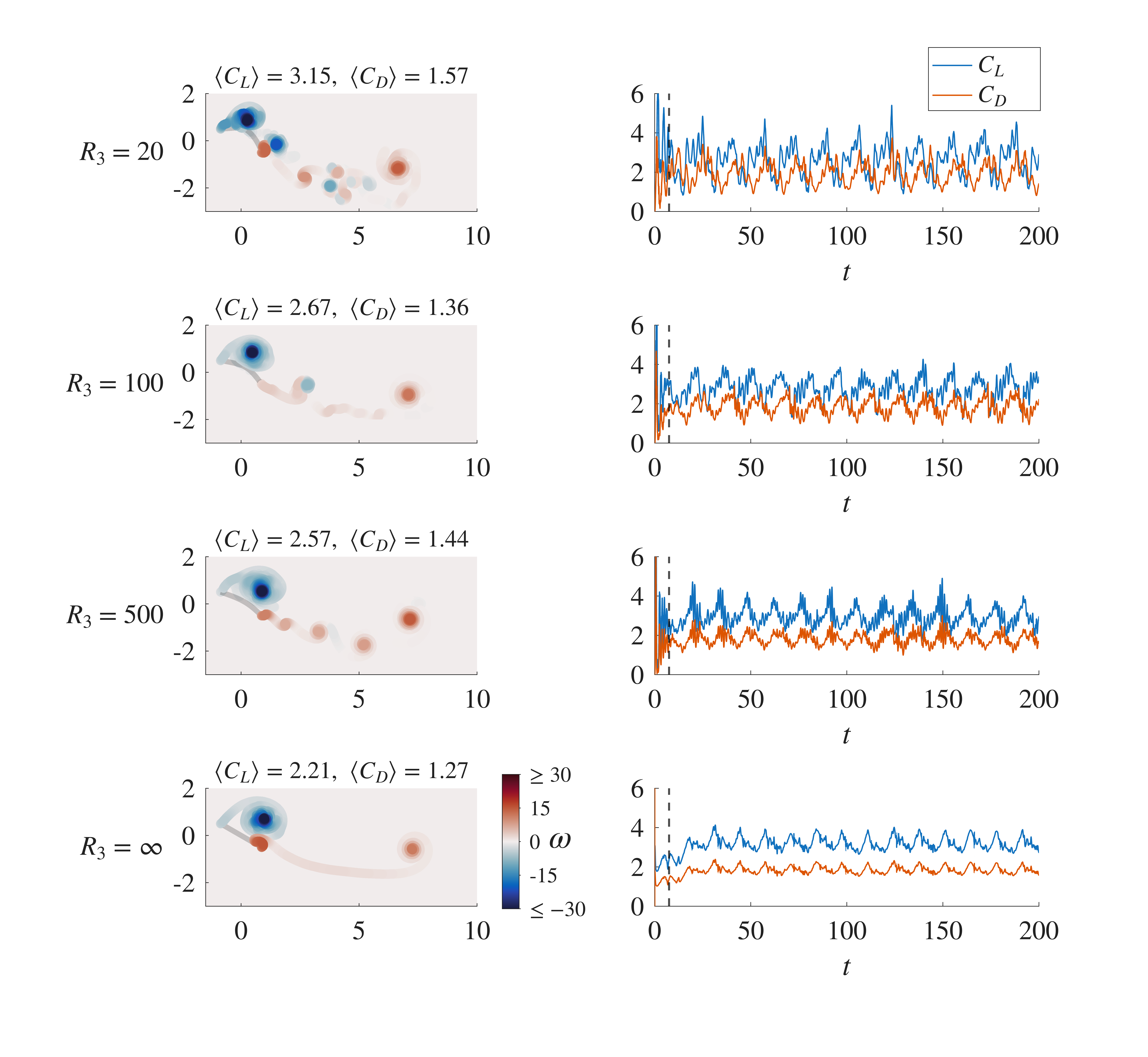}
    \caption{Vorticity fields at $t = 7.5$ (left column) and lift and drag coefficients ($C_L$, $C_D$ respectively) versus time $t$ (right column). In all panels $R_1 = R_2 = 0$ and the rows correspond to $R_3 = 20,\ 100,\ 500,\ \infty$. The average lift and drag ($\langle C_L\rangle$ and $\langle C_D\rangle$ respectively) until the snapshot time is displayed on the top right of the vorticity snapshots; the snapshot time $(t= 7)$ is marked with a dotted line in the right column.}
    \label{fig: elasticity effect on LES 30 deg}
\end{figure}

At shown in figure \ref{fig: elasticity effect on LES 30 deg}, at $\beta =30^\circ$, decreasing $R_3$ increases $\langle C_L\rangle$, by as much as $3.15/2.21 \approx 1.43$ times at the price of a minor $1.57/1.27 \approx 1.23$ times increase in drag for $R_3 =20$ versus a rigid plate (i.e. $R_3 = \infty$).
At $\beta =60^\circ$, decreasing $R_3$ increases $\langle C_L\rangle$, by as much as $2.71/1.75 \approx 1.55$ times at the price of a (comparatively) minor $3.61/3.03 \approx 1.19$ times increase in drag for $R_3 =20$ versus a rigid plate (i.e. $R_3 = \infty$). In both cases, decreasing $R_3$ (and thereby increasing extensibility)  increases both the average lift $\langle C_L\rangle$ and efficiency $\langle C_L\rangle/\langle C_D\rangle$ by at least $40\%$ and $15\%$ respectively during the transient start-up. This is consistent with the experiments in \cite{gehrke2025highly,gehrke2022aeroelastic} on cambered deformable wings in hovering motions.

\begin{figure}[H]
    \centering
    \includegraphics[width=0.9\linewidth, trim = 1cm 2.5cm 1cm 0.5cm, clip]{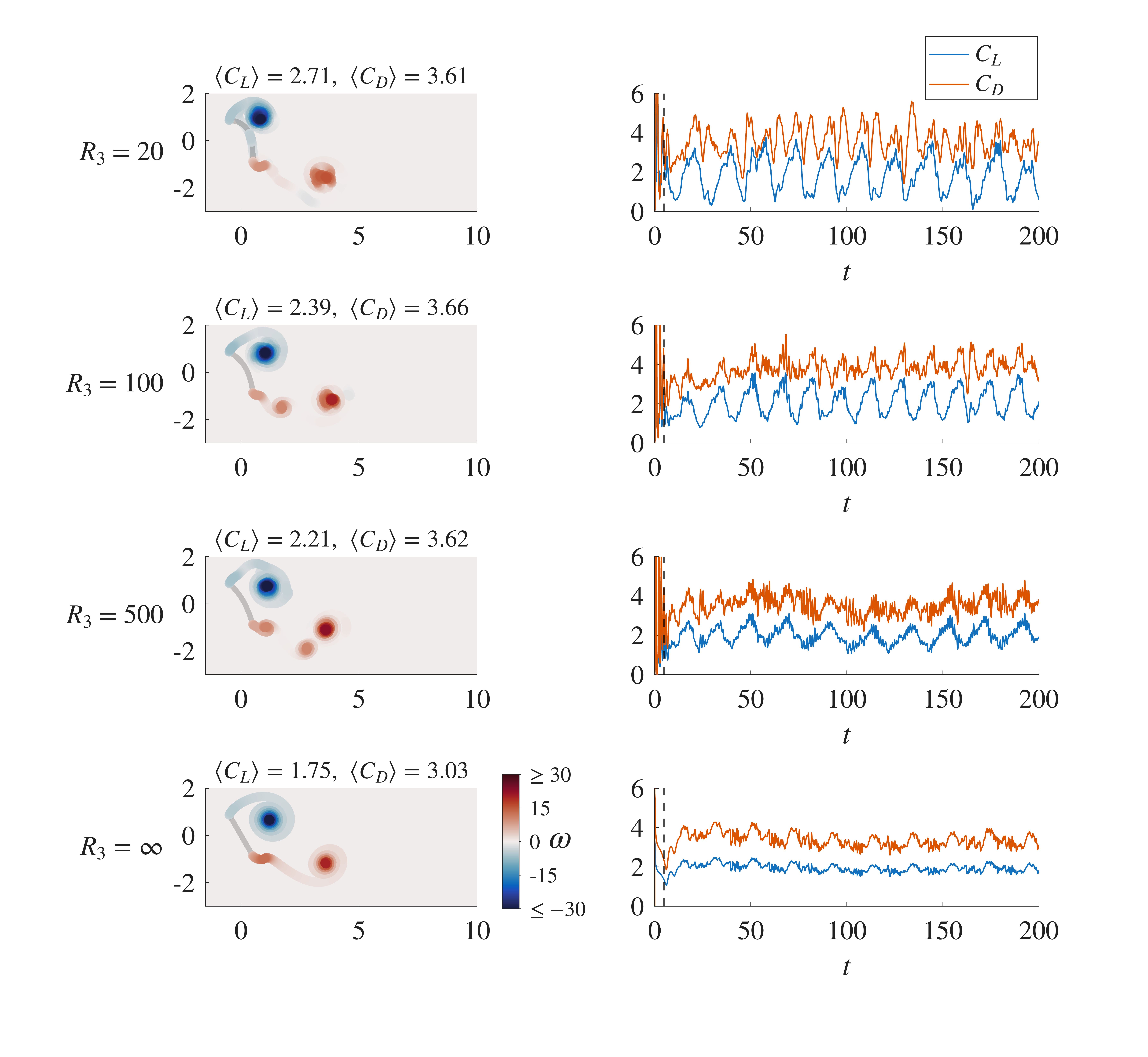}
    \caption{Same quantities as in figure \ref{fig: elasticity effect on LES 30 deg} but with snapshot time $t = 5$ and $\beta = 60^\circ$. }
    \label{fig: elasticity effect on LES 60 deg}
\end{figure}

As can be seen in the vorticity snapshots of figure \ref{fig: elasticity effect on LES 60 deg} (and to a smaller extent in figure \ref{fig: elasticity effect on LES 30 deg}), this increase in both lift and efficiency is achieved by reducing the distance between the membrane body and the leading-edge vortex. Elasticity allows this by extending the body towards the low pressure region near the leading-edge vortex and enhancing the lift; as shown in figure \ref{fig: elasticity effect on LES 60 deg}, as $R_3$ decreases, lift increases, and the leading-edge vortex gets progressively closer to the body. A similar interaction was observed in \cite{point_vortex_attraction} between a point vortex and a flexible filament.
\newline

In both figures \ref{fig: elasticity effect on LES 30 deg} and \ref{fig: elasticity effect on LES 60 deg}, the lift and drag coefficients ($C_L, C_D$) are quasi-periodic and oscillate about nonzero mean values. Figure \ref{fig: average lift comparison heatmap} shows the time-averaged lift $\langle C_L\rangle$ and drag $\langle C_D\rangle$ until $T = 200$ against $R_1,R_3$ for a fixed-fixed membrane held at angle of attack $\beta = 30^\circ$ (top row) and $60^\circ$ (bottom row). When $\beta = 30^\circ$, the steady state lift decreases monotonically as $R_3$ decreases, the opposite of the lift behavior during the transient start-up period where extensibility affords a decisive advantage in lift and efficiency. This apparent contradiction can be explained by the quasi-periodic shedding caused by the membrane. The membrane periodically sheds an alternating sequence of leading and trailing edge vortices. While the elasticity of the membrane allows it to stretch towards the leading-edge vortex to enhance peak lift, it also allows it to extend towards the trailing edge vortex, leading to larger troughs, i.e. smaller minima. Extensibility thus increases the amplitude of the oscillation of both the lift and drag about the mean value. This can be clearly seen in the right columns of figures \ref{fig: elasticity effect on LES 30 deg} and \ref{fig: elasticity effect on LES 60 deg}. The values in figure \ref{fig: average lift comparison heatmap} often have very small variations with $R_1$, even though the dynamics are often highly unsteady and the typical oscillation frequency varies strong with $R_1$.
\newline

When $\beta = 60^\circ$ the peak lift and drag occur at an intermediate $R_3$ of 100. Hence greater extensibility does not confer a definitive advantage in increasing steady state lift at all angles of attack. As we will see for fixed-free membranes, with very small $R_3$, the stretching and deformations may be very large and the membrane may no longer resemble a cambered airfoil in this regime.

\begin{figure}[H]
    \centering
    \includegraphics[width=0.85\linewidth, trim = 2cm 0.5cm 2cm 1cm, clip]{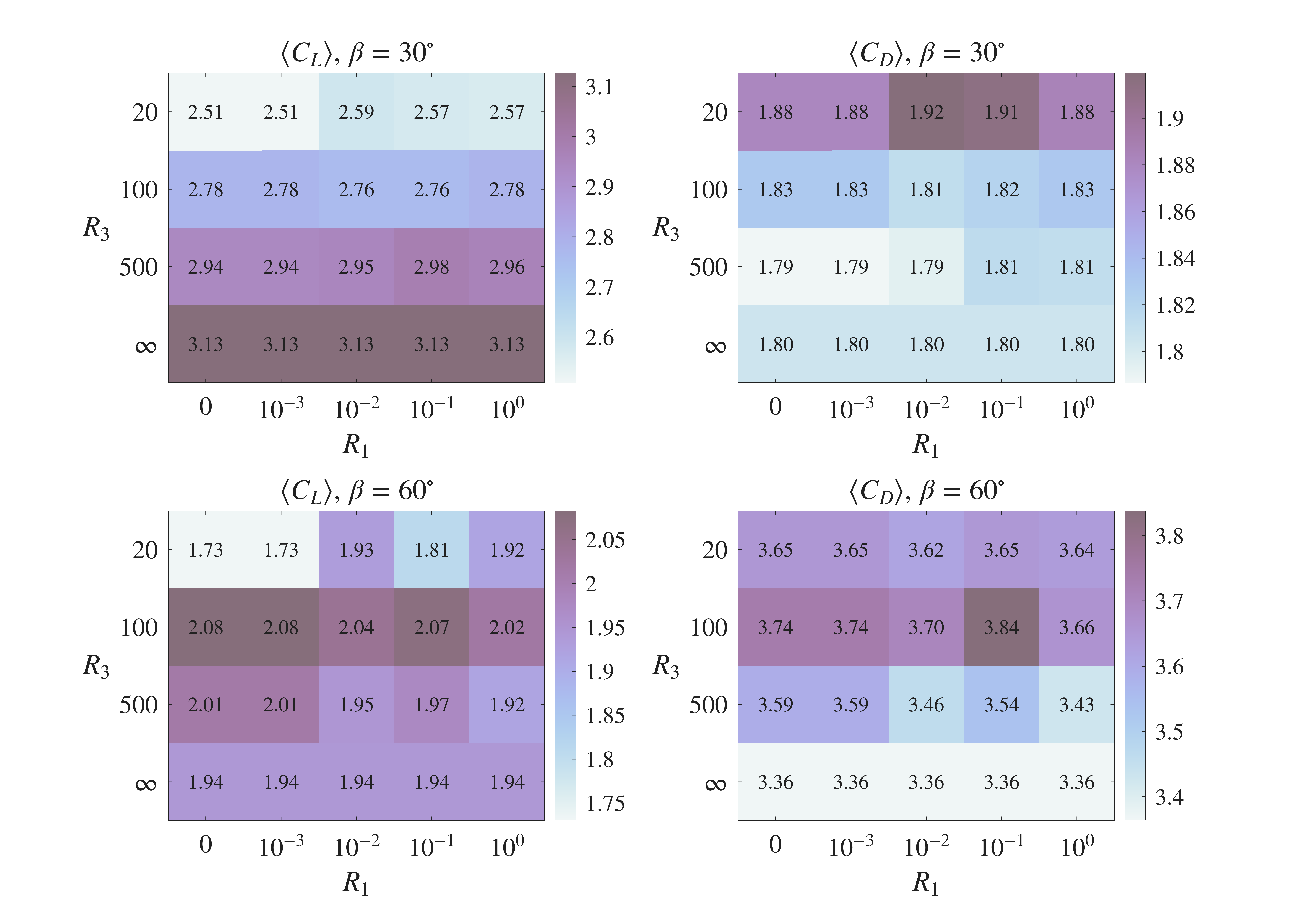}
    \caption{Heat map showing the time-averaged lift $\langle C_L\rangle$ and drag $\langle C_D\rangle$ against $R_1,R_3$ for a fixed-fixed membrane held at angle of attack $\beta = 30^\circ$ (top row) and $60^\circ$ (bottom row).}
    \label{fig: average lift comparison heatmap}
\end{figure}

\subsubsection{Normal angle of attack}

At a $90^\circ$ angle of attack, the body is normal to the oncoming flow. The flow is thus symmetric about the membrane centerline, the time-averaged lift vanishes, and the aerodynamic effect of compliance appears entirely in the drag. Figure \ref{fig: normal drag} shows the vorticity fields for $R_3 = 20,\ 100,\ 500,\ \infty$ with the time-averaged drag listed above each panel. As $R_3$ decreases, the membrane bows farther downstream toward the low-pressure cores of the counter-rotating vortex pair. This strengthens the body–vortex interaction and increases the drag averaged over $0\leq t\leq 10$ monotonically, from $\langle C_D\rangle=2.38$ for a rigid plate to $3.26$ for $R_3=20$, an increase of approximately $37\%$. Thus, the same geometric mechanism that enhances lift at moderate angles amplifies drag in the normal-incidence case, where symmetry precludes any lift benefit. We note that the initial vibrations induced by elasticity cause the resulting vorticity field to become increasingly irregular as $R_3$ decreases and extensibility increases.

\begin{figure}[H]
    \centering
    \includegraphics[width=1\linewidth, trim = 1.5cm 0cm 0cm 0.5cm, clip]{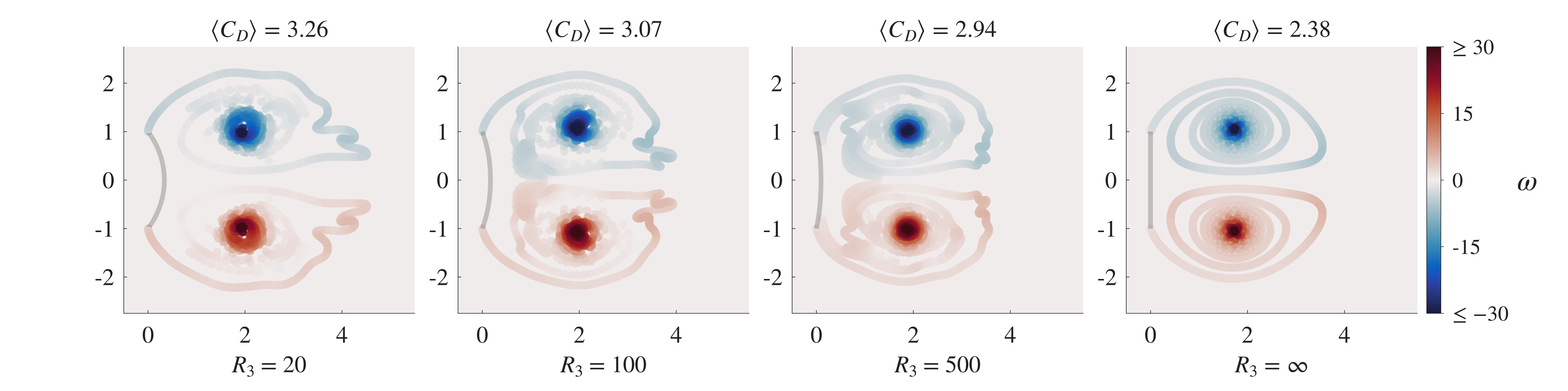}
    \caption{Vorticity fields at time $t = 10$ for membranes with $R_1 = R_2 = 0$ and $R_3 = 20,\ 100,\ 500,\ \infty$ held normal to a uniform flow. The average drag until the current time (t = 10) is shown above each plot and the corresponding stretching modulus $R_3$ is shown below.}
    \label{fig: normal drag}
\end{figure}

\subsubsection{Compression induced wrinkling during stall}
\label{compression example section}
In this section we show that for heavy membranes, vortex detachment during stall can sometimes induce compression or negative tension. This can be explained intuitively as follows: As the leading or trailing edge vortex develops, it forms a low pressure region that pulls the membrane towards it \cite{point_vortex_attraction}, creating a concentrated pressure extremum. When the vortex detaches, the sudden loss of fluid loading causes the tension to drop rapidly as the membrane relaxes, turning negative in some regions. Physically, this mechanism resembles violently plucking a heavy string.
\newline

In figure \ref{membrane compression example figure} we demonstrate this with a single example: a fixed-fixed membrane held at a constant angle of attack $\beta = 30^\circ$ in a uniform flow with $(R_1,R_2,R_3) = (10,10^{-6},20)$. Panels (a-c) show the vorticity $\omega$, pressure jump $[p]^-_+$ and tension $T$ at time $t = 12.75$, right after the trailing-edge vortex begins to detach and the membrane starts to rebound. In panel (a), a small red box indicates the regions shown in panels (d-h), which show snapshots of the membrane leading edge at times $t = 12.25,12.5,12.75,13,13.25$ respectively. These snapshots show the wrinkling nucleating from the negative tension regions shown in panel (c). As shown in panel (b), the pressure field oscillates spatially, and acts to smooth out the high-frequency wrinkles induced by compression. Indeed, by $t = 13.25$ shown in panel (h), the wrinkles have mostly smoothed into lower frequency wavelets. 
\begin{figure}[H]
    \centering
    \includegraphics[width=1.0\linewidth, trim= 5cm 2cm 5cm 0.5cm]{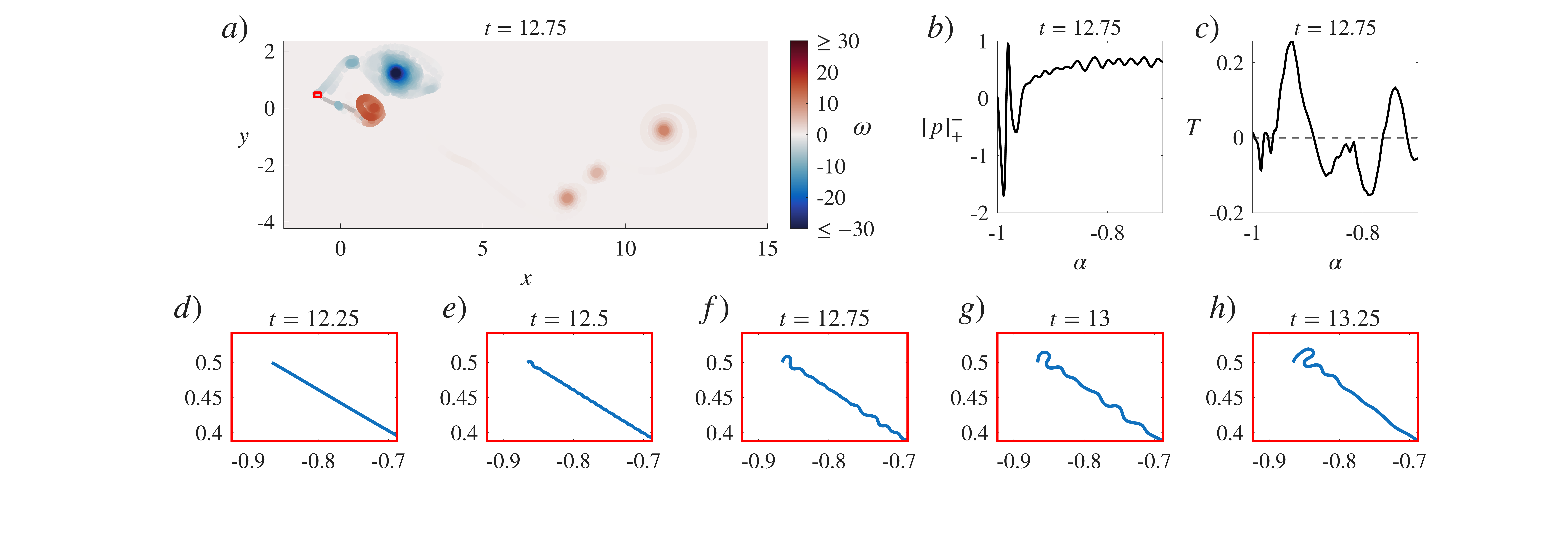}
    \caption{Panels (a-c) show the vorticity $\omega$, pressure jump $[p]^-_+$ and tension $T$ at time $t = 12.75$ for a fixed-fixed membrane with $(R_1,R_2,R_3) = (10,10^{-6},20)$ held at a $30^\circ$ angle in a uniform flow. A red box indicates the region shown in panels (d-h) which show snapshots of the membrane leading edge at times $t = 12.25,12.5,12.75,13,13.25$ respectively. }
    \label{membrane compression example figure}
\end{figure}

As figure \ref{membrane compression example figure}(a), shows, just as the trailing edge vortex detaches, it ``plucks" the membrane, resulting in a sudden loss of fluid loading that incites the formation of high frequency wrinkles near the leading edge. This is shown more clearly in panel (c) of the supplementary movie \href{https://drive.google.com/file/d/1CGcG94oydSZJySB4c7xGAOxGPoc7pN9P/view?usp=drive_link}{``wrinkling\_movie.avi"}. Panel (a) of the supplementary movie shows the full membrane with a red box at the leading edge indicating the close up region (shown in panel (b)) where the wrinkling is concentrated.
\newline

As the dispersion relation \eqref{dispersion relation equation} for the linearized problem with tension $T = -R_3|e| < 0$ predicts, the unstable wave number with maximum exponential growth rate is $k^2 =\frac{R_3|e|}{2R_2}$. For practical purposes, this means that in order to resolve the wrinkling induced by compression, one requires a maximum mesh spacing of $\mathcal{O}(\sqrt{\frac{R_2}{R_3}})$. This is intuitively clear. Increasing the bending rigidity $R_2$ damps the fine-scale wrinkles more strongly. Hence the mesh size can be larger. Increasing $R_3$ increases the effect of compression, however, increasing the wave number of the fastest growing mode of (\ref{linearized membrane equation with negative tension and R2}), so the mesh size should decrease in kind. To resolve the wrinkling in figure \ref{membrane compression example figure} we use $n = 500$ modes in our Galerkin approximation.
\newline

The approximate equation (\ref{linearized membrane equation with negative tension and R2}) indicates that in the limit of $R_2 = 0$ one requires infinitely many modes to resolve the membrane under compression, though the effect of the fluid pressure loading was neglected. As discussed in section \ref{R2 > 0 section}, when $R_2 = 0$ the membrane equation becomes ill-posed, and so the simulation cannot be continued past the critical time where the tension becomes negative. When $R_2 > 0$ however, the simulation can continue. As shown in figure \ref{membrane compression example figure} and predicted by the linearization argument above, instead of becoming numerically unstable, the membrane instead develops fine wrinkles on its surface to reduce the elastic energy stored in compression. 
\newline

The wrinkling seen during stall is a symptom rather than a cause: it appears where the membrane has lost tension and gone slack due to the loss of aerodynamic load. The lift decreases at the moment of unloading when a vortex detaches, and not during the subsequent wrinkling. Bending rigidity sets the scale of the wrinkles that form, but does not affect the pressure loading much. The practical target for a flexible lifting surface is therefore to keep the membrane in tension throughout the maneuver, avoiding the slack, compressed states altogether; wrinkling is most useful here as a diagnostic that marks where and when that tension is lost. Where slack regions are unavoidable, bending rigidity governs the resulting deformation scale and limits flogging and fatigue, but it does not recover the lift lost when the vortex detached and the pressure loading decreased.
\newline

We observe that only \textit{heavy} membranes undergo compression wrinkling and not \textit{light} membranes, and there is an intuitive explanation. Indeed, consider a massless membrane ($R_1 =0$). For such a membrane, the tangential force on every membrane element must be zero. Since the fluid force only applies a force normal to each membrane element, the tension is constant across the membrane's length. Hence for compression wrinkling to occur, the tension must be negative on the whole membrane. At that instant, the physical length of the membrane must be less than its material length. However, since the two ends of the membrane are held fixed this is impossible. A massless membrane thus never wrinkles. Hence in the limit of zero density ($R_1 \rightarrow 0$), it becomes more and more difficult for the membrane to wrinkle.
\newline

Note also that compression-induced wrinkling is a transient event. The argument that compressed states do not occur for $R_1 = 0$ applies more generally when the inertia term $R_1\partial_{tt}\zeta$ is identically zero in the membrane equation \eqref{dimensionless membrane equation}, which can occur for $R_1 = 0$ or $\partial_{tt}\zeta \equiv 0$. Hence a wrinkled configuration cannot be a steady solution to the membrane equation \eqref{dimensionless membrane equation}. This is consistent with the supplementary movie \href{https://drive.google.com/file/d/1CGcG94oydSZJySB4c7xGAOxGPoc7pN9P/view?usp=drive_link}{``wrinkling\_movie.avi"} where the wrinkling forms around $t \approx 12.5$ and dissipates by $t\approx 15$, persisting only for $\approx2.5$ time units. In addition, these wrinkles also have an extremely small amplitude. As shown in figure \ref{membrane compression example figure}, the maximum amplitude of these wrinkles is $\approx 2\%$ of the body length.
\newline

Wrinkling can also play a functional role in biological lifting surfaces. In bat wing membranes, pre-stretched elastin fibers compress the surrounding tissue, producing wrinkles that unfold as the wing is stretched \cite{cheney2015wrinkle}. This allows substantial extension before the stiffer tissue bears the load. Although compression arises from internal elastic forces in that case and from aerodynamic unloading in the present simulations, both illustrate how wrinkling accommodates excess membrane length. The wrinkles observed here therefore mark a loss of aerodynamic loading, while those in bat wings also contribute to the membrane's ability to extend as loading is restored.

\subsection{Fixed-free membranes}
\label{fixed free membranes section}
The fixed-fixed results in the previous section show two ways in which a membrane can accommodate concentrated fluid loading. While the membrane remains in tension, it deforms globally and changes the attachment of the shed vortices; when it loses tension, bending rigidity regularizes the resulting compression through finite-scale wrinkling. Releasing the trailing edge changes the problem more fundamentally, because the tension is then forced to vanish at the free endpoint. The fixed-free membrane therefore possesses an intrinsic degeneracy in the boundary conditions when $R_2 = 0$ even in the absence of compression. For fixed-free membranes, the singular limit $R_2 \to 0$ is our main focus.
\newline

Many previous studies have considered fixed-free membranes (flags) with small but nonzero bending rigidity in both inviscid \cite{alben2022dynamics, shelley2011flapping} and viscous \cite{connell2007flapping} flows. However, none have considered a flag with zero bending rigidity ($R_2 = 0$). This is because with fixed-free boundary conditions the resulting partial differential equation is in a sense underdetermined. Indeed, as discussed in appendix \ref{solving the membrane equation appendix}, at the free end of the membrane, the boundary condition for 
\begin{equation}
    \partial_{tt}\zeta = R_3\partial_\alpha\big((|\partial_\alpha \zeta|-1)\hat{\boldsymbol{s}}\big) + \boldsymbol{f}
    \label{fixed free membrane equation}
\end{equation}

at the free end determined by the variational formulation is zero resultant force at the free end. However, since $|\partial_\alpha\zeta| = \lambda \not=0$, both coordinate equations for zero resultant force at the free end are linearly dependent when $R_2 = 0$, and equivalent to the single condition that the tension vanishes. More specifically,
\begin{equation}
    R_3(1-\lambda^{-1})\partial_\alpha\zeta  \big|_{\alpha = 1}= 0 \iff \lambda|_{\alpha = 1} = 1 \label{fixed free boundary condition}
\end{equation}

Since the ($x,y$) position of the leading edge is prescribed, a total of 3 boundary conditions are provided. However, the equation is second-order in $\alpha$. With two degrees of freedom $x(\alpha,t)$ and $y(\alpha,t)$ such an equation typically requires $2\times2 = 4$ boundary conditions in total. Hence it would seem that the solution to the equation is underdetermined.
\newline

Some studies overcome this issue by augmenting the seemingly underdetermined system with an additional boundary condition. In \cite{mavroyiakoumou2020large, mavroyiakoumou2022membrane} for example, the horizontal distance between the free and fixed ends is prescribed. The free end thus slides vertically on a frictionless pole. This boundary condition has a physical motivation; It is the variational boundary condition at the free end of the membrane equation linearized around small deflections in $x$. Indeed, this linearized equation is considered in \cite{manela2017hanging}. By considering $x \approx \alpha$, a single degree of freedom is removed, and the deflections are confined to the $y$-direction. In the present notation, the result is the following equation:

\begin{equation}
    R_1\partial_{t t} y = R_3\partial_\alpha\big((1 -\alpha)\partial_\alpha y) + \boldsymbol{f}.
    \label{simplified membrane equation}
\end{equation}
This equation is now second-order in space but only with a single degree of freedom, requiring only 2 boundary conditions. At the fixed end, $y$ is prescribed. At the free end, the variational boundary condition is $\alpha = 1$ or $x=1$, which is exactly the ``frictionless pole" boundary condition considered in \cite{mavroyiakoumou2020large, mavroyiakoumou2022membrane}.
\newline

In the present study, we show that the fixed-free membrane equation \eqref{fixed free membrane equation} is well-posed with only three boundary conditions: the prescribed leading edge position and the single variational condition $|\partial_\alpha \zeta|=1$ when $\alpha = 1$. The apparently missing boundary condition is not so much missing as it is replaced by the requirement that the solution have bounded stretching energy (i.e., the solution lies in $H^1_\alpha$). To enforce this constraint, we solve the weak formulation with a spectral Galerkin method, which automatically imposes the finite-energy constraint and the natural boundary condition simultaneously. The details are provided in appendix \ref{solving the membrane equation appendix}. As shown in appendix \ref{convergence in the limit of zero bending rigidity section}, this solution coincides with the $R_2 \rightarrow 0$ limit (with respect to the energy norm $C^1_tH^1_\alpha$) of solutions to the bending-regularized equation, for which all solution branches have finite energy and the standard boundary condition count applies.
\newline

In appendix \ref{missing boundary conditions appendix} we explain carefully how the bounded energy condition combined with a single boundary condition for each unknown selects the unique solution with some simple examples.

\subsubsection{Tip-snapping in the limit of zero bending rigidity}
To illustrate the effect of $R_2$ we fix the value of $R_1 = 0.25$ and $R_3 = 100$. For $R_2 \geq 10^{-2}$ the undeflected state is stable. Decreasing the bending rigidity below $10^{-2}$ destabilizes the undeflected state, and the membrane begins to flutter and whip around. Initially ($10^{-3} \lessapprox R_2 \lessapprox 10^{-2}$), this flutter is periodic \cite{alben2008flapping} (see the following section \ref{effect of R3 section} for one such example), but as $R_2\rightarrow 0$, only chaotic fluttering motions can be observed. 
\newline

Figure \ref{fig: snapping flag vorticity} shows snapshots of the vorticity fields induced by membrane flags with $R_1 = 0.25,\ R_3 = 100$, $R_2 = 10^{-2}, 10^{-4},10^{-6},10^{-8},10^{-10}, 0$ at time $t = 5$. The wakes are composed of streets of irregular vortex spirals that have similar spatial structures and temporal frequencies for $R_2 < 10^{-4}$, though different configurations due to the sensitivity of the dynamics to small perturbations in parameters (such as $R_2$). 
These peculiar wakes are due to high-frequency snapping of the membrane tips combined with the slower overall fluttering motions of the flags.
\begin{figure}[H]
    \centering
    \includegraphics[width=0.7\linewidth, trim = 1cm 2.5cm 1cm 1.5cm, clip]{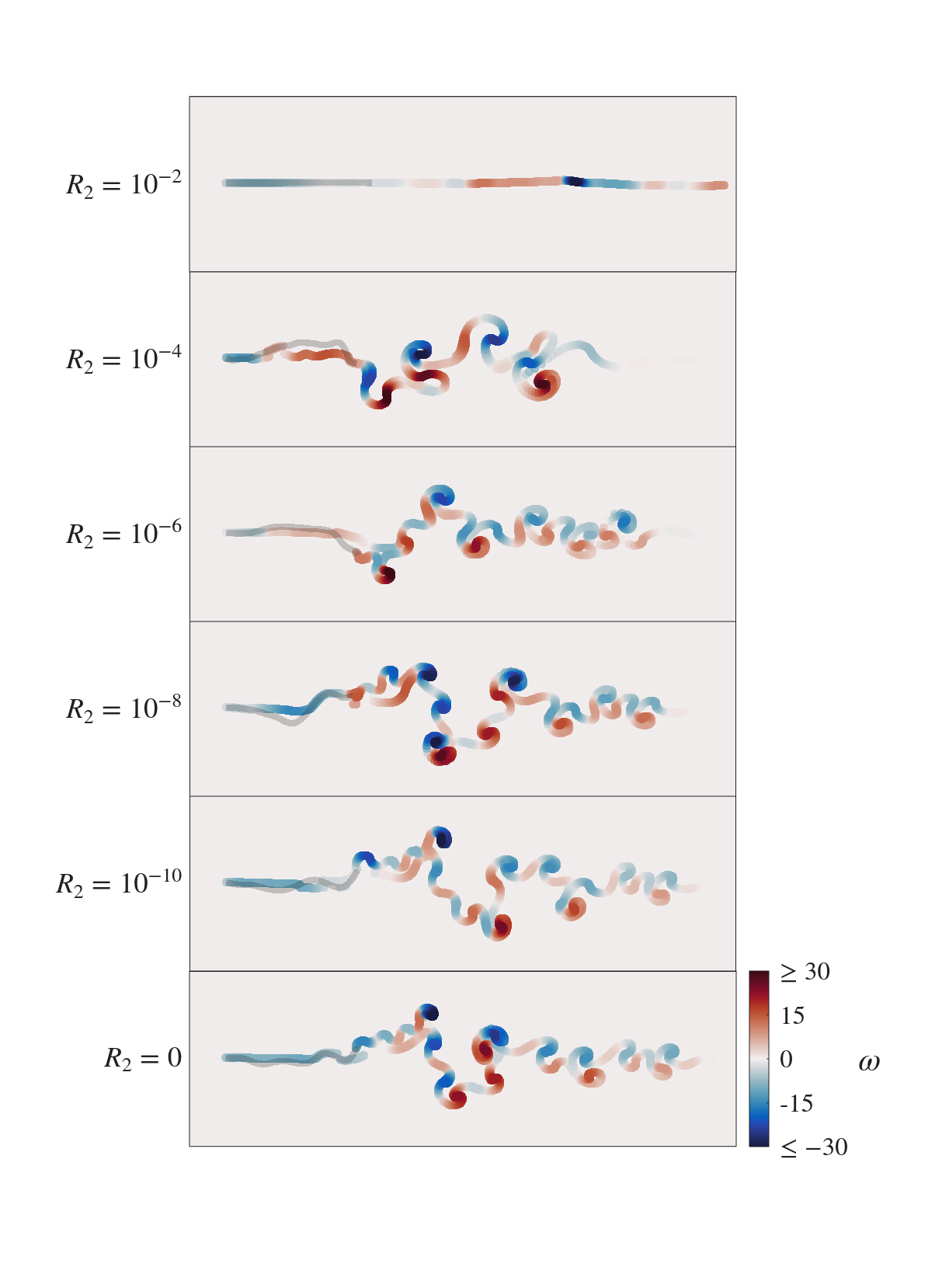}
    \caption{Vorticity fields for membrane flags with $R_1 = 0.25,\ R_3 = 100$, $R_2 = 10^{-2}, 10^{-4},10^{-6},10^{-8},10^{-10}, 0$ at time $t = 5$.}
    \label{fig: snapping flag vorticity}
\end{figure}
Qualitatively, flag-fluttering motions with $R_2 \leq 10^{-3}$ are similar everywhere except for a thin bending boundary layer near the free end, where the tip of the membrane flag flutters and snaps with a frequency that slowly increases as $R_2 \rightarrow 0$ before saturating at a critical value $f \approx 3.40$. This is shown in figure \ref{fig: snapping flag snapshots}. Figure \ref{fig: snapping flag snapshots} shows snapshots of the full membrane (left column), snapshots of the membrane tip (middle column) and the tip deflection amplitude against time (right column) for a membrane flag with $R_1 = 0.25,\ R_3 = 100$, $R_2 = 10^{-2}, 10^{-3},10^{-4},10^{-5},10^{-8}, 0$ until time $t = 10$. Each row corresponds to a different value of $R_2$. Starting from $R_2 = 10^{-4}$ the snapshots of the overall membrane positions (left column of figure \ref{fig: snapping flag snapshots}) agree extremely well qualitatively. However, comparing the snapshots of the tip (middle column of figure \ref{fig: snapping flag snapshots}) show that the membranes with different $R_2$ differ primarily due to the quasiperiodic snapping of the free end whose frequency increases very slowly with decreasing $R_2$, with approximately constant amplitude, as shown in the right column of figure \ref{fig: snapping flag snapshots}.

\begin{figure}[H]
    \centering
    \includegraphics[width=0.7\linewidth, trim = 0cm 0.5cm 0cm 0cm, clip]{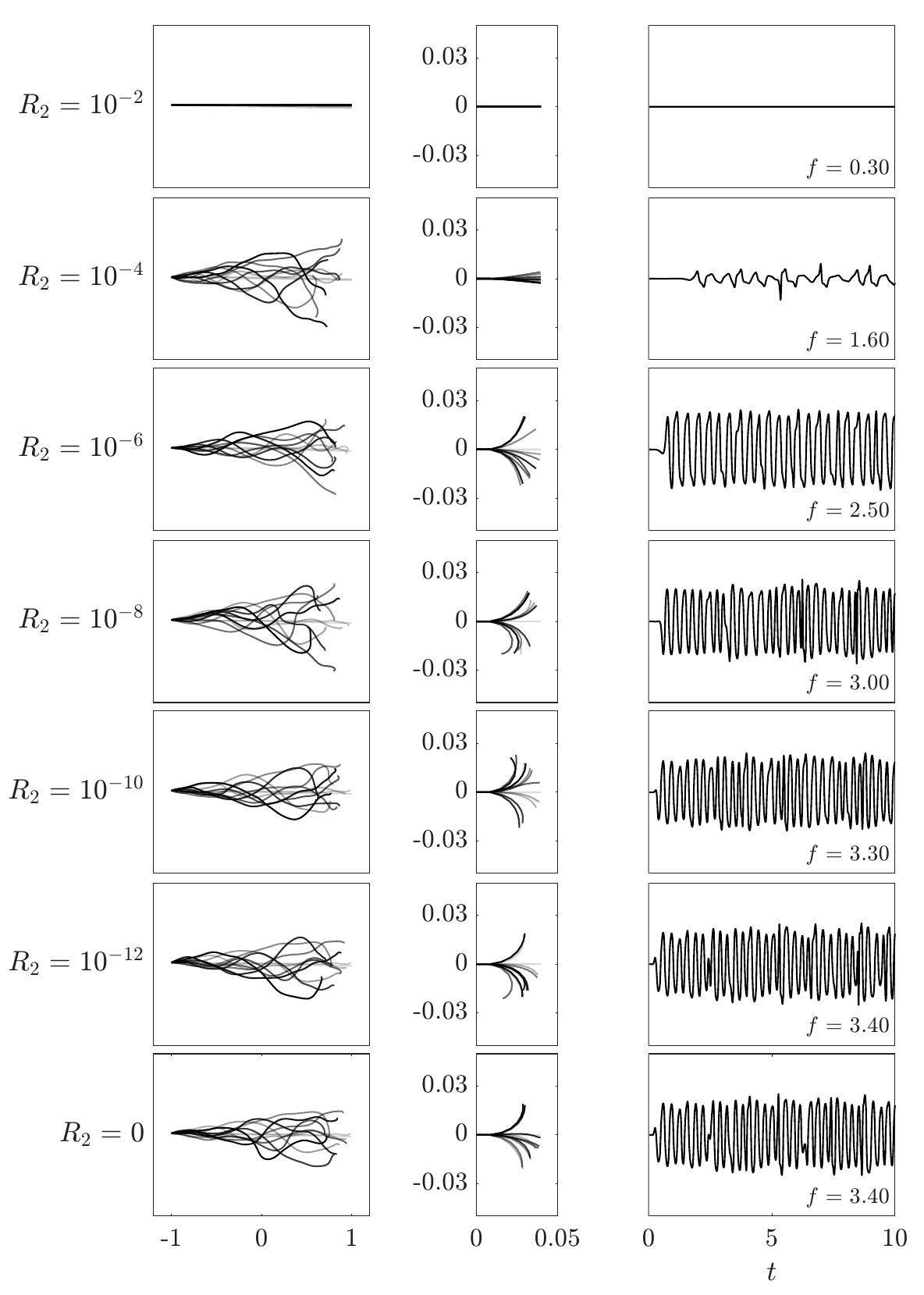}
    \caption{Snapshots of the full membrane (left column), snapshots of the membrane tip (middle column) and the tip deflection amplitude against time (right column) for a membrane flag with $R_1 = 0.25,\ R_3 = 100$, $R_2 = 10^{-2}, 10^{-3},10^{-4},10^{-5},10^{-8}, 0$ until time $t = 10$. Each row corresponds to a different value of $R_2$. }
    \label{fig: snapping flag snapshots}
\end{figure}

This tip-snapping is clearly shown in the supplementary movie \href{https://drive.google.com/file/d/1KGXbDMJiGA9wVuYbdKkdgYOSV8dSV5cg/view?usp=drive_link}{``flicking\_movie\_R2-0.avi"}. This movie considers a fixed-free membrane with $(R_1,R_2,R_3) = (0.25,0,10^2)$ in a uniform flow. Panels (a-c) of the movie show the vorticity field induced by the membrane flag, a close up of the membrane flapping, and the pressure jump along the membrane respectively. Panel (c) of the movie shows that the pressure jump obtains a global extremum at a critical point near the free end, before vanishing at exactly the free end point. This global extremum alternates in sign from positive to negative, causing the tip to whip back and forth and form small hooks as shown in the middle column of figure \ref{fig: snapping flag snapshots}.
\newline

The formation of a global pressure extremum near the membrane tip can be explained intuitively as follows. Let $R_2 = 0$. Since the fluid pressure jump acts normal to the membrane, it excites primarily transverse perturbations about a locally straight state. Linearizing the transverse dynamics about this locally straight region and ignoring external forces then gives a wave equation,
\begin{equation}
R_1\partial_{tt}y = R_3 e\,\partial_{\alpha\alpha}y,
\end{equation}
where $e=(\lambda-1)/\lambda>0$. The corresponding local wave speed is

$$
    c=\sqrt{\frac{R_3 e}{R_1}},
$$

and therefore vanishes as the tension approaches zero at the free edge when $R_2=0$ (see equation \eqref{fixed free boundary condition}). Consequently, transverse disturbances propagating along the membrane are increasingly slowed near the free end and cannot propagate beyond it, causing the associated pressure disturbance to concentrate near the tip and form a pressure-jump extremum. As discussed in section \ref{pressure jump equation as a conservation law section}, pressure-jump extrema coincide with $\gamma$ extrema that are swept downstream by the uniform flow and ejected into the fluid in the form of an alternating-sign vortex sheet, as shown in figure \ref{fig: snapping flag vorticity}.
\newline

It can be shown that this slowly-growing tip flapping frequency obeys an inverse-logarithmic scaling law in the bending rigidity $R_2$. Figure \ref{inverse logarithmic relation f-R2 figure} shows the maximum tip deflection frequency $f$ against $R_2$. The best-fit approximation of the form
\begin{equation}
    f \approx A + \frac{B}{\log R_2}
\end{equation}
is shown as a red dashed line, while the value corresponding to $R_2=0$ is shown as a red square. The best-fit approximation has $R^2=0.989$, indicating that the inverse-logarithmic dependence gives a remarkably good description of the numerical data, especially for $R_2 \geq 10^{-11}$ before the frequency becomes constant in $R_2$. This $\mathcal{O}(1/|\log R_2|)$ approach of the tip frequency to its zero-rigidity limit can be explained by examining the weakly singular branch of the linearized problem.
\begin{figure}[H]
    \centering
    \includegraphics[width=0.7\linewidth, trim = 0cm 0cm 0cm 0cm, clip]{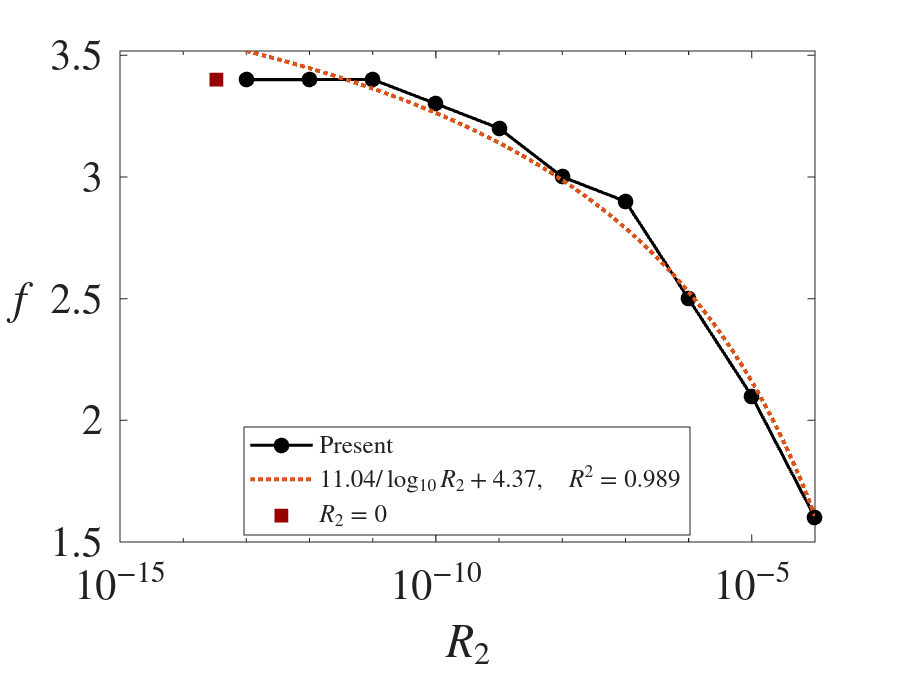}
    \caption{Maximum tip deflection frequency $f$ against $R_2$. The best-fit approximation of the form $f \approx A + \frac{B}{\log(R_2)}$ is shown as a red dashed line, and the value corresponding to $R_2 = 0$ is shown as a red square.}
    \label{inverse logarithmic relation f-R2 figure}
\end{figure}

By rescaling the constants $R_1$, $R_2$, and $R_3$, we may assume that $0\leq\alpha \leq 1$ so that the fixed endpoint lies at $\alpha = 0$ and the free endpoint lies at $\alpha = 1$, which simplifies the algebra without changing the scaling laws we will derive. 
As discussed in appendix \ref{missing boundary conditions appendix}, solutions to the linearized problem with $R_2=0$ are linear combinations of Bessel functions
\begin{equation}
    J_0(k\sqrt{1-\alpha}),
    \qquad
    Y_0(k\sqrt{1-\alpha}),
\end{equation}
with time-dependent coefficients, where $k$ is the corresponding wave number of the eigenmode after rescaling. The second branch has the endpoint asymptotic behavior
\begin{equation}
    Y_0(k\sqrt{1-\alpha})
    =
    C(k) + \frac{1}{\pi}\log(1-\alpha)
    + o(1),
    \qquad \alpha\rightarrow 1^{-},
\end{equation}
and is therefore logarithmically singular at the free endpoint, and hence is not in $H^1$.
\newline

As shown both numerically and via asymptotic analysis in \cite{manela2017hanging}, the introduction of nonzero bending rigidity $R_2$ produces a bending layer of thickness
\begin{equation}
    \delta = \mathcal{O}(R_2^p)
\end{equation}
for some $p>0$ at the free end of the membrane flag. Away from this layer, bending is asymptotically negligible and the outer solution is therefore governed to leading order by the zero-rigidity equation. It may consequently be written in the form
\begin{equation}
    \phi_{\mathrm{out}}(\alpha;R_2)
    =
    A(R_2)J_0(k\sqrt{1-\alpha})
    +
    B(R_2)Y_0(k\sqrt{1-\alpha}),
    \qquad
    0<\alpha<1-\delta .
\end{equation}
Note that we have suppressed the time dependence of the coefficients $A,B$. As discussed in appendix \ref{convergence in the limit of zero bending rigidity section}, the solutions with $R_2>0$ converge to the finite-energy $R_2=0$ solution as $R_2\rightarrow0$. It follows that
\begin{equation}
    B(R_2)\rightarrow0.
\end{equation}

At the edge of the bending layer, $1-\alpha\sim\delta$, and the singular branch therefore has magnitude
\begin{equation}
    Y_0(k\sqrt{\delta})
    =
    \mathcal{O}(\log\delta)
    =
    \mathcal{O}(\log R_2).
\end{equation}
regardless of the value of $p$. Thus, for the contribution of the singular outer branch to remain bounded while the outer solution converges to the regular zero-rigidity solution, its coefficient must satisfy
\begin{equation}
    B(R_2)|\log\delta|= \mathcal{O}(1),
\end{equation}
and hence
\begin{equation}
    B(R_2)
    =
    \mathcal{O}\left(\frac{1}{|\log\delta|}\right)
    =
    \mathcal{O}\left(\frac{1}{|\log R_2|}\right).
    \label{B inverse log scaling}
\end{equation}

It remains to connect this inverse-logarithmic perturbation of the spatial eigenmode to the observed oscillation frequency. To do so, normalize the regular component of the outer solution and write
\begin{equation}
    \phi_{\mathrm{out}}(\alpha;R_2)
    =
    J_0(k\sqrt{1-\alpha})
    +
    \beta(R_2)
    Y_0(k\sqrt{1-\alpha}),
    \qquad
    \beta(R_2)=\frac{B(R_2)}{A(R_2)}.
\end{equation}
Equation \eqref{B inverse log scaling} implies
\begin{equation}
    \beta(R_2)
    =
    \mathcal{O}\left(\frac{1}{|\log R_2|}\right).
\end{equation}
The wave number is selected by the boundary condition at the fixed end.  Writing the outer solution in normalized form as
\begin{equation}
    \phi(\alpha)
    =
    J_0\!\left(k\sqrt{1-\alpha}\right)
    +
    \beta(R_2)
    Y_0\!\left(k\sqrt{1-\alpha}\right),
\end{equation}
the fixed-end condition gives an eigenvalue equation for $k$.
\newline

Since the fixed end at $\alpha=0$ lies outside the bending-dominated free-end layer, the outer solution satisfies the prescribed displacement condition to leading order. Indeed, corrections arising from the finite bending rigidity are algebraically small in $R_2$ \cite{manela2017hanging} and can therefore be neglected relative to the inverse-logarithmic contribution. This gives 
\begin{equation}
J_0(k)+\beta Y_0(k)=0.
\label{fixed end eigenvalue condition}
\end{equation}
Here $\beta=\beta(R_2)$ is determined by the regularization of the logarithmic branch near the free end, while \eqref{fixed end eigenvalue condition} determines the corresponding wave number $k=k(R_2)$. 

At zero bending rigidity, $\beta=0$, so the limiting wave number $k_0$ is determined by
\begin{equation}
J_0(k_0)=0.
\end{equation}
Because $J_0'(k_0)\neq0$ and $Y_0(k_0)\neq0$, the implicit function theorem applied to \eqref{fixed end eigenvalue condition} gives
\begin{equation}
\left.\partial_k\beta\right|_{k=k_0}
=
-\frac{J_0'(k_0)}{Y_0(k_0)}
\neq0.
\end{equation}
Thus the relation is locally invertible, and, writing $k=k_0+\Delta k$,
\begin{equation}
\Delta k
=
-\frac{Y_0(k_0)}{J_0'(k_0)}\beta
+
\mathcal{O}(\beta^2).
\end{equation}
Consequently,
\begin{equation}
\beta(R_2)
=
\mathcal{O}\left(\frac{1}{|\log R_2|}\right)
\quad\Longrightarrow\quad
k(R_2)-k_0
=
\mathcal{O}\left(\frac{1}{|\log R_2|}\right).
\label{k inverse log scaling}
\end{equation}

Hence we see that the wave numbers of the eigenmodes converge at worst as an inverse-logarithmic scaling law. Finally, as shown in \cite{alben2022dynamics}, the linearized membrane-fluid system supplies a dispersion relation
\begin{equation}
    D(\omega,k)=0,
\end{equation}
or, locally (via the implicit function theorem),
\begin{equation}
    \omega=\Omega(k).
\end{equation}
Expanding about the zero-rigidity eigenmode gives
\begin{equation}
    \omega-\omega_0
    =
    \Omega'(k_0)(k-k_0)
    +
    \mathcal{O}\left((k-k_0)^2\right).
\end{equation}
Combining this relation with \eqref{k inverse log scaling} yields
\begin{equation}
    \omega(R_2)-\omega_0
    =
    \mathcal{O}\left(\frac{1}{|\log R_2|}\right),
\end{equation}
and hence
\begin{equation}
    f(R_2)-f_0
    =
    \mathcal{O}\left(\frac{1}{|\log R_2|}\right).
    \label{frequency inverse log scaling}
\end{equation}

To summarize, the inverse-logarithmic dependence observed in figure
\ref{inverse logarithmic relation f-R2 figure} arises naturally from the logarithmically singular branch of the linearized problem at the free end. Small bending rigidity regularizes this singular branch only within a bending boundary layer at the free end whose thickness vanishes algebraically with $R_2$. Consequently, the coefficient of this branch must vanish inverse-logarithmically as $R_2\rightarrow 0$, and each wave number and frequency mode inherits the same scaling through the linear eigenvalue and dispersion relations. This argument explains the slow convergence of the tip flapping frequency to its zero-rigidity value, in particular its approximate inverse-logarithmic scaling in $R_2$.
\newline

Let 
\begin{equation}
    k_{eff}(f) = \frac{\|\partial_\alpha \langle\zeta, e^{-2\pi i ft}\rangle_{L^2_t}\|_{L^2_\alpha}}{\| \langle\zeta, e^{-2\pi i ft}\rangle_{L^2_t}\|_{L^2_\alpha}}
\end{equation}
where $\langle\zeta, e^{-2\pi i ft}\rangle_{L^2_t} = \int_0^{\tau_f} \zeta e^{2\pi i ft}dt$, $\tau_f$ is the final time, and $f$ denotes the maximum frequency discussed above. If $\zeta = e^{ik\alpha  - i2\pi ft}$ is a traveling wave, then $k_{eff} = \mathcal{O}(k)$. Thus, for arbitrary body positions $\zeta$, $k_{\mathrm{eff}}$ provides an effective measure of the wave number associated with the dominant frequency $f$ as $R_2$ varies. One consequence of the scaling argument above is that $k_{\mathrm{eff}}$ converges to a finite value at a rate no slower than $\mathcal{O}\!\left(1/|\log R_2|\right)$ as $R_2 \rightarrow 0$. Indeed, figure \ref{inverse logarithmic relation k-R2 figure} appears to verify this hypothesis. We show the computed values of $k_{\mathrm{eff}}$ against $R_2$ and report the best-fit approximation of the form $A + B/\log R_2$. This approximation has an $R^2$ value of $0.952$, providing further numerical evidence for the inverse-logarithmic scaling of the tip-snapping frequency and wave number predicted by the linearization argument above. The fit is somewhat weaker than that for the frequency because the relation between $k_{\mathrm{eff}}$ and the underlying wave number depends on the finite observation window, whose number of contained oscillation periods varies with $R_2$ as the dominant frequency changes; in particular, $[0,\tau_f]$ need not span an integer number of periods.

\begin{figure}[H]
    \centering
    \includegraphics[width=0.7\linewidth, trim = 0cm 0cm 0cm 0cm, clip]{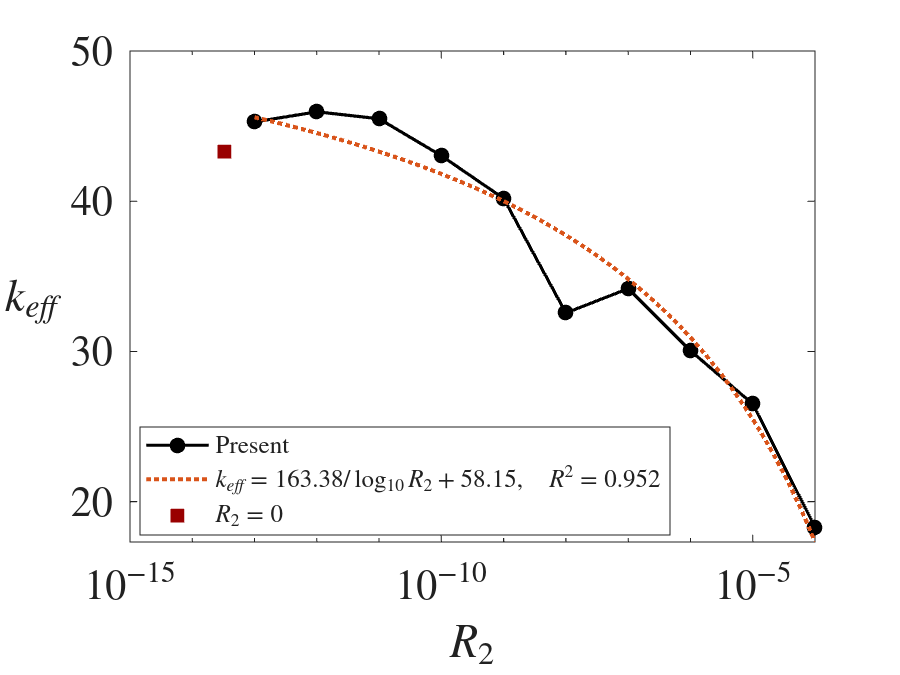}
    \caption{Wave number $k_{eff}$ corresponding to the maximum tip deflection frequency $f$ against $R_2$. The best-fit approximation of the form $f \approx A + \frac{B}{\log(R_2)}$ is shown as a red dashed line, and the value corresponding to $R_2 = 0$ is shown as a red square.}
    \label{inverse logarithmic relation k-R2 figure}
\end{figure}

\subsubsection{The effect of the stretching modulus $R_3$ on flag dynamics}
\label{effect of R3 section}
In this section we consider the effect of the stretching modulus on flag dynamics. To this end we fix $R_1 = 0.25$ and $R_2 = 5\times10^{-3}$ so that an extensible flag exhibits periodic flapping and vary $R_3$. Figure \ref{fig: R3 non effect} shows snapshots of the vorticity fields of membrane flags with $(R_1,R_2) = (0.25,5\times10^{-3})$ and $R_3 = 4,\ 2\times10^{1},\ 10^{2},\ 5\times10^{2}$ at $t = 20$. The membrane is shown as a bright green line.

\begin{figure}[H]
    \centering
    \includegraphics[width=0.8\linewidth, trim = 5cm 2cm 2cm 1cm, clip]{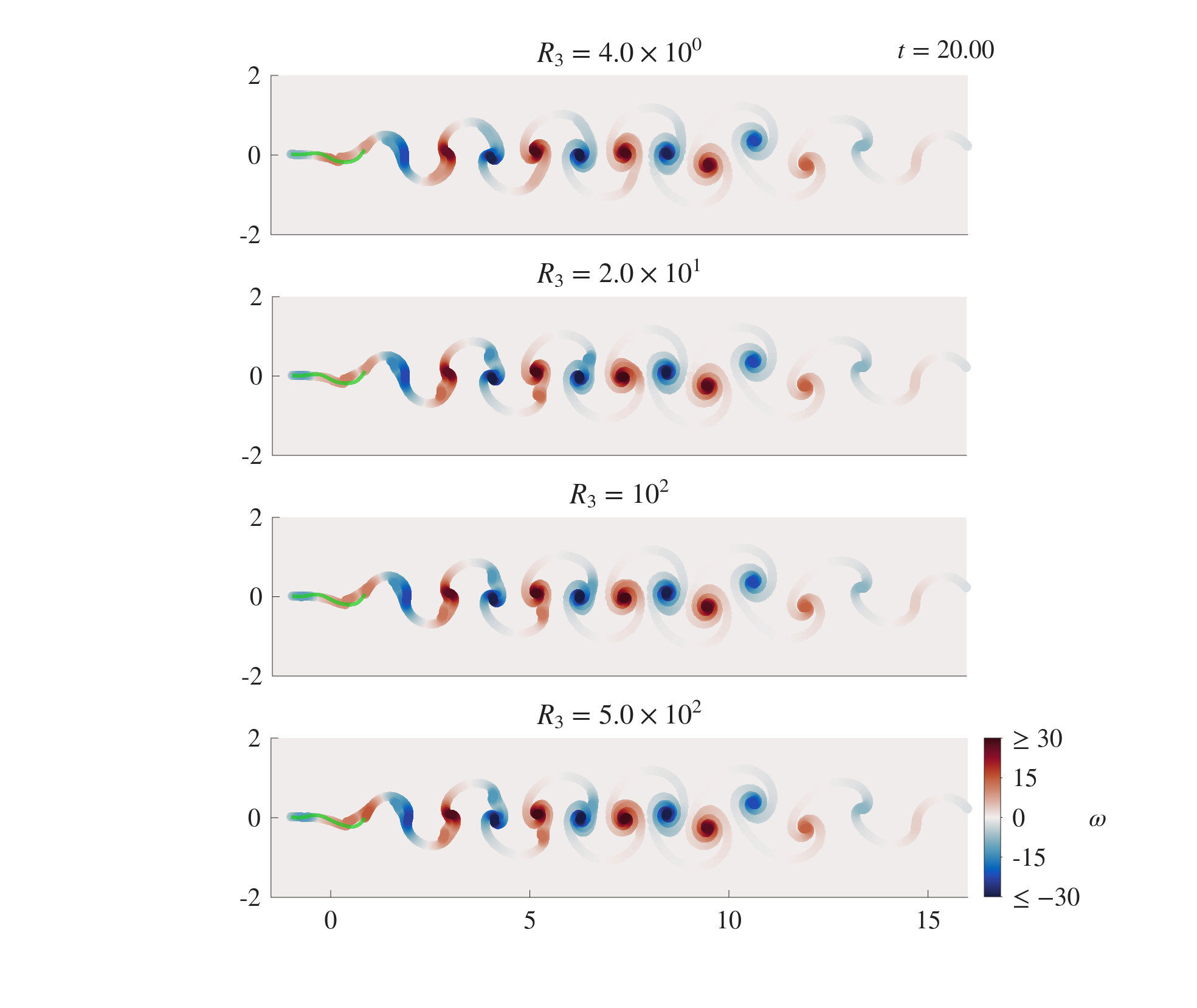}
    \caption{Vorticity fields of flags with $(R_1,R_2) = (0.25,5\times10^{-3})$ and $R_3 = 4,\ 2\times10^{1},\ 10^{2},\ 5\times10^{2}$ at $t = 20$. The membrane is shown as a bright green line for clarity.}
    \label{fig: R3 non effect}
\end{figure}
As can be seen in the figure, when $R_3 \geq 4$ the effect of varying $R_3$ on the periodic orbit is very mild. Indeed, for this range of $R_3$, the membrane flag behaves as though it is inextensible $(R_3 = \infty)$, and the membrane flag dynamics in this regime are virtually identical. Decreasing $R_3$ below $R_3 = 0.8$, however, introduces non-negligible effects caused by extensibility, destabilizing the periodic orbit.  
Figure \ref{fig: R3 effect} is analogous to  \ref{fig: R3 non effect} but considers smaller values of $R_3 \leq 4$, namely, $R_3 = 3.2\times10^{-2},\ 1.6\times10^{-1},\ 8\times10^{-1},\ 4 $ at $t = 30.01$. 
\begin{figure}[H]
    \centering
    \includegraphics[width=0.8\linewidth, trim = 5cm 2cm 2cm 1cm, clip]{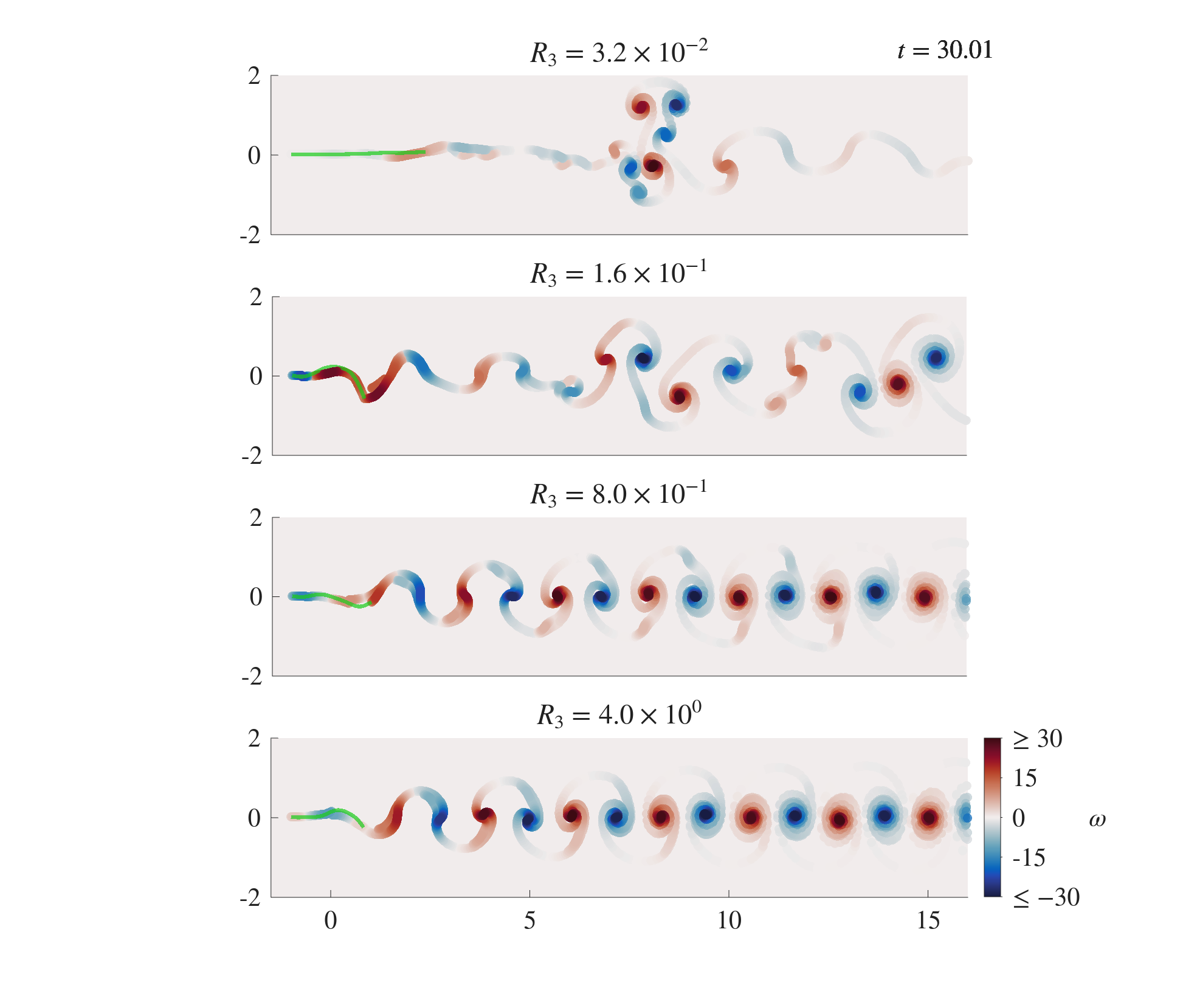}
    \caption{Vorticity fields of flags with $(R_1,R_2) = (0.25,5\times10^{-3})$ and $R_3 = 3.2\times10^{-2},\ 1.6\times10^{-1},\ 8\times10^{-1},\ 4 $ at $t = 30.01$. The membrane is shown as a bright green line for clarity.}
    \label{fig: R3 effect}
\end{figure}
In plain language, figure \ref{fig: R3 non effect} shows that for weakly extensible flags, the instability of the undeflected state is saturated by a stable periodic flapping motion. For $R_1 = 0.25$ and $R_3 \ge 1$ we observe such periodic orbits for $10^{-3}\lessapprox R_2 \lessapprox10^{-2}$. For the small bending rigidities $R_2$ considered here, this behavior persists for sufficiently large stretching modulus ($R_3\ge \mathcal{O}(1)$), where extension is negligible and the membrane instead accommodates the fluid loading primarily through bending and local rotation. As $R_3$ is decreased below $\mathcal{O}(1)$, however, figure \ref{fig: R3 effect} shows that the tension becomes too weak to oppose the fluid stresses and appreciable stretching develops (e.g. figure \ref{fig: R3 effect}, $R_3 = 3.2\times10^{-2}$). The motion initially remains close to the periodic flapping state (e.g. figure \ref{fig: R3 effect}, $R_3 = 8\times10^{-1}$), but the accumulating extension eventually destabilizes it. This can be seen in the supplementary movie showing the effect of increasing extensibility: \href{https://drive.google.com/file/d/1_z9MffwxBTtyBGtEEO2kDjkDcOf6o5Ns/view?usp=drive_link}{``R3\_effect\_movie.avi"}
\newline

The free-end boundary conditions require the tension and shear resultants to vanish at the free end. Likewise, the Kutta condition ensures that the pressure must vanish at the endpoints of the membrane flag \cite{loo2025falling}. The absence of fluid and tension forces at the trailing edge causes the free end to lag behind the motion of the neighboring membrane and it is pulled around as the rest of the flag deforms, producing curved hooks near the free end. Such hooks can be seen near $t$ = 36 and 48 in the supplementary movie, for the case with the smallest $R_3$. This deformation is self-reinforcing: once a hook forms, the pressure pushes it downstream, increasing both its curvature and its extension. Eventually the curvature becomes large enough that bending forces straighten the hook and rotate the membrane toward alignment with the flow. In this nearly aligned state the pressure loading becomes weak, and, in the absence of appreciable fluid shear, the membrane slowly retracts toward its undeflected configuration. 
\newline

The return toward the undeflected state does not terminate the motion, because the straight flag in uniform flow is itself unstable to flutter. The fluid therefore destabilizes the membrane again, producing another flapping and stretching excursion. This chaotic flapping regime ($R_3 \le \mathcal{O}(1)$) can thus be viewed physically as a repeated competition between fluid-driven destabilization and elastic restoration: the flow drives the flag away from the straight unstretched equilibrium, stretching and bending return it toward equilibrium, and the process repeats irregularly. The supplementary movie \href{https://drive.google.com/file/d/1_z9MffwxBTtyBGtEEO2kDjkDcOf6o5Ns/view?usp=drive_link}{``R3\_effect\_movie.avi"} makes this competition visually apparent.
\newline

This physical return mechanism also suggests a dynamical-systems interpretation of the transition from periodicity to chaos. For large $R_3$, trajectories leaving the unstable undeflected state are captured by the stable periodic flapping orbit. Once stretching destabilizes that orbit, the same outward fluttering motion is no longer saturated periodically. Instead, nonlinear elastic restoration returns trajectories toward the undeflected state, creating a global return from the unstable manifold of the rest state back toward its stable manifold. If these manifolds intersect transversely, their repeated intersections form a homoclinic tangle. Nearby trajectories are then repeatedly stretched apart during the fluttering excursion and folded back together by the elastic return, so successive passages near the undeflected state lead to increasingly different subsequent excursions; this provides the natural phase-space mechanism for the transition from periodic flapping to chaotic flutter. Indeed, a transverse intersection between the stable and unstable manifolds of a hyperbolic state generates a homoclinic tangle and, by the Smale–Birkhoff homoclinic theorem, implies the existence of horseshoe dynamics and therefore symbolic chaos \cite{katok1995introduction, smale1967differentiable}.

\section{Conclusion}
\label{conclusion section}

We have developed a two-dimensional inviscid model for extensible membranes with continuous vortex shedding from both edges. Together with a weak spectral Galerkin discretization of the membrane equation, the model permits fixed-fixed and fixed-free membranes to be studied across broad elastic parameter ranges, including the singular limit \(R_2=0\). The pressure-jump conservation law provides a common physical interpretation by describing the pressure jump as a flux for the bound vorticity.
\newline

For fixed-fixed membranes, compliance delays leading-edge-vortex detachment and enhances lift primarily during start-up; after periodic shedding develops, its principal effect is to amplify the force oscillations. Vortex detachment can also unload heavy membranes into compression, producing finite-scale wrinkles when \(R_2>0\), while the positive-rigidity membrane equation remains locally well-posed even under negative tension. For fixed-free membranes, the weak formulation selects the bounded-energy solution of the formally underdetermined zero-rigidity problem. The vanishing tension at the free edge produces rapid tip-snapping, whose dominant frequency and effective wave number approach their \(R_2=0\) limits with an inverse-logarithmic correction. Sufficient extensibility also destabilizes periodic flutter, and the resulting fluid-driven departure and elastic return suggest a homoclinic-tangle mechanism for the transition to chaos.
\newline

Taken together, these results show that even a relatively simple inviscid membrane model supports a wide range of dynamics whose mechanisms can be identified and explained mathematically. Membrane geometry determines how vortices remain attached and where aerodynamic loading is concentrated, while bending rigidity and extensibility determine how the body responds when that loading is gained or lost. The present model is restricted to two-dimensional inviscid flow with vortex sheet separation prescribed at the membrane edges, and therefore omits viscous boundary layers, spanwise transport, and three-dimensional deformation. Comparisons with viscous simulations and experiments will be needed to determine the quantitative range of validity of the results, but the mechanisms identified here provide concrete hypotheses for what may occur in viscous models.
\newline

\vspace{2cm}

\appendix
\section*{Appendix preliminaries}
Appendices A–F develop the constitutive, numerical, and analytical properties of the membrane equation, while appendix G gives a self-contained derivation of the pressure-jump conservation law used to couple the membrane to the fluid.
\newline 

Throughout the appendices, let \(I=(-1,1)\) denote the material domain of the membrane. Let $L^2(I)$ denote the space of square integrable functions with domain $I$. For each nonnegative integer \(k\), the Sobolev space $H^k = H^k(I)$ consists of functions \(v\in L^2(I)\) whose weak derivatives \(\partial_\alpha^jv\) belong to \(L^2(I)\) for \(1\leq j\leq k\), with norm

$$
\|v\|_{H^k}^2
=
\sum_{j=0}^k
\|\partial_\alpha^jv\|_{L^2}^2.
$$

These definitions are applied componentwise to vector-valued functions, and we suppress the domain \(I\) when no ambiguity is possible. For a Banach space \(X\), we write

$$
C_t^rX=C^r([0,T];X),
\qquad
L_t^pX=L^p(0,T;X).
$$

Thus, for example, \(C_t^0H_\alpha^2=C([0,T];H^2(I))\). All endpoint conditions are understood in the trace sense. Since \(H^1(I)\hookrightarrow C^0(\overline I)\) and \(H^2(I)\hookrightarrow C^1(\overline I)\), the endpoint values of \(v\in H^1(I)\), as well as the endpoint values and pointwise immersion condition for \(v\in H^2(I)\) are well-defined.

\section{Deriving the stretching and bending energy}
\label{bending and stretching energy derivation appendix}
In this section we derive the form of the stored stretching and bending energies for an elastic sheet. Hence, this determines the constitutive relations satisfied by the elastic sheet. In our model, the elastic sheet is a pencil of Hookean fibers with uniform elastic coefficient $E$, parallel to the normal and in-plane direction. This is essentially the classical Cosserat rod theory and the formulae we derive are equivalent to what is described in \cite[98-105]{antman2005nonlinear}.
\newline

Recall that the thickness of the body is given by $h \ll 1$, with out-of-plane width $W$ and material arc length $\alpha \in [-L,L]$ so that the sheet has initial length $2L$. Let $\zeta$ denote the midline of the elastic sheet. At each instant $t$ the fiber located at thickness $\eta \in[-h/2,h/2]$ is parameterized by 
\begin{equation}
    \boldsymbol{p}(\alpha, t) = \zeta(\alpha, t) + \eta\,\hat{\boldsymbol{n}}(\alpha, t),\ \alpha \in [-L,L]
\end{equation}
$\eta \in[-h/2,h/2]$ is the normal coordinate of the fiber relative to the midline. The strain of this fiber is 
\begin{align}
    |\partial_\alpha\boldsymbol{p}| - 1 &= | \big(|\partial_\alpha\zeta|- \eta\kappa_0\big)\hat{\boldsymbol{s}}| -1 \\
    &=||\partial_\alpha\zeta|- \eta\kappa_0| - 1 \\
    &= |\partial_\alpha\zeta| -1- \eta\kappa_0
\end{align}
where the last equality follows from the assumption that $h \ll 1$ so that $|\partial_\alpha\zeta| \gg \eta\kappa_0$. Note that $\kappa_0 =\partial_\alpha\theta$ is not the curvature of the midline $\kappa = \kappa_0/|\partial_\alpha\zeta|$. 
\newline

Let $\lambda_\eta =|\partial_\alpha\boldsymbol{p}|$ and $\lambda:= \lambda_0=|\partial_\alpha\zeta|$. Then $\lambda_\eta =\lambda - \eta\kappa_0$. To proceed, we must assume a constitutive relation for each fiber. We assume they are Hookean: the stress force is a linear function of the strain $\lambda_\eta -1$ given by
\begin{equation}
    E(\lambda_\eta -1)
\end{equation}
where $E$ is the Young's modulus.

Hence, the energy $\mathcal{U}_\eta$ in each fiber is such that $\partial_{\lambda_\eta}\mathcal{U}_\eta = E(\lambda_\eta -1)$. Hence we have:

\begin{equation}
    \mathcal{U}_\eta = \int_1^{\lambda_\eta} E(\lambda_\eta -1) =\frac{1}{2}E(\lambda_\eta-1)^2
\end{equation}
Since the total energy in the sheet is the sum of the energies in each of its component fibers, by integrating in both the normal direction $\eta$ and the in-plane direction, the total stored energy is given by 

\begin{align}
    \mathcal{U} &=\int_0^W\int_{-h/2}^{h/2} \mathcal{U}_\eta d\eta\, dz \\
    &=\int_0^W\int_{-h/2}^{h/2}  \frac{1}{2}E(\lambda_\eta-1)^2d\eta\, dz \\
    &=\int_0^W\int_{-h/2}^{h/2}  \frac{1}{2}E(\lambda - \eta\kappa_0-1)^2d\eta\, dz \\
    &=\frac{EW}{2}\int_{-h/2}^{h/2}\big(  (\lambda -1)^2  - 2(\lambda-1)\eta\kappa_0 + \eta^2\kappa_0^2 \big)d\eta\\
    &=\frac{EhW}{2}(\lambda-1)^2+ \frac{Eh^3W}{24}\kappa_0^2 \\
    &=: \mathcal{U}_{stretch} + \mathcal{U}_{bend} 
\end{align}

where \begin{equation}
    \mathcal{U}_{stretch} = \frac{EhW}{2}(\lambda-1)^2,\ \mathcal{U}_{bend}= \frac{1}{2}B\kappa_0^2
\end{equation} with $B:= \frac{Eh^3W}{12}$.
\newline

In summary, assuming the elastic sheet is a pencil of parallel Hookean fibers uniquely determines its constitutive relations. It is worth noting that this set of equations is different from similar variational formulations of the elastic equations \cite{tadjbakhsh1966variational} that use the physical curvature $\kappa = \kappa_0/\lambda$ instead of the material curvature $\kappa_0$ in the bending energy. That is a phenomenological constitutive relation that is different from the one considered in this study, but for small stretch ($\lambda \approx 1$) they are equivalent.
\section{Deriving the membrane equation}
\label{deriving membrane equation appendix}

In this section of the appendix we explain the details in deriving the weak form of the membrane equation by taking the first variation of the energy and applying Hamilton's principle.
\newline

Recall that the Lagrangian is given by
\begin{align}
\mathcal{L}&= \int_{-1}^1 \frac{1}{2}R_1 |\partial_{t} \zeta|^2d\alpha - \bigg(\int_{-1}^1\frac{1}{2}R_3(\lambda-1)^2d\alpha + \int_{-1}^{1}\frac{1}{2}R_2|\kappa_0|^2d\alpha\bigg)
\\
    &=: \mathcal{K} - \bigg(\mathcal{U}_{stretch} + \mathcal{U}_{bend}\bigg).
\end{align}
Hence the extended Hamilton principle \cite[pp.~268--272]{meirovitch2001fundamentals} therefore gives
\begin{equation}
\delta\int_0^{\tau_f}
\mathcal{L}\,dt
+
\int_0^{\tau_f}\delta\mathcal{W}_{ext}\,dt = 
\int_0^{\tau_f}
\delta\left(
\mathcal{K}
-\mathcal{U}_{stretch}
-\mathcal{U}_{bend}
+ \mathcal{W}_{ext}\right)\,dt
=: \int_0^{\tau_f}\delta\mathcal{S}\,dt =0
\label{hamilton's principle}
\end{equation}
where $\tau_f$ is the final time and $\mathcal{S}$ is the action.
\newline

Since equation \eqref{hamilton's principle} must hold for any final time $\tau_f$, the equations of motion are given by $\delta \mathcal{S} = 0$, where $\delta\mathcal{S}$ denotes the first variation of the action. We now discuss how to compute this quantity, with the assumption that our virtual displacements $\delta\zeta$ are smooth and compactly supported in the open rectangle interval $(-1,1)\times(0,\tau_f)$. First, note the following identities: 
\begin{align}
    \lambda &= |\partial_\alpha\zeta| \\
    \lambda^2 &= |\partial_\alpha\zeta|^2 \\
    2(\delta\lambda)\lambda &= 2\partial_\alpha\zeta\cdot\partial_\alpha\delta\zeta \\ 
    \delta\lambda &= \frac{\partial_\alpha\zeta\cdot\partial_\alpha\delta\zeta}{\lambda} = \hat{\boldsymbol{s}}\cdot\partial_\alpha\delta\zeta.
\end{align}
The final equation says that the change in stretching is equal to the tangential component of the material derivative of the displacement perturbation. Also note:
\begin{align}
    e^{i\theta} &= \partial_s\zeta = \lambda^{-1}\partial_\alpha\zeta \\
    ie^{i\theta}\delta\theta &= \frac{\delta(\lambda^{-1})}{\lambda}\partial_s\zeta + \lambda^{-1}\partial_\alpha\delta\zeta \\
    \hat{\boldsymbol{n}}\delta\theta &= \frac{\delta(\lambda^{-1})}{\lambda}\hat{\boldsymbol{s}} + \lambda^{-1}\partial_\alpha\delta\zeta \\ 
    \delta\theta &= \hat{\boldsymbol{n}} \cdot\lambda^{-1}\partial_\alpha\delta\zeta =  \hat{\boldsymbol{n}} \cdot\partial_s\delta\zeta
\end{align}
The final equation says that the change in tangent angle is equal to the normal component of the arc length derivative of the displacement perturbation.
\newline 

Using these two identities, obtaining the variation $\delta\mathcal{S}$ is immediate. Indeed, working term by term: 
\begin{align}
    \delta\mathcal{K} &= \int_0^{\tau_f}\int_{-1}^1R_1\partial_t\zeta\cdot\partial_t\delta\zeta\,d\alpha\, dt \\ 
    &= \int_0^{\tau_f}\int_{-1}^1-R_1\partial_{tt}\zeta\cdot\delta\zeta\,d\alpha\, dt
\end{align}

\begin{align}
    \delta\mathcal{U}_{stretch} &= \int_0^{\tau_f}\int_{-1}^1R_3 (\lambda-1)\delta\lambda \,d\alpha\, dt \\
    &= \int_0^{\tau_f}\int_{-1}^1 R_3 (\lambda-1)\hat{\boldsymbol{s}}\cdot\partial_\alpha\delta\zeta \,d\alpha\, dt \\ 
    &=\int_0^{\tau_f}\int_{-1}^1-R_3 \partial_\alpha\big(R_3 (\lambda-1)\big)\hat{\boldsymbol{s}}\big)\cdot\delta\zeta \,d\alpha\, dt,
\end{align}

where the last equality follows from integration by parts in $\alpha$. The compact support assumption on $\delta \zeta$ ensures that the boundary terms vanish. 
Likewise,
\begin{align}
    \delta\mathcal{U}_{bend} &= \int_0^{\tau_f}\int_{-1}^1\frac{1}{2}R_2 2\kappa_0\delta\kappa_0 \,d\alpha\, dt \\
    &= \int_0^{\tau_f}\int_{-1}^1R_2 \kappa_0\partial_\alpha\delta\theta \,d\alpha\, dt \\
    &= \int_0^{\tau_f}\int_{-1}^1R_2 \kappa_0\partial_\alpha(\hat{\boldsymbol{n}} \cdot\lambda^{-1}\partial_\alpha\delta\zeta) \,d\alpha\, dt \\
    &=\int_0^{\tau_f}\int_{-1}^1R_2\partial_\alpha( \partial_s\kappa_0\hat{\boldsymbol{n}}) \cdot\delta\zeta \,d\alpha\, dt,
\end{align}
where the last equality follows from integration by parts in $\alpha$ twice and the  boundary condition that $\kappa_0(\pm1,t) = 0$. Finally, we note that 

\begin{align}
    \delta\mathcal{W}_{ext} = \int_0^{\tau_f}\int_{-1}^{1}\boldsymbol{f}\cdot\delta\zeta\,d\alpha\,dt.
\end{align}
Hence by summing all these terms and applying Hamilton's principle we obtain

\begin{equation}
    \delta\mathcal{S} =\int_0^{\tau_f}\int_{-1}^1\bigg(-R_1\partial_{tt}\zeta +\partial_\alpha\big(\big(R_3 (\lambda-1)\hat{\boldsymbol{s}}\big) - R_2\partial_\alpha( \partial_s\kappa_0\hat{\boldsymbol{n}}) +\boldsymbol{f} \bigg)\cdot\delta\zeta\,d\alpha\, dt = 0,
    \label{weak form of the membrane equation appendix}
\end{equation}
 which is exactly \eqref{weak form membrane equation}.

\section{Solving the membrane equation}
\label{solving the membrane equation appendix}
In this section of the appendix we discuss how the weak form of the membrane equation with linearized bending rigidity \begin{equation}
    \delta\mathcal{S} =\int_0^{\tau_f}\int_{-1}^1\bigg(-R_1\partial_{tt}\zeta +R_3\partial_\alpha\big( (\lambda-1)\hat{\boldsymbol{s}}\big) - R_2\partial_\alpha^4\zeta +\boldsymbol{f} \bigg)\cdot\delta\zeta\,d\alpha\, dt = 0
    \label{linearized weak form membrane equation appendix}
\end{equation}

can be solved using the spectral Galerkin method with numerical integration. A complete description of such methods can be found in \cite{canuto2006spectral}. For a more general introduction to spectral methods similar to the one discussed here, we refer the reader to \cite{trefethen2000spectral}. In this section we will make use of the Einstein summation convention and sum over repeated indices.
\newline

It is convenient to view the virtual displacement $\delta\zeta(\alpha, t)$ as a complex-valued, smooth compactly supported function on the rectangle $[-1,1]\times(0,\tau_f)$. Hence, we may replace the real dot product of vectors with complex multiplication. Moreover, as the integral equation \eqref{weak form membrane equation} must hold for any final time $\tau_f$, the integrand itself must vanish, and we may write the weak form of the membrane equation as 

\begin{equation}
\int_{-1}^1\bigg(-R_1\partial_{tt}\zeta +\partial_\alpha\big(R_3 (\lambda-1)\hat{\boldsymbol{s}}\big) - R_2\partial_\alpha^4\zeta +\boldsymbol{f} \bigg)\delta\zeta\,d\alpha\ = 0
    \label{weak form of the membrane equation appendix}
\end{equation}

for all complex-valued virtual displacements $\delta\zeta$. Then, using integration by parts it is easy to write \eqref{weak form membrane equation} as
\begin{multline}
\int_{-1}^1-R_1
\partial_{tt}\zeta\, \delta\zeta 
-R_3\big((1-\lambda^{-1})\partial_\alpha\zeta\big)\partial_\alpha(\delta\zeta) + R_2\partial_\alpha^2\zeta\, \partial_\alpha^2(\delta\zeta) 
+\boldsymbol{f}\, \delta\zeta \,d\alpha \\ + \bigg[R_3\big((1-\lambda^{-1})\partial_\alpha\zeta +R_2\partial_\alpha^3\zeta\big)\,\delta\zeta + R_2\partial_\alpha^2\zeta\,\delta(\partial_\alpha\zeta)\bigg]^{\alpha=1}_{\alpha=-1}= 0.
    \label{weak form membrane equation in C2H2 equation}
\end{multline}
The variational boundary conditions are exactly those that force the boundary terms in the square brackets to vanish.
Namely,
\begin{align}
    R_3(1-\lambda^{-1})\partial_\alpha\zeta +R_2\partial_\alpha^3\zeta = 0 &\text{ or } \delta\zeta =0, \\ 
    &\text{and} \\
    R_2\partial_\alpha^2\zeta = 0 &\text{ or }  \delta(\partial_\alpha\zeta) = 0.
\end{align}
From a physical point of view, $\delta\zeta$ is the displacement in position and $\delta(\partial_\alpha\zeta)$ is the displacement in the tangent vector. Hence $\delta\zeta =0$ and $\delta(\partial_\alpha\zeta) = 0$ correspond to Dirichlet boundary conditions for $\zeta$ and $\partial_\alpha \zeta$ respectively. $R_3(1-\lambda^{-1})\partial_\alpha\zeta$ can be thought of as the tension force, 
conjugate to displacement in the energy, and so is the force applied to the boundary point. Likewise,  $R_2\partial_\alpha^2\zeta$ is conjugate to the material tangent vector $\partial_\alpha\zeta$ in the energy and can be thought of as a generalized moment. Asking that these terms vanish corresponds to zero force and moment boundary conditions respectively. For the fixed-fixed membranes we allow free rotation at the ends, so the boundary conditions are those of a pinned beam
\begin{align}
    \delta\zeta\big|_{\alpha = \pm1} = 0,\ \partial_\alpha^2\zeta\big|_{\alpha =\pm1} = 0.
\end{align}
For the fixed-free membranes we apply the same boundary conditions at the leading edge and free boundary conditions at the trailing edge: 
\begin{align}
    \delta\zeta\big|_{\alpha = -1} = 0,&\ \partial_\alpha^2\zeta\big|_{\alpha =-1} = 0. \\ &\text{and} \\
    \big(R_3(1-\lambda^{-1})\partial_\alpha\zeta +R_2\partial_\alpha^3\zeta\big) \big|_{\alpha = 1}= 0,&\ R_2\partial_\alpha^2\zeta\big|_{\alpha =1} = 0.
\end{align}

Hence the weak form of the membrane equation can be written as:
\begin{equation}
\int_{-1}^1-R_1
\partial_{tt}\zeta\, \delta\zeta 
-R_3\big((1-\lambda^{-1})\partial_\alpha\zeta\big)\, \partial_\alpha(\delta\zeta) + R_2\partial_\alpha^2\zeta\, \partial_\alpha^2(\delta\zeta) 
+\boldsymbol{f}\, \delta\zeta \,d\alpha
\label{weak form membrane equation in C2H2 equation}
\end{equation}

for all $0<t <\tau_f$. The highest derivatives that appear in this integral are second-order in space and time. Hence we may seek a weak solution to it in the Hilbert space $C^2_tH^2_\alpha$, where $H^2_\alpha$ denotes the Sobolev space of twice-differentiable functions with square-integrable derivatives in $\alpha$, and $C^2_t$ denotes the space of twice-differentiable functions with continuous derivatives in $t$.
\newline 

To approximate this weak solution, we require a numerical discretization of \eqref{weak form membrane equation in C2H2 equation} on the interval $\alpha \in [-1,1]$. To this end, we let $$\alpha_j = -\cos(\pi j/n),\text{ } j = 0,1,\ldots\,\ ,n$$
denote the nodes on the Chebyshev-Lobatto grid and let $w_j ,\text{ } j = 0,1,\ldots\,\ ,n$ denote the corresponding Clenshaw-Curtis quadrature weights \cite{trefethen2008gauss,trefethen2000spectral}. The Clenshaw-Curtis weights are the unique quadrature weights on the Chebyshev-Lobatto grid that integrate polynomials of degree $\leq n$ exactly. Similarly, let $D = [D_{ij}]_{0 \leq i,j \leq n}$ denote the Chebyshev spectral first derivative matrix that is exact for polynomials of degree $n$ on the Chebyshev-Lobatto grid \cite{trefethen2000spectral}. Stated precisely, $$D_{ij}p(\alpha_j) =\partial_\alpha p(\alpha_i)$$
 for all $0\leq i\leq n$. where we note that we used the summation convention. Both exact formulae and algorithms in MATLAB to compute both the spectral quadrature weights $w$ and the differentiation matrix $D$ are provided in \cite{trefethen2000spectral}.
\newline

Let $$\ell_0(\alpha),\ell_1(\alpha),\ldots\ , \ell_n(\alpha)$$ denote the Lagrange interpolating polynomials through $\alpha_0,\ldots\ ,\alpha_n$ defined by the property that $\ell_j(\alpha_k) =\delta_j^k$ where $\delta^k_j$ denotes the Dirac delta function. As a consequence of this definition, if $p(x) = p_j\ell_j(x)$ then we have $p(\alpha_k) = p_j\delta_j^k = p_k$. That is, the coefficients of a polynomial as a Lagrange series are equal to its values on the underlying grid. In addition, the entries of the Chebyshev spectral differentiation matrix can be given in terms of the derivatives of the Lagrange polynomials. This follows immediately from the definition of $D$: $D_{ij} = D_{ik}\delta^k_{j}=D_{ik}\ell_j(\alpha_k) =\partial_\alpha \ell_j(\alpha_i)$. Hence, $$D_{ij} = \partial_\alpha\ell_j(\alpha_i)$$

We define our trial and test functions to lie in the finite-dimensional subspace $V =\text{span}\{\ell_0,\ldots\ ,\ell_n\} \subset H^2_\alpha$ and seek a solution to equation \eqref{weak form membrane equation in C2H2 equation} of the form $\underline{\zeta}(\alpha, t) = \zeta_k(t)\ell_k(\alpha)\in C^2_tH^2_\alpha$ where $\underline{\zeta}(\cdot,t) \in V$ for all $t$. Hence we set $\delta\zeta = \ell_j,\ j = 0,\, \ldots n$. Let $\underline{\boldsymbol{f}} = [\boldsymbol{f}(\alpha_0),\,\ldots,\boldsymbol{f}(\alpha_n)]^T$ and $\underline{\boldsymbol{\zeta}} =[\underline{\zeta}(\alpha_0,t),\ldots, \underline{\zeta}(\alpha_n,t)] =[\zeta_0,\ldots\zeta_n]^T $. Substituting $\zeta = \underline{\zeta}$ into \eqref{weak form membrane equation in C2H2 equation} and approximating the integral by summing against the Clenshaw-Curtis weights $w$ we obtain:

\begin{align}
&\int_{-1}^1-R_1
\partial_{tt}\zeta\, \delta\zeta 
-R_3\big((1-\lambda^{-1})\partial_\alpha\zeta\big)\, \partial_\alpha(\delta\zeta) + R_2\partial_\alpha^2\zeta\, \partial_\alpha^2(\delta\zeta) 
+\boldsymbol{f}\, \delta\zeta \,d\alpha = 0\\
\implies &\textrm{diag}(w)\partial_{tt}\underline{\boldsymbol{\zeta}}  
-R_3D^T\textrm{diag}\big((1-\lambda^{-1})w)D\underline{\boldsymbol{\zeta}}+ R_2 (D^2)^T\textrm{diag}\big(w)D^2\underline{\boldsymbol{\zeta}}  + \textrm{diag}(w)\underline{\boldsymbol{f}}  = 0
\label{discretized weak membrane equation}
\end{align}
The result is a column vector of $n+1$ equations for the unknowns $\zeta_0,\ldots,\zeta_n$. The $j$th row corresponds to the weak equation evaluated with test function $\ell_j$. If either of the boundary values $\zeta(-1) = \zeta(\alpha_0)$ or $\zeta(+1) = \zeta(\alpha_n)$ are known, then the space of test functions must vanish at the corresponding node. Hence, the corresponding entry  of \eqref{discretized weak membrane equation} need be replaced with the appropriate boundary equation $\zeta_{0} = \zeta(-1)$ or $\zeta_{n} = \zeta(1)$.  All other boundary conditions are imposed implicitly through the derivation of the weak form of the equations.
\newline

By contrast, a collocation discretization of the strong equations results in 

\begin{equation}
    \partial_{tt}\underline{\zeta} -R_3D\textrm{diag}(1-\lambda^{-1})D\underline{\zeta} +R_2D^4\underline{\zeta} + \boldsymbol{f} = 0
    \label{collocation discretization of strong equations}
\end{equation}
Note that as $D$ is not symmetric, none of the matrix operators in this this discretization are symmetric. On the other hand, all the matrix operators in \eqref{discretized weak membrane equation} are symmetric, and hence all its eigenvalues are real, and have the same signs as their continuous analogues. Without symmetry, the operators $D^4$ may have complex eigenvalues, which will excite spurious oscillatory modes. 

\section{Well-posedness of the membrane equation with $R_2 >0$}
\label{well-posedness of the membrane equation appendix}
In this section, we prove that when $R_2 > 0$, a weak solution to the 
membrane equation 
\begin{equation}
     R_1\partial_{tt}\zeta  = R_3\partial_\alpha\big( (\lambda-1)\hat{\boldsymbol{s}}\big)- R_2\partial_{\alpha}^4\zeta \\ +\boldsymbol{f}
     \label{bending linearized membrane equation wellposed appendix}
\end{equation}
exists and is unique under some mild assumptions. Importantly, this shows that the membrane equation with positive bending rigidity remains well-posed even if the tension becomes negative. There are many ways to prove such a result. For example, it is possible to show that a subsequence of the finite energy Galerkin solutions converge to the unique weak solution to equation \eqref{bending linearized membrane equation wellposed appendix} by invoking Aubin-Lions \cite{boyer2012mathematical}\cite[Ch.~III, \S2.2, Thm.~2.1]{temam1979navier}. 
For the purpose of this study, we opt for a simpler argument based on Duhamel's principle.
\newline

\subsection{Local well-posedness for positive bending rigidity}

To be precise, we prove that the membrane equation with \(R_2>0\) and fixed-endpoint boundary conditions is locally well-posed as long as the membrane remains an immersion. Let $\zeta_0 \in H^2((-1,1); \mathbb{R}^2)$ and let
$$
    \zeta(\pm1,t)=\zeta_0(\pm1),
    \qquad
    \partial_\alpha^2\zeta(\pm1,t)=0
$$
denote the fixed-endpoint boundary conditions, which we consider here.
\newline

Let 
\begin{equation}
R_1\partial_{tt}\zeta+R_2\partial_\alpha^4\zeta
=
R_3\partial_\alpha G(\partial_\alpha\zeta)+f,
\qquad
G(q)=\left(1-\frac{1}{|q|}\right)q.
\label{positive-R2-evolution}
\end{equation}

For simplicity, we assume the forcing function $f$ depends only on time; unlike the pressure jump forcing term, $f$ is not coupled to the body motion.
Let
$$
    H=L^2((-1,1);\mathbb R^2),
    \qquad
    V=H^2((-1,1);\mathbb R^2)\cap H_0^1((-1,1);\mathbb R^2),
$$

where \(V\) is spanned by the differences of functions satisfying the fixed endpoint conditions. Define
$$
    \mathcal{L}=\frac{R_2}{R_1}\partial_\alpha^4:D(\mathcal{L}) \rightarrow H
$$

with domain
$$
    D(\mathcal{L})=
    \left\{
        v\in H^4((-1,1);\mathbb R^2):
        v(\pm1)=\partial_\alpha^2v(\pm1)=0
    \right\}.
$$

Then \(\mathcal{L}\) is positive and self-adjoint on \(H\), \(D(\mathcal{L}^{1/2})=V\), and the norm
$$
    \|v\|_V=\|\mathcal{L}^{1/2}v\|_H
$$

is equivalent to the \(H^2\) norm on \(V\). Henceforth we suppress the domain and image and write $H^k((-1,1),\mathbb{R}) = H^k$ for all natural numbers $k$.

\begin{theorem}\label{Thm1}
Let $0<\tau_f < \infty$. Suppose that \(R_1,R_2>0\), \(R_3\geq0\), and
$$
    f\in C([0,\tau_f];H).
$$

Let 
$$
    \zeta_0\in H^2,
    \qquad
    \zeta_1\in H,
$$

be functions, and suppose that the initial membrane is an immersion:
$$
    m_0:=
    \min_{\alpha\in[-1,1]}
    |\partial_\alpha\zeta_0(\alpha)|>0.
$$

Then there exists \(\tau_f^*>0\) and a unique weak solution

$$
    \zeta\in C([0,\tau_f^*];H^2)
    \cap C^1([0,\tau_f^*];H)
$$

of \eqref{positive-R2-evolution} with

$$
    \zeta(0)=\zeta_0,
    \qquad
    \partial_t\zeta(0)=\zeta_1.
$$

Moreover,

$$
    \inf_{0\leq t\leq \tau_f^*}
    \min_{\alpha\in[-1,1]}
    |\partial_\alpha\zeta(\alpha,t)|
    \geq\frac{m_0}{2}.
$$

Therefore, a solution exists and is unique on any time interval over which the membrane remains uniformly immersed.

\end{theorem}
\begin{remark}
    In plain language, the condition that the membrane remain uniformly immersed is equivalent to asking that everywhere the membrane strain is bounded below by a constant $> -100\%$, where $-100\%$ corresponds to compression of a line segment to a point. 
\end{remark}
\begin{proof}
We first show that the nonlinear stretching force is locally Lipschitz near an immersed curve. Writing

$$
    G_i(q)=\left(1-\frac{1}{|q|}\right)q_i,
$$

direct differentiation gives

$$
    \frac{\partial G_i}{\partial q_j}(q)
    =
    \left(1-\frac{1}{|q|}\right)\delta_{ij}
    +\frac{q_iq_j}{|q|^3},
$$

and

$$
    \frac{\partial^2G_i}{\partial q_j\partial q_k}(q)
    =
    \frac{
        \delta_{ij}q_k+\delta_{ik}q_j+\delta_{jk}q_i
    }{|q|^3}
    -
    \frac{3q_iq_jq_k}{|q|^5}.
$$

Consequently, for every \(m>0\),

$$
    \left|
        \frac{\partial G_i}{\partial q_j}(q)
    \right|
    \leq C\left(1+\frac1m\right),
    \qquad
    \left|
        \frac{\partial^2G_i}{\partial q_j\partial q_k}(q)
    \right|
    \leq \frac{C}{m^2},
    \qquad |q|\geq m.
$$

Let

$$
    m_0=\min_{\alpha\in[-1,1]}
    |\partial_\alpha\zeta_0(\alpha)|>0.
$$

Since $H^2 \hookrightarrow C^1$, suppose that \(\zeta\) and \(\eta\) lie in an \(H^2\)-neighborhood of \(\zeta_0\) sufficiently small that

$$
    \|\partial_\alpha\zeta-\partial_\alpha\zeta_0\|_{L^\infty},
    \,
    \|\partial_\alpha\eta-\partial_\alpha\zeta_0\|_{L^\infty}
    \leq \frac{m_0}{2}.
$$

For \(0\leq\theta\leq1\), let

$$
    q_\theta
    :=
    \theta\partial_\alpha\zeta
    +(1-\theta)\partial_\alpha\eta
$$

lie on the line segment between $\partial_\alpha \zeta$ and $\partial_\alpha \eta$. Then

$$
\begin{aligned}
    |q_\theta|
    &\geq
    |\partial_\alpha\zeta_0|
    -
    \theta
    |\partial_\alpha\zeta-\partial_\alpha\zeta_0|
    -
    (1-\theta)
    |\partial_\alpha\eta-\partial_\alpha\zeta_0|\\
    &\geq \frac{m_0}{2}.
\end{aligned}
$$

Thus the entire line segment joining
\(\partial_\alpha\zeta\) and \(\partial_\alpha\eta\) remains in the region where the first two derivatives of \(G\) are bounded. By the fundamental theorem of calculus,

$$
\begin{aligned}
    &\frac{\partial G_i}{\partial q_j}
        (\partial_\alpha\zeta)
    -
    \frac{\partial G_i}{\partial q_j}
        (\partial_\alpha\eta)\\
    &\qquad=
    \int_0^1
    \frac{\partial^2G_i}{\partial q_j\partial q_k}
        (q_\theta)
    \left(
        \partial_\alpha\zeta_k
        -
        \partial_\alpha\eta_k
    \right)\,d\theta,
\end{aligned}
$$

and hence

$$
    \left|
    \frac{\partial G_i}{\partial q_j}
        (\partial_\alpha\zeta)
    -
    \frac{\partial G_i}{\partial q_j}
        (\partial_\alpha\eta)
    \right|
    \leq
    \frac{C}{m_0^2}
    |\partial_\alpha\zeta-\partial_\alpha\eta|.
$$

Using

$$
    \partial_\alpha G_i(\partial_\alpha\zeta)
    =
    \frac{\partial G_i}{\partial q_j}
        (\partial_\alpha\zeta)
    \partial_\alpha^2\zeta_j,
$$

we obtain

$$
\begin{aligned}
    &\partial_\alpha G_i(\partial_\alpha\zeta)
    -
    \partial_\alpha G_i(\partial_\alpha\eta)\\
    &\quad=
    \frac{\partial G_i}{\partial q_j}
        (\partial_\alpha\zeta)
    \left(
        \partial_\alpha^2\zeta_j
        -
        \partial_\alpha^2\eta_j
    \right)\\
    &\qquad+
    \left[
        \frac{\partial G_i}{\partial q_j}
            (\partial_\alpha\zeta)
        -
        \frac{\partial G_i}{\partial q_j}
            (\partial_\alpha\eta)
    \right]
    \partial_\alpha^2\eta_j.
\end{aligned}
$$

Therefore,

$$
\begin{aligned}
    &
    \left\|
        \partial_\alpha G(\partial_\alpha\zeta)
        -
        \partial_\alpha G(\partial_\alpha\eta)
    \right\|_{L^2}\\
    &\qquad\leq
    C\left(1+\frac1{m_0}\right)
    \|\partial_\alpha^2\zeta-\partial_\alpha^2\eta\|_{L^2}
    +
    \frac{C}{m_0^2}
    \|\partial_\alpha\zeta-\partial_\alpha\eta\|_{L^\infty}
    \|\partial_\alpha^2\eta\|_{L^2}.
\end{aligned}
$$

The one-dimensional Sobolev embedding \(H^1\hookrightarrow L^\infty\) now gives

\begin{equation}
    \left\|
        \partial_\alpha G(\partial_\alpha\zeta)
        -
        \partial_\alpha G(\partial_\alpha\eta)
    \right\|_{L^2}
    \leq
    C(m_0,M)\|\zeta-\eta\|_{H^2},
    \label{lipschitz estimate}
\end{equation}

provided

$$
    \|\zeta\|_{H^2},\|\eta\|_{H^2}\leq M.
$$

It follows that

$$
    \mathcal N(\zeta)
    :=
    \frac{R_3}{R_1}
    \partial_\alpha G(\partial_\alpha\zeta)
$$

is locally Lipschitz from \(H^2\) into \(L^2\) in a neighborhood of every immersed curve.
\newline

Let

$$
    A=\frac{R_2}{R_1},
$$

and recall that

$$
    \mathcal L=A\partial_\alpha^4
$$

denotes the positive self-adjoint beam operator with domain

$$
    D(\mathcal L)
    =
    \left\{
        v\in H^4((-1,1);\mathbb R^2):
        v(\pm1)=\partial_\alpha^2v(\pm1)=0
    \right\}.
$$

Its energy space is

$$
    V=D(\mathcal L^{1/2})
      =H^2((-1,1);\mathbb R^2)\cap H_0^1((-1,1);\mathbb R^2),
$$

and the norm \(\|\mathcal L^{1/2}v\|_{L^2}\) is equivalent to \(\|v\|_{H^2}\) on \(V\).

Using the functional calculus for the positive self-adjoint operator \(\mathcal L\), define

$$
    C(t)=\cos(t\mathcal L^{1/2}),
    \qquad
    S(t)=\mathcal L^{-1/2}\sin(t\mathcal L^{1/2}).
$$

These are the two fundamental solution operators for the linear beam equation

$$
    \partial_{tt}\zeta+\mathcal L\zeta=0.
$$

One may think of $\cos$ and $\sin$ here as an infinite polynomial series. More precisely, \(C(t)\zeta_0\) is the solution with initial displacement \(\zeta_0\) and zero initial velocity, while \(S(t)v_0\) is the solution with zero initial displacement and initial velocity \(v_0\). Thus the solution of the linear initial-value problem is

$$
    \zeta_{\mathrm{lin}}(t)
    =
    C(t)\zeta_0+S(t)v_0.
$$

It is easy to verify (for example with the Fourier transform) that the propagators satisfy

$$
    \|C(t)v\|_V\leq\|v\|_V,
    \qquad
    \|\partial_tC(t)v\|_{L^2}\leq\|v\|_V,
$$

for \(v\in V\), and

$$
    \|S(t)w\|_V\leq\|w\|_{L^2},
    \qquad
    \|\partial_tS(t)w\|_{L^2}\leq\|w\|_{L^2},
$$

for \(w\in L^2\).
\newline

With

$$
    \mathcal N(\zeta)
    =
    \frac{R_3}{R_1}
    \partial_\alpha G(\partial_\alpha\zeta),
$$

equation \eqref{positive-R2-evolution} becomes

$$
    \partial_{tt}\zeta+\mathcal L\zeta
    =
    \mathcal N(\zeta)+\frac{1}{R_1}f.
$$

By Duhamel's principle, $\zeta$ satisfies the implicit relation
\begin{equation}
\begin{aligned}
\zeta(t)
&=
C(t)\zeta_0+S(t)v_0+
\int_0^t
S(t-\tau)
\left(
\mathcal N(\zeta(\tau))
+\frac{1}{R_1}f(\tau)
\right)\,d\tau\ \\
&=
\zeta_{\mathrm{lin}}(t)
+
\int_0^t
S(t-\tau)
\left(
\mathcal N(\zeta(\tau))
+\frac{1}{R_1}f(\tau)
\right)\,d\tau.
\end{aligned}
\label{positive-R2-duhamel}
\end{equation}

We now close the immersion bootstrap. By the Sobolev embedding theorem, there exists positive constant \(C_E\) such that

$$
    \|\partial_\alpha v\|_{L^\infty}
    \leq C_E\|v\|_{H^2},
    \qquad v\in V.
$$

Choose $\rho$ such that

$$
    0<\rho<\frac{m_0}{2C_E},
    \qquad
    m_0=
    \min_{\alpha\in[-1,1]}
    |\partial_\alpha\zeta_0(\alpha)|.
$$

Consider the closed ball

$$
    \mathcal B_{\tau_f,\rho}
    =
    \left\{
        \zeta\in C([0,\tau_f];H^2):
        \begin{array}{l}
        (i) \ \ \ \zeta\text{ satisfies the fixed endpoint conditions},\\[1mm]
        (ii)\ \ \zeta(0)=\zeta_0,\\[1mm]
        (iii)\ \displaystyle
        \sup_{0\leq t\leq \tau_f}
        \|\zeta(t)-\zeta_0\|_{H^2}\leq\rho
        \end{array}
    \right\}.
$$

For every \(\zeta\in\mathcal B_{\tau_f,\rho}\),

$$
\begin{aligned}
    |\partial_\alpha\zeta(\alpha,t)|
    &\geq
    |\partial_\alpha\zeta_0(\alpha)|
    -
    |\partial_\alpha(\zeta-\zeta_0)(\alpha,t)|\\
    &\geq
    m_0-C_E\|\zeta(t)-\zeta_0\|_{H^2}\\
    &\geq\frac{m_0}{2}.
\end{aligned}
$$

Every curve in \(\mathcal B_{\tau_f,\rho}\) is therefore uniformly immersed, and the local Lipschitz estimate
\eqref{lipschitz estimate} holds uniformly on this set.
\newline

Define the fixed-point map

$$
\begin{aligned}
    (\Phi\zeta)(t)
    &=
    \zeta_{\mathrm{lin}}(t)\\
    &\quad+
    \int_0^t
    S(t-\tau)
    \left(
        \mathcal N(\zeta(\tau))
        +\frac{1}{R_1}f(\tau)
    \right)\,d\tau,
\end{aligned}
$$

and observe that 
\begin{equation}
\begin{aligned}
    \sup_{0\leq t\leq \tau_f}
    \|\Phi\zeta(t)-\zeta_0\|_{H^2}
    \leq{}&
    \sup_{0\leq t\leq \tau_f}
    \|\zeta_{\mathrm{lin}}(t)-\zeta_0\|_{H^2}\\
    &+
    C\int_0^{\tau_f}
    \left(
        \|\mathcal N(\zeta(\tau))\|_{L^2}
        +\frac{1}{R_1}\|f(\tau)\|_{L^2}
    \right)\,d\tau.
    \label{fixed point estimate}
\end{aligned}
\end{equation}

Since

$$
    \zeta_{\mathrm{lin}}(t)\longrightarrow\zeta_0
    \quad\text{in }H^2
    \qquad\text{as }t\to0,
$$

and \(\mathcal N\) is bounded on \(\mathcal B_{\tau_f,\rho}\), by shrinking \(\tau_f\) we may ensure that the right-hand side of equation \eqref{fixed point estimate} is smaller than \(\rho\). Hence \(\Phi\) maps \(\mathcal B_{\tau_f,\rho}\) into itself.
\newline

If \(\zeta,\eta\in\mathcal B_{\tau_f,\rho}\), then

$$
\begin{aligned}
    \sup_{0\leq t\leq \tau_f}
    \|\Phi\zeta(t)-\Phi\eta(t)\|_{H^2}
    &\leq
    C\int_0^{\tau_f}
    \|
        \mathcal N(\zeta(\tau))
        -
        \mathcal N(\eta(\tau))
    \|_{L^2}\,d\tau\\
    &\leq
    \tau_fC(m_0,\rho,\|\zeta_0\|_{H^2})
    \sup_{0\leq t\leq \tau_f}
    \|\zeta(t)-\eta(t)\|_{H^2}.
\end{aligned}
$$

By shrinking \(\tau_f\) again, we may ensure that \(\Phi\) is a contraction. Banach's fixed-point theorem therefore gives a unique solution of \eqref{positive-R2-duhamel}. Because the fixed point lies in \(\mathcal B_{\tau_f,\rho}\),

$$
    |\partial_\alpha\zeta(\alpha,t)|
    \geq\frac{m_0}{2}
$$

throughout its interval of existence. Differentiating the Duhamel formula in time and using

$$
    \partial_tS(t)=C(t)
$$

also gives

$$
    \partial_t\zeta\in C([0,\tau_f];L^2).
$$

Thus we obtain local existence and uniqueness of $\zeta$ on the interval $[0,\tau_f]$. where $\tau_f$ depends on $m_0$, $\|\zeta_0\|_{H^2}$ and $\|\zeta_1\|_{L^2}$ and $\|f\|_{L^2}$ only.
\newline 

It remains to show that this solution may be continued on any interval where $\zeta$ remains an immersion. To see this, suppose that the maximal existence time \(T_{\max}\) is finite but that, for some \(m>0\),

$$
    |\partial_\alpha\zeta(\alpha,t)|\geq m
$$

for every \(\alpha\in[-1,1]\) and \(0\leq t<T_{\max}\). Since

$$
    \partial_\alpha G_i(\partial_\alpha\zeta)
    =
    \frac{\partial G_i}{\partial q_j}
        (\partial_\alpha\zeta)
    \partial_\alpha^2\zeta_j
$$

and

$$
    \left|
        \frac{\partial G_i}{\partial q_j}(q)
    \right|
    \leq C_m
    \qquad\text{whenever }|q|\geq m,
$$

we have

$$
    \|\mathcal N(\zeta)\|_{L^2}
    \leq C_m\|\zeta\|_{H^2}.
$$

The Duhamel estimates therefore imply

$$
\begin{aligned}
    \|\zeta(t)\|_{H^2}
    +\|\partial_t\zeta(t)\|_{L^2}
    \leq{}&
    C\left(
        \|\zeta_0\|_{H^2}
        +\|v_0\|_{L^2}
    \right)\\
    &+
    C_m\int_0^t
    \|\zeta(\tau)\|_{H^2}\,d\tau\\
    &+
    \frac{C}{R_1}
    \int_0^t
    \|f(\tau)\|_{L^2}\,d\tau.
\end{aligned}
$$

Then by invoking Gronwall's inequality, we have 

$$
    \sup_{0\leq t<T_{\max}}
    \left(
        \|\zeta(t)\|_{H^2}
        +\|\partial_t\zeta(t)\|_{L^2}
    \right)<\infty.
$$

Hence by the same argument above, the solution may be continued for some additional time $\tau_f^* > 0$. which depends only on $m$, $\|\zeta(T_{max})\|_{H^2}$,$\|\partial_t\zeta|_{T_{max}}\|_{H^2}$ and $\|f\|_{L^2}$. This contradicts the maximality of \(T_{\max}\). Thus, if \(T_{\max}<\infty\), then

$$
    \liminf_{t\uparrow T_{\max}}
    \min_{\alpha\in[-1,1]}
    |\partial_\alpha\zeta(\alpha,t)|=0.
$$

Hence, the solution persists for as long as the membrane remains uniformly immersed. Note that no assumption was made on the sign of

$$
    T=R_3\bigl(|\partial_\alpha\zeta|-1\bigr)
$$
in the argument.
\end{proof}

\section{The bounded energy constraint as a solution selector}
\label{missing boundary conditions appendix}
In this section, we discuss how the bounded-energy constraint enforced by the weak formulation can be used to supplement a PDE with a formally insufficient number of boundary conditions.
\newline

Consider the following example. Let $y = y(\alpha) \in H^1([0,1])$. Recall that $H^1$ is the space of functions $f$ such that both $f$ and $f'$ are square integrable. Physically, $H^1$ represents the space of elastic profiles with bounded stretching energy. Suppose also that $y$ satisfies the second order differential equation 
\begin{equation}
    \alpha y'' +y' = 0, \ \alpha\in(0,1).
    \label{example equation appendix}
\end{equation}

It is easy to verify by direct computation that the general solution to \eqref{example equation appendix} is 
\begin{equation}
    y(\alpha ) = A + B\ln(\alpha ).
\end{equation}
However, since $\int_0^1|\ln(\alpha )'|^2d\alpha = \int_0^1\frac{1}{\alpha^2}d\alpha = [\frac{-1}{\alpha}]^1_0 = +\infty $, we see that inside the space of functions with bounded energy $H^1$, the only solutions must be constant
\begin{eqnarray}
    y(\alpha ) = A.
\end{eqnarray}
Hence, although $\alpha y'' +y' = 0$ is a second order differential equation, it is well-posed with only a single boundary condition in the space of bounded energy solutions $H^1$.
\newline

On the other hand, consider the equation 
 \begin{equation}
    \alpha y'' +y' + R_2y''''= 0, \ \alpha \in(0,1).
    \label{example equation with R2 appendix}
\end{equation}

For each $R_2 > 0$, by letting $\boldsymbol{x} = [y,y',y'',y''']$ it is easy to write \eqref{example equation with R2 appendix} in the form of 
\begin{eqnarray}
    \boldsymbol{x}' = A(\alpha)\boldsymbol{x},
\end{eqnarray}
where $A(\alpha)$ is some $4\times4$ matrix in $C^1$. Hence by classical ODE theory, equation \eqref{example equation with R2 appendix} admits a unique solution in $C^4$ and (hence with bounded energy) when exactly 4 boundary conditions are given.
\newline

In summary, for every $R_2 > 0$, equation \eqref{example equation with R2 appendix} requires exactly 4 boundary conditions, and all 4 solution branches have finite energy (lie in $H^1$). When $R_2 =0$ however, only one of the two branches have finite energy. Solving the equation requires exactly one boundary condition, and the ``missing" boundary condition is replaced by the constraint that the solution has finite energy.
\newline

Equation \eqref{example equation with R2 appendix} is essentially equivalent to the linearized membrane equation at steady state. More generally, the exact same analysis holds for the more general equation 

\begin{equation}
    R_1y =  R_3(\alpha y')' + R_2y''''
    \label{general example equation with R2 appendix}
\end{equation}
Indeed, by the same argument as before, when $R_2 >0$ it has exactly 4 solution branches all with finite energy, and when $R_2 =0$, 

\begin{equation}
    R_1y =  R_3(\alpha y')'
    \label{general example equation appendix}
\end{equation}
The substitution $\tau = \sqrt{\frac{4R_1}{R_3}\alpha}$ transforms it into the Bessel equation

\begin{equation}
    \tau^2y'' + \tau y' + (\tau^2 -0^2) y =0
\end{equation}
which has fundamental pair of solutions ($J_0, Y_0$), the Bessel functions of the first and second kind. Hence, the general solution is given as the linear combination

\begin{equation}
    y(\alpha) = AJ_0(\sqrt{\frac{4R_1}{R_3}\alpha}) + BY_0(\sqrt{\frac{4R_1}{R_3}\alpha}) 
\end{equation}
Like before, $J_0(\sqrt{\frac{4R_1}{R_3}\alpha}) \in H^1$ but $Y_0(\sqrt{\frac{4R_1}{R_3}\alpha}) \notin H^1$. Indeed, for $\tau \ll 1$ we have $Y_0(\tau)\sim\frac{2}{\pi}\ln(\frac{\tau}{2})$ so its first derivative is not square integrable.
Hence in $H^1$ the general solution to equation \eqref{general example equation appendix} is 

\begin{equation}
    y(\alpha) = AJ_0(\sqrt{\frac{4R_1}{R_3}\alpha}).
\end{equation}
Consequently, equation \eqref{general example equation appendix} is well posed with a single boundary condition in the space of bounded energy functions. Naturally, the same also holds for the inhomogeneous problem 
\begin{equation}
    R_1y =  R_3(\alpha y')' + R_2y'''' + f.
    \label{general example inhomogeneous equation with R2 appendix}
\end{equation}
By separation of variables, the linearized membrane equation
\begin{equation}
    R_1\partial_{tt}y =  R_3(\alpha y')' + R_2y'''' + f,
    \label{unsteady general example inhomogeneous equation with R2 appendix}
\end{equation}
reduces to equations of the form \eqref{general example inhomogeneous equation with R2 appendix}, one for each eigenvalue, and hence only requires $4$ boundary conditions when $R_2 >0$ and exactly $1$ when $R_2 = 0$ in $H^1$ as well \cite{triantafyllou1991paradox}.
\newline

With these examples in mind, it is now intuitively clear why the full nonlinear equation
\begin{equation}
R_1\partial_{tt}\zeta  = \partial_\alpha\big(R_3(1 - |\partial_\alpha \zeta|^{-1})\partial_\alpha\zeta\big)+ \boldsymbol{f}
\end{equation}
with $\zeta = x+ iy$ can be solved in the weak formulation numerically by only specifying 3 boundary conditions (($
x,y$) position of the fixed end and vanishing tension at the free end). The key observation is that the tension $T = R_3(|\partial_\alpha \zeta| - 1)$ vanishes at the free end, and hence the equation degenerates there in the same manner as the linearized problems above. More precisely, suppose $\alpha\in[0,1]$, the fixed end is at $\alpha = 1$ and the free end is at $\alpha = 0$. Near the free end, the tension admits a Taylor expansion $T(\alpha) = T'(0)\alpha + \mathcal{O}(\alpha^2)$ so that to leading order each component of $\zeta$ satisfies an equation of the form \eqref{unsteady general example inhomogeneous equation with R2 appendix} for some constants $R_1,R_2,R_3$.
\newline

As in the model problems, this equation admits two local solution branches per component; one branch has bounded energy and the other has unbounded energy. Since the weak formulation is solved with a spectral Galerkin method whose trial functions are polynomials (and hence lie in $H^1$), the unbounded-energy branches cannot be represented, and the system is well-posed with a single boundary condition per component at the free end. It is possible to make the formal argument outlined above fully rigorous, but the details are rather technical.

\section{Convergence in the limit of zero bending rigidity}
\label{convergence in the limit of zero bending rigidity section}
In this section, we show that solutions of the elastic membrane equation with nonzero linearized bending rigidity converge, in the limit \(R_2\to 0\), to solutions of the corresponding zero-bending-rigidity membrane equation. Let \(\zeta^0\) satisfy
\begin{equation}
R_1\partial_{tt}\zeta
= \partial_\alpha\big((T_0+R_3(\partial_\alpha s-1))\hat{\boldsymbol{s}}\big)
+ \boldsymbol{f},
\label{membrane with force equation}
\end{equation}
and, for \(R_2>0\), let \(\zeta^{R_2}\) satisfy
\begin{equation}
R_1\partial_{tt}\zeta
= \partial_\alpha\big((T_0+R_3(\partial_\alpha s-1))\hat{\boldsymbol{s}}\big)
-R_2\partial_\alpha^4\zeta
+ \boldsymbol{f}.
\label{membrane with linearized bending rigidity with force equation}
\end{equation}
Here we define \(T_0\) to be the tension of the undeformed membrane, following \cite{mavroyiakoumou2020large}; $T_0 = 0$ for all the numerical results in the main body of this paper.
Under the same initial and boundary conditions, together with several mild assumptions specified below, we prove that

$$
    \zeta^{R_2}\to\zeta^0
    \qquad\text{in }H^1
    \quad\text{as }R_2\to0.
$$

For the fluid--structure problem considered in this study, the forcing is

$$
    \boldsymbol{f}
    = i[p]^+_-\,(\partial_\alpha s)^{-1}\partial_\alpha\zeta.
$$

Indeed, for $R_1 = 0.25,0.5,1$ we observe that in the limit of $R_2 \rightarrow 0$ our simulations converge to the simulation with $R_2 =0$ in both position $\zeta$ and velocity $\partial_t\zeta$  with scaling law $\mathcal{O}(R_2)$ at early times. Figure \ref{l2 convergence figure fixed-free} shows this convergence in the case $R_1 = 1, R_3 = 100$. Panels (a-b) in the top row show log-log plots of the $L^2$ errors in position and velocity respectively against $R_2$ at $t = 1$. The bottom row shows the same quantities at $t = 2$. At $t=1$ the convergence rate is approximately linear ($\mathcal{O}(R_2)$). By $t=2$ however, the convergence rate for $\partial_t\zeta$ deteriorates to approximately $\mathcal{O}(R_2^{1/2})$. We note that convergence only occurs above extremely small values of $R_2 \lessapprox 10^{-10}$. When $R_2 \gtrapprox  10^{-10}$, the error decreases only slightly with decreasing $R_2$.
\newline

For $R_1$ = 0.25 and 0.5, the convergence follows a similar pattern. It is linear in $R_2$ at early times, then deteriorates to between $\mathcal{O}(R_2^{1/2})$ and $\mathcal{O}(R_2)$ at moderate times. Over longer times, e.g. $t \gtrsim 3$ for the case in figure \ref{l2 convergence figure fixed-free}, convergence with respect to $R_2$ is difficult to observe in our simulations due to the chaotic evolution of the flapping membrane flag and vorticity distribution, as discussed for point vortices by \cite{aref1983integrable}. 

\begin{figure}[H]
    \centering
    \includegraphics[width=1\linewidth]{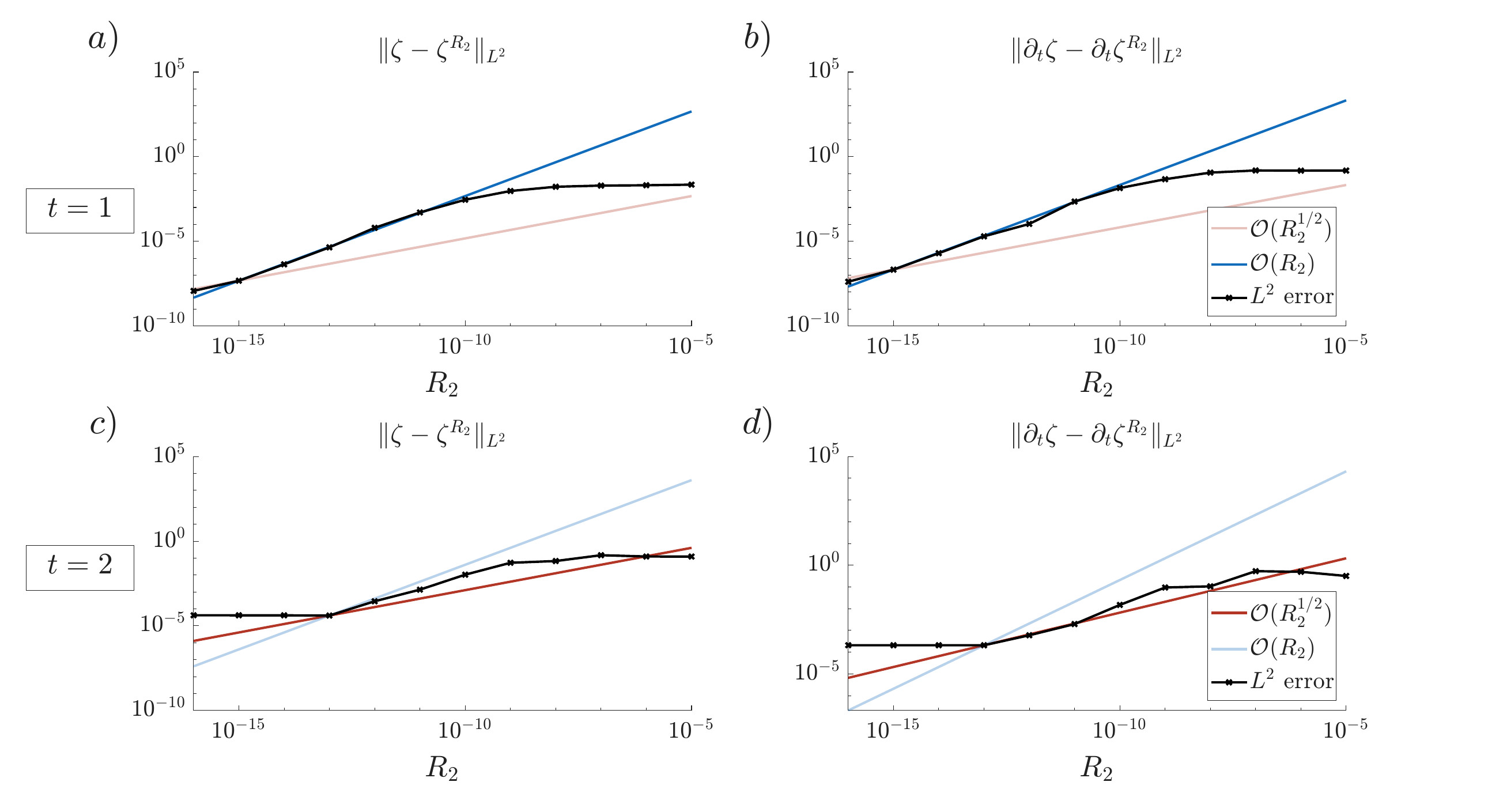}
    \caption{Log-log error plots for a fixed-free membrane with $R_1 = 1, R_3 = 100$ in a uniform background fluid flow. (a-b) $L^2$ error in position and velocity against $R_2$ at $t = 1$. (c-d) $L^2$ error in position and velocity against $R_2$ at $t = 2$.}
    \label{l2 convergence figure fixed-free}
\end{figure}

The numerical convergence rates observed in figure \ref{l2 convergence figure fixed-free} are consistent with the following theorem.

\begin{theorem}
    Let $\epsilon = R_2 >0$,  Suppose that $\zeta^0 = \zeta \in C^1_tH^2_\alpha$
    solves the membrane equation with Dirichlet boundary conditions:
    \begin{equation}
        R_1\partial_{tt}\zeta  = \partial_\alpha\big((T_0 +R_3(|\partial_\alpha \zeta| -1 ))\hat{\boldsymbol{s}}\big)+ \boldsymbol{f}.
        \label{membrane without epsilon force equation}
    \end{equation} where $\hat{\boldsymbol{s}} =\frac{\partial_\alpha\zeta}{|\partial_\alpha\zeta|}$ denotes the unit tangent. Suppose that for every $\epsilon > 0$ sufficiently small we have a family of solutions $\zeta^\epsilon \in C^1_tH^2_\alpha$ to 
    \begin{equation}
         R_1\partial_{tt}\zeta^\epsilon  = \partial_\alpha\big((T_0 +R_3(|\partial_\alpha \zeta^\epsilon| -1 ))\hat{\boldsymbol{s}}^\epsilon\big)- \epsilon\partial_\alpha^4\zeta^\epsilon + \boldsymbol{f}^\epsilon,
         \label{membrane with epsilon force equation}
    \end{equation}     
    that satisfy the same initial and Dirichlet boundary conditions as $\zeta$. Let $\tau$ denote the maximal time of existence of $\zeta \in C^1_tH^2_\alpha$, and let $\tau_f$ be such that $0 \leq t\leq\tau_f < \tau \leq \infty$ and the following assumptions hold:
    \newline
    
    (i) There are constants $c_0, C_0 >0$ such that for all $\epsilon \geq 0$ we have 
    \begin{equation}
      c_0 \leq |\partial_\alpha\zeta^\epsilon| \leq C_0
     \label{tension nonnegative assumption}
    \end{equation}
    (The elastic bodies never undergo infinite compression or expansion in the limit of zero bending rigidity.) 
    \newline
    
    (ii) $1-\frac{T_0}{R_3}< c_0$. This ensures $T(
    |\partial_\alpha\zeta^\epsilon|) := T_0 +R_3(|\partial_\alpha\zeta^\epsilon|-1) > 0$ so the bodies are under tension everywhere.
    \newline
    
    (iii) There exists constant $M>0$ such that for all $t\in [0,\tau_f]$ and $\epsilon \geq 0$ we have $\|\zeta^{\epsilon}\|_{H^2}\leq M$. This ensures the stored bending energy does not blow up in the limit of $\epsilon \rightarrow 0$.
    \newline

    (iv) $r_\epsilon(t) := \|\boldsymbol{f}^\epsilon - \boldsymbol{f}\|_{L^2_\alpha} \rightarrow 0 $ in $L^2_t$. as $\epsilon \rightarrow 0$. The external forcing is continuous at $\epsilon = 0$ with respect to the norm $\|\cdot\|_{L^2_t L^2_\alpha }=\| \|\cdot\|_{L^2_\alpha}\|_{L^2_t}$.
    \newline
    
    Then we have $\lim_{\epsilon \rightarrow0}\|\zeta^\epsilon- \zeta\|_{H^1} = \lim_{\epsilon \rightarrow0}\|\partial_t\zeta^\epsilon- \partial_t\zeta\|_{L^2} = 0$ uniformly on $t\in[0,\tau_f]$. In addition, there exists constant $C_{\tau_f}$ such that
    \begin{equation}
    \|\zeta - \zeta^\epsilon\|_{H^1} + \|\partial_t\zeta - \partial_t\zeta^\epsilon\|_{L^2} \leq C_{\tau_f}(\sqrt{\epsilon} +\|r_\epsilon\|_{L^2_t}).
\end{equation}

Moreover, if the membrane solution has additional regularity, $\zeta \in C^1_tH^2_\alpha \cap C^0_tH^4_\alpha$, and additional boundary agreement, $\zeta$ and $\zeta^\epsilon$ share the same tangents at endpoints $\partial_\alpha \zeta(t,\pm1) = \partial_\alpha \zeta^\epsilon(t,\pm1)$, the rate of convergence improves to
\begin{equation}
    \|\zeta - \zeta^\epsilon\|_{H^1} + \|\partial_t\zeta - \partial_t\zeta^\epsilon\|_{L^2} \leq C_{\tau_f}(\epsilon +\|r_\epsilon\|_{L^2_t}),
\end{equation}
    \label{convergence theorem in L2}
\end{theorem}
\begin{remark}
   It is possible to show that (iii) follows from (ii). That's intuitively clear. Indeed, if the tension is uniformly bounded below, the membrane is always uniformly taut, so the curvature is uniformly bounded above.
\end{remark}
The proof of theorem \ref{convergence theorem in L2} is somewhat technical. However, the main ideas here are well known \cite{kato1966perturbation}. Before we begin the proof, we fix the notation and explain some important facts that we will use:
\newline

Fix $T_0,R_1,R_3$. Define the stretching energy density as 

\begin{equation}
    W(\xi) = T_0|\xi| + \frac{1}{2}R_3(|\xi|-1)^2.
\end{equation}
Then the associated stretching energy is $\mathcal{U}_{stretch} = \int_{-1}^1W(\partial_\alpha\zeta)d\alpha$. Also, \begin{equation}
    (T_0 +R_3(|\partial_\alpha \zeta^\epsilon| -1 ))\hat{\boldsymbol{s}}^\epsilon = DW(\zeta^\epsilon).
\end{equation}Correspondingly, define the relative bending energy density 

\begin{equation}
    W(\xi_1;\xi_2) = W(\xi_1) - W(\xi_2) - DW(\xi_2)(\xi_1 -\xi_2).
\end{equation}

The following result holds:
\begin{lemma}
    Suppose that \begin{equation}
    \max\{1-\frac{T_0}{R_3},0\}< c_0 \leq |\xi_1|,|\xi_2| \leq C_0
    \end{equation}
    
    There exist $c_1$ and $\delta: 0<\delta<c_0$ such that 
    $$ 0 <c_1 \leq T(c_0 -\delta).$$
    Moreover, we may find constants $\tilde{c}_0, \tilde{C}_0$ such that for all $|\xi_1 - \xi_2|\leq\delta$ we have 
    \begin{equation}
        \tilde{c}_0|\xi_1 - \xi_2|^2 \leq W(\xi_1;\xi_2) \leq  \tilde{C}_0|\xi_1 - \xi_2|^2.
        \label{stretching energy asymptotics}
    \end{equation}
    \label{hessian lemma}
\end{lemma}
\begin{proof}
    Since $\max\{1-\frac{T_0}{R_3},0\} < c_0$, we may find $0 <\delta < c_0$ and $0 <c_1$ such that $1-\frac{T_0}{R_3} +\frac{c_1}{R_3} +\delta \leq c_0$. By rearranging we have $$0 < c_1 \leq T(c_0-\delta).$$ Consider the line segment $\ell(z) = \xi_2 + z(\xi_1 -\xi_2),\ 0\leq z \leq 1$. Then $|\ell(z)| \geq |\xi_1| - |\xi_1 - \xi_2| \geq c_0 - \delta$ since $c_0 \leq |\xi_1|,|\xi_2|$. Hence the line segment $\ell(z)$ lies on the annulus $c_0-\delta \leq \ell(z) \leq C_0$ where the tension $T(\ell(z))\geq T(c_0-\delta)\geq c_1$ is bounded below uniformly by $c_1$.
    \newline
    
    By Taylor expansion with Lagrange remainder, $$W(\xi_1;\xi_2) = W(\xi_1) - (W(\xi_2) + DW(\xi_2)(\xi_1 -\xi_2)) = \frac{1}{2}(\xi_1 - \xi_2)^TD^2W(\ell(z_0))(\xi_1 - \xi_2)$$ for some $0 < z_0 < 1$. Thus, we need only show that for some $\tilde{c}_0,\tilde{C}_0$ we have 
    \begin{equation}
        \tilde{c}_0 I \leq \frac{1}{2}D^2W \leq \tilde{C}_0 I
    \end{equation}
    on the annulus $0 <c_0 -\delta\leq \xi \leq C_0$. This is immediate from a direct computation in polar coordinates. Indeed, $W(\xi)$ is a radial function. It is easy to check that in orthonormal polar coordinates the Hessian $D^2W(\xi)$ is given by 
    \begin{equation}
        D^2W(\xi) = \begin{bmatrix} 
        R_3 & 0 \\
        0 & T(\xi)/|\xi|
        \end{bmatrix}
    \end{equation}
    Hence we may take $\tilde{C}_0 = \max(R_3, T(C_0)/c_0)$ and $\tilde{c}_0 = \min(R_3,c_1/C_0)$.
\end{proof}

We also need the following version of the Gagliardo-Nirenberg interpolation inequality:
\begin{lemma}
    Let $f \in H^1_\alpha[-1,1]$. Then
    \begin{equation}
        \|f\|_{L^\infty} \leq C\|f\|_{L^2}^{1/2}\|f\|_{H^1}^{1/2}
    \end{equation}
    \label{gagliardo lemma}
\end{lemma}
\begin{proof}
    By the fundamental theorem of calculus,
    \begin{equation}
        f^2(\alpha) - f^2(b) = \frac{1}{2}\int_b^\alpha f'(\beta)f(\beta)d\beta \leq B\|f\|_{L^2}\|f\|_{H^1}
    \end{equation}
    for some $B$. Now choose $b$ such that $f^2(b) = \frac{1}{2}\|f\|_{L^2}^2$, which can always be done by the mean value theorem. Then by relabeling $B$,
    \begin{align}
        f^2(\alpha) &\leq B(\|f\|_{L^2}\|f\|_{H^1} +  \|f\|_{L^2}^2) \\ 
        &\leq B\|f\|_{L^2}\|f\|_{H^1} 
    \end{align}
    Hence for some $B$, $\|f\|_{L^\infty} \leq B\|f\|_{L^2}^{1/2}\|f\|_{H^1}^{1/2}$
    
\end{proof}

We may now prove theorem \ref{convergence theorem in L2}.

\begin{proof}
      Let $w(\alpha, t;\epsilon) =\zeta(\alpha, t) - \zeta^\epsilon(\alpha, t)$ be the error function. By assumption $w(0,\alpha;\epsilon) = 0$ since $\zeta$ and $\zeta^\epsilon$ have the same initial conditions. Hence since $w(\alpha, t;\epsilon) \in C^1_tH^2_\alpha$ the following is nonzero: 
      \begin{equation}
          \tau_\epsilon :=\sup\{t < \tau_f :\|\partial_\alpha\zeta(t,\cdot) -\partial_\alpha\zeta^\epsilon(r,\cdot)\|_{L^\infty} \leq\delta\} >0
      \end{equation}
      We prove the result first on $[0,\tau_\epsilon]$. Later we bootstrap the result to show for sufficiently small $\epsilon$ we have $\tau_\epsilon = \tau_f$. Here $\delta >0$ is chosen as in lemma \ref{hessian lemma}. 
      Define the \textit{relative energy} between $\zeta,\zeta^\epsilon$ to be 
    \begin{equation}
       E(t) = \frac{R_1}{2}\|\partial_tw\|^2_{L^2}
            +\int_{-1}^1W(\partial_\alpha\zeta^\epsilon;\partial_\alpha\zeta)d\alpha
            + \frac{\epsilon}{2}\|\partial_{\alpha\alpha}\zeta^\epsilon\|^2_{L^2}.
    \end{equation} 
    
    By lemma \ref{hessian lemma} we know that for all $t \leq \tau_\epsilon$ we have
    \begin{equation}
         \tilde{c}_0|\partial_\alpha\zeta -\partial_\alpha\zeta^\epsilon|^2\leq W(\partial_\alpha\zeta^\epsilon;\partial_\alpha\zeta) \leq \tilde{C}_0|\partial_\alpha\zeta -\partial_\alpha\zeta^\epsilon|^2
    \end{equation} which also guarantees that our relative energy is always non-negative $E(t)\geq 0$. 
    \newline
    
    Notice that 
    \begin{multline}
        \frac{d}{dt}E(t) = R_1\int_{-1}^1(\partial_tw)(\partial_t^2w)d\alpha  \\ + \int_{-1}^1\bigg((DW(\partial_\alpha\zeta^\epsilon) -DW(\partial_\alpha\zeta))\partial_{\alpha t}\zeta^\epsilon - \partial_{\alpha t}\zeta \cdot D^2W(\partial_\alpha\zeta)w\bigg)\, d\alpha  + \epsilon\int_{-1}^1\partial_{\alpha\alpha}\zeta^\epsilon\partial_{\alpha\alpha t}\zeta^\epsilon  d\alpha
        \label{time derivative energy step 1 equation}
    \end{multline}

    Furthermore, by subtracting equations \eqref{membrane without epsilon force equation} and \eqref{membrane with epsilon force equation} we may obtain 
    \begin{equation}
         R_1\partial_{tt}w= \partial_\alpha DW(\partial_\alpha\zeta) - \partial_\alpha DW(\partial_\alpha\zeta^\epsilon) + \epsilon\partial_{\alpha\alpha\alpha\alpha} \zeta^\epsilon+(\boldsymbol{f} -\boldsymbol{f}^\epsilon)
         \label{strong error equation}
    \end{equation}
    Hence after multiplying \eqref{strong error equation} by $\partial_tw$ and integrating by parts we have
     \begin{equation}
         R_1\int_{-1}^1(\partial_tw)(\partial_t^2w)d\alpha = \int_{-1}^1\big(DW(\partial_\alpha\zeta^\epsilon) -DW(\partial_\alpha\zeta)\big)\partial_{\alpha t}w\, d\alpha +\epsilon\int_{-1}^1\partial_{\alpha\alpha}\zeta^{\epsilon}\partial_{\alpha\alpha t}wd\alpha + \int_{-1}^1(\boldsymbol{f} -\boldsymbol{f}^\epsilon)\partial_tw\, d\alpha.
     \end{equation}

     Note that the boundary terms incurred by integrating by parts vanish exactly because $\zeta$ and $\zeta^\epsilon$ share the same Dirichlet boundary conditions so that $\partial_tw(t,\pm1) = 0$. Substituting this identity into equation \eqref{time derivative energy step 1 equation} yields

     \begin{align}
        \frac{d}{dt}E(t) &= \int_{-1}^1\bigg((DW(\partial_\alpha\zeta^\epsilon) -DW(\partial_\alpha\zeta))\partial_{\alpha t}\zeta - \partial_{\alpha t}\zeta \cdot D^2W(\partial_\alpha\zeta)w\bigg)\, d\alpha \nonumber \\ & \ \ \ \ \ \ \ \ \ \ \ \ \ \ \ + \int_{-1}^1(\boldsymbol{f} -\boldsymbol{f}^\epsilon)\partial_tw\, d\alpha
         + \epsilon\int_{-1}^1\partial_{\alpha \alpha}\zeta^\epsilon\partial_{\alpha\alpha t}\zeta d\alpha \\
        &=: J_1 + J_2 + J_3 
        \label{time derivative energy step 2 equation}
    \end{align}

    By the choice of $\delta$, by Taylor expansion we know that $$DW(\partial_\alpha\zeta^\epsilon) - DW(\partial_\alpha\zeta) - D^2W(\partial_\alpha\zeta^\epsilon-\partial_\alpha\zeta) = \mathcal{O}(|\partial_\alpha\zeta^\epsilon-\partial_\alpha\zeta|^2) = \mathcal{O}(W(\partial_\alpha\zeta^\epsilon;\partial_\alpha\zeta))$$
    where the last equality follows again from lemma \ref{hessian lemma}.
    
    Hence $$J_1 \leq  B_0\int_{-1}^1W(\partial_\alpha\zeta^\epsilon;\partial_\alpha\zeta)d\alpha \leq B_1E(t)$$
    for some $B_0,B_1 > 0$. Similarly, 
    $$J_2 \leq \|\boldsymbol{f}-\boldsymbol{f}^\epsilon\|_{L^2}\|w_t\|_{L^2} \leq \frac{1}{2}\|\boldsymbol{f}-\boldsymbol{f}^\epsilon\|^2_{L^2} + \frac{1}{2}\|w_t\|_{L^2}^2 
 \leq r_\epsilon(t)^2 + E(t)$$
     Finally, 
    $$J_3= \epsilon\int_{-1}^1\partial_{\alpha \alpha}\zeta^\epsilon\partial_{\alpha\alpha t}\zeta d\alpha \leq \epsilon\|\zeta^\epsilon\|_{H^2}\|\partial_{t}\zeta\|_{H^2} = \mathcal{O}(\epsilon)$$ 

    since $\zeta \in C^1_tH^2_\alpha$. Hence we may find constants $B_1,B_2$ such that on the closed interval $[0,\tau_\epsilon]$ we have
    \begin{equation}
        \frac{d}{d}E(t) \leq B_1 E(t) + r_\epsilon(t)^2 + B_2\epsilon
    \end{equation}

    This is exactly the set up in Gr\"onwall's inequality. By using integrating factors we see that on the interval $t \in[0,\tau_\epsilon]$
    \begin{equation}
        \|\partial_t\zeta - \partial_t\zeta^\epsilon \|^2_{L^2} + \|\partial_\alpha\zeta - \partial_\alpha\zeta^\epsilon \|_{L^2}^2\leq A_1E(t) \leq A_2e^{B_1t}(\epsilon t + \int_{-1}^tr_\epsilon^2(t')\,dt') 
    \end{equation}
    For some constants $A_1$, $A_2$. Hence we may find constant $C_{\tau_f}$ such that 
    \begin{equation}
        \|\partial_t\zeta - \partial_t\zeta^\epsilon \|^2_{L^2} + \|\partial_\alpha\zeta - \partial_\alpha\zeta^\epsilon \|_{L^2}^2\leq C_{\tau_f}(\epsilon + \|r_\epsilon\|^2_{L^2_t}) \rightarrow 0 \text{ as $\epsilon \rightarrow 0$.}
        \label{dt da inequality equation}
    \end{equation}
    \newline
    
    Now we bootstrap the inequality. Let $t = \tau_\epsilon$. Towards a contradiction assume $\tau_\epsilon <\tau_f$. By lemma \ref{gagliardo lemma}, for sufficiently small $\epsilon$ we have 
    \begin{align}
        \delta &\leq \|\partial_\alpha\zeta(\tau_\epsilon,\cdot) - \partial_\alpha\zeta^\epsilon(\tau_\epsilon,\cdot)\|_{L^\infty} \\ 
        &\leq B\|\partial_\alpha\zeta(\tau_\epsilon,\cdot) - \partial_\alpha\zeta^\epsilon(\tau_\epsilon,\cdot)\|_{L^2}^{1/2}\|\partial_\alpha\zeta(\tau_\epsilon,\cdot) - \partial_\alpha\zeta^\epsilon(\tau_\epsilon,\cdot)\|_{H^1}^{1/2} \\ &\leq 2BMC_{\tau_f}^{1/2}(\epsilon + \|r_\epsilon\|^2_{L^2_t}) ^{1/4} \\ 
        &\leq \delta/2.
    \end{align}
The first inequality follows from the definition of $\tau_\epsilon$ and the second inequality follows from lemma \ref{gagliardo lemma}. The third inequality follows from equation \eqref{dt da inequality equation} and assumption (iii) while the last inequality follows from assumption (iv) and sufficiently small $\epsilon$. This is a contradiction, hence $\tau_\epsilon =\tau_f$ for sufficiently small $\epsilon$. 
\newline

Thus, for any $\tau_f <\tau \leq \infty$ there exists constant $C_{\tau_f}$ such that the inequality 
\begin{equation}
        \|\partial_t\zeta - \partial_t\zeta^\epsilon \|^2_{L^2} + \|\partial_\alpha\zeta - \partial_\alpha\zeta^\epsilon \|^2_{L^2}\leq C_{\tau_f} (\epsilon t + \|r_\epsilon\|^2_{L^2_t})
    \end{equation}
holds for all $t\in[0,\tau_f]$.
\newline

It remains to upgrade this convergence to $\|\zeta - \zeta^\epsilon\|_{L^2}\rightarrow 0$. To this end, define 
$$I(t) = \frac{1}{2}\int_{-1}^1w^2d\alpha =\frac{1}{2}\|\zeta - \zeta^\epsilon\|_{L^2}^2 $$
Then \begin{align}
    \frac{d}{dt}I(t) &= \int_{-1}^1(\zeta -\zeta^\epsilon)(\partial_t\zeta - \partial_t\zeta^\epsilon)d\alpha \\
    &\leq \|\zeta -\zeta^\epsilon\|_{L^2}\|\partial_t\zeta -\partial_t\zeta^\epsilon\|_{L^2} \\
    &\leq \frac{1}{2}\|\zeta -\zeta^\epsilon\|^2_{L^2}\| + \frac{1}{2}\|\partial_t\zeta -\partial_t\zeta^\epsilon\|^2_{L^2} \\
    &= \frac{1}{2}I(t) +\frac{1}{2}C_{\tau_f} (\epsilon t + \|r_\epsilon\|^2_{L^2_t})
\end{align}

So by using integrating factors again, we obtain yet another Gronwall inequality

\begin{align}
    I(t) &\leq \frac{e^{1/2t} C_{\tau_f}}{2}(\epsilon \frac{t^2}{2} + t \|r_\epsilon\|^2_{L^2_t}) \\ 
    &\leq \tilde{C}_{\tau_f}(\epsilon  + \|r_\epsilon\|^2_{L^2_t})
\end{align}
for some $\tilde{C}_{\tau_f}$. Hence by re-labeling constants, we may find $C_{\tau_f}$ such that
\begin{equation}
    \|\zeta - \zeta^\epsilon\|_{H^1} + \|\partial_t\zeta - \partial_t\zeta^\epsilon\|_{L^2} \leq C_{\tau_f}(\sqrt{\epsilon} +\|r_\epsilon\|_{L^2_t}) \rightarrow 0 \text{ as $\epsilon \rightarrow 0$,}
\end{equation}

which concludes the main argument.
\newline

We now sketch how to improve the rate of convergence in $\epsilon$ from $\sqrt{\epsilon}$ to $\epsilon$. To do so, we need to modify the relative energy. Define the \textit{modified relative energy} to be 
\begin{equation}
           \tilde{E}(t) = \frac{R_1}{2}\|\partial_tw\|^2_{L^2}
            +\int_{-1}^1W(\partial_\alpha\zeta^\epsilon;\partial_\alpha\zeta)d\alpha
            + \frac{\epsilon}{2}\|\partial_{\alpha\alpha}w\|^2_{L^2}.
\end{equation}

The only difference is that the relative bending energy  $\frac{\epsilon}{2}\|\partial_{\alpha\alpha}\zeta^\epsilon\|^2_{L^2}$ is replaced with $\frac{\epsilon}{2}\|\partial_{\alpha\alpha}w\|^2_{L^2}.$ The same argument goes through unchanged, except that in this case 

\begin{align}
    J_3 &= \frac{d}{dt}\frac{\epsilon}{2}\|\partial_{\alpha\alpha}w\|^2_{L^2}  + \epsilon\int_{-1}^1
\partial_{\alpha\alpha\alpha\alpha}\zeta^\epsilon\partial_{t}w\,d\alpha \\
&=\epsilon\int_{-1}^1 \partial_{\alpha\alpha\alpha\alpha}\zeta\partial_{t}w \,d\alpha \\
&\leq \epsilon\|\zeta\|_{H^4}\|w_t\|_{L^2} \\ 
&\leq\epsilon^2 \frac{1}{2}\|\zeta\|^2_{H^4} + \frac{1}{2}\|w_t\|^2_{L^2} \\
&\leq B_3(\epsilon^2 + \tilde{E}(t))
\end{align}
for some $B_3$. The boundary terms from integration by parts vanish exactly because of the additional assumption $\partial_\alpha \zeta(t,\pm1) = \partial_\alpha \zeta^\epsilon(t,\pm1)$. Hence after relabeling constants we obtain 
\begin{equation}
    \frac{d}{d}\tilde{E}(t) \leq B_1 \tilde{E}(t) + r_\epsilon(t)^2 + B_2\epsilon^2
\end{equation}

 The rest of the proof goes through exactly the same as before but with $\epsilon^2$ instead of $\epsilon$, improving the convergence rate in $\epsilon$ to $\mathcal{O}(\epsilon)$ from $\mathcal{O}(\epsilon^{1/2})$.
 \end{proof}

In summary, we showed that under some mild assumptions we can expect solutions to the elastic equation with small $R_2$ to converge to solutions to the membrane equation with rate $\mathcal{O}(R_2^{1/2} + \|r_\epsilon\|_{L^2_t})$. In addition, given some additional regularity and boundary assumptions, this convergence can be improved to $\mathcal{O}(R_2+ \|r_\epsilon\|_{L^2_t})$. This linear convergence rate was observed numerically at early times in figure \ref{l2 convergence figure fixed-free}. We note that an important condition for convergence is that the force $\boldsymbol{f}^\epsilon$ is continuous in $\epsilon = R_2$. As with Theorem \ref{Thm1}, the force is not that due to the pressure jump but models its behavior with respect to $R_2$. From a purely physical point of view, that the force induced by the inviscid fluid $\boldsymbol{f}^\epsilon = [p^\epsilon]^+_-\partial_\alpha s^\epsilon\,\hat{\boldsymbol{n}}^\epsilon$ is continuous in $\epsilon$ is a reasonable assumption, and the numerical results support it. 
\newline

Note that the theorem as stated does not apply to the fixed-free problem. Indeed, at the free endpoint, the tension must vanish, and assumption (iii) in theorem \ref{convergence theorem in L2} fails. Nonetheless, the numerical results in figure \ref{l2 convergence figure fixed-free} suggests that a similar result holds. Indeed, the numerical results suggest that the string is uniformly tensioned away from the free end. By using smooth cutoff functions, in the same argument above it is possible to prove a similar convergence result in the fixed-free case on any closed material interval $\alpha \in [-1,a], a<1$ where the bodies are uniformly in tension $T(\partial_\alpha\zeta^\epsilon) \geq c_1 >0$ on $[-1,a]$ for all $\epsilon \geq 0$ at $t\in[0,\tau_f]$. The estimates however become more technical.
\newline

$R_2$ is sometimes used as a regularization parameter in fluid-membrane simulations \cite{connell2007flapping,zhang2026flapping}.
To approximate the membrane regime of $R_2 = 0$ with nonzero $R_2$, $R_2$ has to be extremely small. Indeed, generically, the rate of convergence is extremely slow at $\mathcal{O}(R_2^{1/2})$. Moreover, the constant increases exponentially in the final time. Any small perturbation caused by shrinking $R_2$ is ultimately magnified exponentially in time. This is unsurprising, since for small $R_2$ the system is chaotic in general. As figure
\ref{l2 convergence figure fixed-free} suggests, practically speaking, convergence only begins when $R_2$ decreases below $10^{-10}$. 
In figures \ref{convergence snapshots figure} and \ref{convergence close snapshots figure} we show snapshots of a flapping elastic body in a uniform flow (left to right) from $t = 2.5$ to $t = 5$ and $t =2.5$ to $t =3$ respectively. In both figures we consider $0\leq R_2 \leq10^{-5}$. When $R_2 \leq 10^{-11}$ the snapshots are visually indistinguishable. When $R_2 \geq 10^{-8}$ many of the high frequency membrane perturbations have been smoothed out. Most notably, the tip is much less erratic and doesn't curl up nearly as tightly as can be seen in figure \ref{convergence close snapshots figure}. By $R_2 = 10^{-5}$, the snapshots are significantly different from those with $R_2 \leq 10^{-11}$. In figure \ref{convergence close snapshots figure}, the body does not reverse direction near the tip. This slow rate of convergence was also observed in \cite{manela2017hanging} where they consider a simplified version of the elastic equation and vortex shedding model. However, qualitatively, when $R_2\leq 10^{-10}$ the resulting dynamics are visually identical.

\begin{figure}[H]
    \centering
    \includegraphics[width=1\linewidth, trim = 0cm 2cm 0cm 0cm, clip]{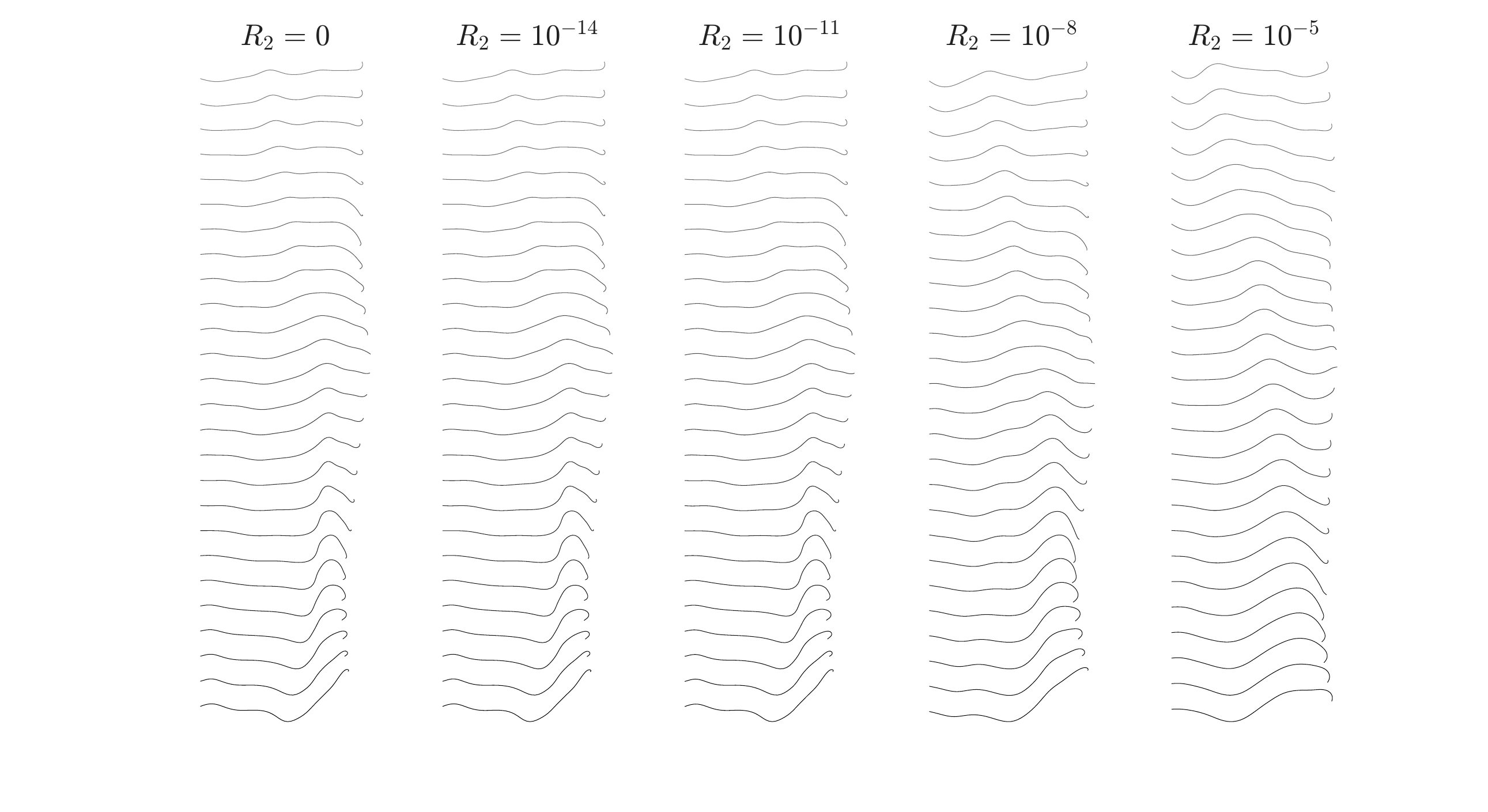}
    \caption{Snapshots of a fixed-free elastic body with $R_2 = 0,\ 10^{-14},\ 10^{-11},\ 10^{-8},\ 10^{-5}$ from $t = 2.5$ to $t = 5$. In all cases $R_1=1$ and $R_3 = 100$. The color changes from light gray to black as $t$ increases.}
    \label{convergence snapshots figure}
\end{figure}

\begin{figure}[H]
    \centering 
    \includegraphics[width=1\linewidth, trim = 0cm 3cm 0cm 0cm, clip]{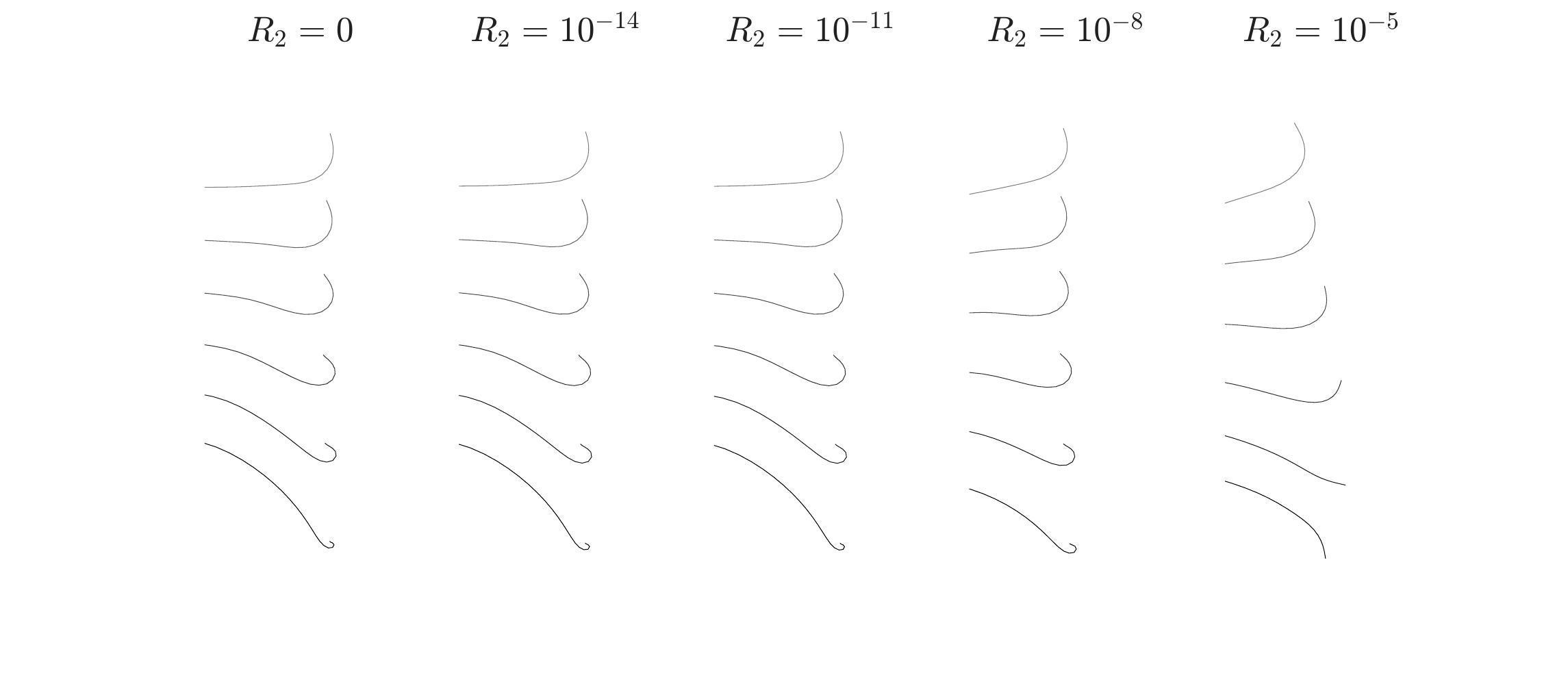}
    \caption{Close-up snapshots of the elastic body in figure \ref{convergence snapshots figure} from $t = 2.5$ to $t = 3$. The color changes from light gray to black as $t$ increases.}
    \label{convergence close snapshots figure}
\end{figure}

\section{The pressure jump equation as a projection of a conservation law}
\label{pressure jump derivation appendix}

Recall the Euler equation given by
\begin{equation}
    \partial_t\boldsymbol{u} + \boldsymbol{u}\cdot\boldsymbol{\nabla}\boldsymbol{u} + \boldsymbol{\nabla}p = 0
\end{equation}
Consider the velocity  $\boldsymbol{u}(\zeta,t)$ of a fluid parcel in an inviscid fluid that follows a material point of the membrane with position $\zeta(\alpha, t)$. Let $\boldsymbol{v} = \partial_t\zeta$. Then from the Euler equation,
\begin{align}
    \frac{d}{dt}(\boldsymbol{u}(\zeta,t)) &= \partial_t\boldsymbol{u} +\boldsymbol{v}\cdot\boldsymbol{\nabla}\boldsymbol{u} \\&= -\boldsymbol{u}\cdot\boldsymbol{\nabla}\boldsymbol{u} + \boldsymbol{v}\cdot\boldsymbol{\nabla}\boldsymbol{u} -\boldsymbol{\nabla}p \\
    &= -(\boldsymbol{u}- \boldsymbol{v})\cdot\boldsymbol{\nabla}\boldsymbol{u} - \boldsymbol{\nabla}p \\
    &=-(\boldsymbol{u}- \boldsymbol{v})\cdot\boldsymbol{\nabla}(\boldsymbol{u}-\boldsymbol v) - \boldsymbol{\nabla}p\\
    &= -\boldsymbol{\nabla}\cdot\big(\boldsymbol{w}\otimes\boldsymbol{w} + pI\big) 
\end{align}
Here, $\boldsymbol{w} =\boldsymbol{u} - \boldsymbol{v}$.
Hence we may write the Euler equation for the velocity of a fluid parcel following a point on the membrane in the form of a conservation law with a time-dependent flux:

\begin{equation}
    \frac{d}{dt}(\boldsymbol{u}(\zeta,t))  + \boldsymbol{\nabla}\cdot\big(\boldsymbol{w}\otimes\boldsymbol{w} + pI\big) = \boldsymbol{0}
    \label{conservation law form of euler }
\end{equation}
This is the starting point of our derivation of the pressure jump equation. As we will now show, projecting this conservation law onto the local tangent vector of the membrane $\hat{\boldsymbol{s}}$ and considering the jump across the membrane yields the pressure jump equation in conservation law form. 
Let $\boldsymbol{w} = w_s\hat{\boldsymbol{s}} + w_n\hat{\boldsymbol{n}}$ and likewise let $\boldsymbol{u} = u_s\hat{\boldsymbol{s}} + u_n\hat{\boldsymbol{n}}$. For convenience, let the superscript $(\cdot)^\pm$ denote the limit of $(\cdot)$ on the $\pm$ side of the membrane.
Then 
\begin{align}  \hat{\boldsymbol{s}}\cdot\bigg(\frac{d}{dt}(\boldsymbol{u}(\zeta,t))  + \boldsymbol{\nabla}\cdot\big(\boldsymbol{w}\otimes\boldsymbol{w} + pI\big)\bigg)^\mp = \frac{d}{dt}u_s^\mp + \hat{\boldsymbol{s}}\cdot\bigg(\boldsymbol{\nabla}\cdot\big(\boldsymbol{w}\otimes\boldsymbol{w} + pI\big)\bigg)^\mp
\end{align}
To obtain the first term on the right-hand side of the previous equation, we used the fact that $d\hat{\boldsymbol{s}}/dt$ is proportional to $\hat{\boldsymbol{n}}$.
Observe that with respect to the basis $\{\hat{\boldsymbol{s}},\hat{\boldsymbol{n}}\}$,
\begin{equation}
    \boldsymbol{w}\otimes\boldsymbol{w} + pI = \begin{bmatrix}
        w_s^2 +p & w_sw_n \\
        w_nw_s & w_n^2 + p
    \end{bmatrix}
\end{equation}
Hence 
\begin{align}
    \hat{\boldsymbol{s}}\cdot\boldsymbol{\nabla}\cdot(\boldsymbol{w}\otimes\boldsymbol{w} + pI) &= 
    \hat{\boldsymbol{s}}\cdot\begin{bmatrix}
        \partial_s(w_s^2 +p) + \partial_n(w_sw_n) \\
        \partial_s(w_nw_s) + \partial_n((w_n^2 + p)
    \end{bmatrix} \\ 
    &=\partial_s(w_s^2 +p) + \partial_n(w_sw_n)\\
    &=\partial_s(w_s^2 +p) + \partial_nw_sw_n + w_s\partial_nw_n \\ 
    &=\partial_s(w_s^2 +p) + \partial_nw_sw_n - w_s\partial_sw_s \\
    &= \partial_s(\frac{1}{2}w_s^2 +p) + \partial_nw_sw_n
\end{align}
where the second to last equality follows from the incompressibility of both $\boldsymbol{u}$ and hence $\boldsymbol{w}$, since $\boldsymbol{v}$ is constant in space. By the no-penetration condition, the normal velocity of the fluid and membrane match on the surface of the membrane, hence $(w_n)^\mp =0$.

Thus 
\begin{equation}
\bigg(\hat{\boldsymbol{s}}\cdot\boldsymbol{\nabla}\cdot(\boldsymbol{w}\otimes\boldsymbol{w} + pI)\bigg)^\mp = \partial_s(\frac{1}{2}(w_s^2)^\mp +p^\mp)
\end{equation}

Furthermore, by taking difference of
\begin{equation}
    \frac{d}{dt}u_s^\mp  + \partial_s(\frac{1}{2}(w_s^2)^\mp +p^\mp) = 0.
\end{equation}
across the $\mp$ sides of the membrane and noting that $(w_s^2)^- - (w_s^2)^+ = (w_s^- + w_s^+ )(w_s^- - w_s^+)$ we obtain
\begin{equation}
    \frac{d}{dt}[u_s]^-_+  + \partial_s(\frac{w_s^+ +w_s^-}{2}[w_s]^-_+ +[p]^-_+) = 0.
\end{equation}
Note that $\mu -\tau = \frac{w_s^+ +w_s^-}{2}$ is the average relative velocity across the membrane. Since $\boldsymbol{v}$ is continuous across the membrane, $[u_s]^-_+ =[w_s]^-_+ = \gamma$ is the vortex sheet strength. Hence this equation is exactly:
\begin{equation}
    \partial_t\gamma  + \partial_s((\mu -\tau)\gamma+[p]^-_+) = 0.
\end{equation}
which is the pressure jump equation in conservation law form for the evolution of $\gamma$.
\newline

In summary, the relative velocity of a fluid parcel following a material point on the membrane admits the structure of a conservation law with a time-dependent flux. Projecting this equation onto the local tangent vector and taking its difference yields a similar conservation law for the vortex sheet strength, since it is equal to the tangential jump in fluid momentum across the membrane. Just as fluid momentum travels from regions of high pressure to low pressure, vortex sheet strength is advected by the average slip velocity of the flow relative to the body and with the arc length derivative of the pressure jump as a source term.

\printbibliography[
heading=bibintoc,
title={References}
]
\end{document}